\documentclass[acmsmall,nonacm]{acmart}

\usepackage{booktabs}       %
\usepackage{tikz-cd}        %

\newcommand{\grain}[1]{G[#1]}              %
\newcommand{\gop}[1]{\mathcal{G}(#1)}      %
\newcommand{\ek}[1]{EK[#1]}                %
\newcommand{\BC}[1]{BC[#1]}                %

\newcommand{\eqg}{\equiv_g}                %
\newcommand{\leg}{\leq_g}                  %
\newcommand{\stg}{<_g}                     %
\newcommand{\incg}{\langle\rangle_g}       %

\newcommand{\gl}[1]{\varphi(#1)}              %

\newcommand{\subt}{\subseteq_{typ}}        %
\newcommand{\psub}{\subset_{typ}}          %
\newcommand{\tun}{\cup_{typ}}              %
\newcommand{\tin}{\cap_{typ}}              %
\newcommand{\tdiff}{-_{typ}}               %
\newcommand{\thra}{\twoheadrightarrow}     %
\newcommand{\ssub}{\sqsubseteq}            %
\newcommand{\spsub}{\sqsubset}             %

\theoremstyle{plain}
\newtheorem{theorem}{Theorem}[section]
\newtheorem{lemma}[theorem]{Lemma}
\newtheorem{corollary}[theorem]{Corollary}
\newtheorem{proposition}[theorem]{Proposition}

\theoremstyle{definition}
\newtheorem{definition}[theorem]{Definition}

\theoremstyle{remark}
\newtheorem{remark}[theorem]{Remark}

\title{Grain Theory: Type-Level Granularity Correctness in Data Pipelines}

\author{Nikos Karayannidis}
\affiliation{%
  \institution{Independent}
  \city{Athens}
  \country{Greece}
}
\email{nkarag@gmail.com}



\begin{document}

\begin{abstract}
Data pipelines fail silently when a transformation changes the \emph{grain} of
the data, the level of detail each element represents, without the designer
intending it; fan traps (metrics duplication) and chasm traps (data loss) are
the familiar symptoms. Dimensional modeling has organized data warehouses around
grain since the 1990s, and declaring it is the first step a designer takes, yet
it has remained prose: no definition, no rules for how transformations change
it, nothing beyond a flat fact table's attributes.

We define the grain of a type by irreducibility and isomorphism alone, with no
appeal to functional dependencies, and the definition settles what three decades
of practice left open. A grain is a property of the \emph{type} and not of any
instance; it is unique up to isomorphism; on a product type it \emph{is} a
candidate key, and a type with several grains is a relation with several
candidate keys. Defining grain without functional dependencies is what makes it
total: a compositional operator computes it over products, sums and the
inductive and coinductive fixed points. As a result, identifying information
is defined for arrays and event streams, and for every nested and recursive type
manipulated by the nested relational calculus and monad comprehensions, types
on which classical dependency theory has no referent.
The candidate-key correspondence is therefore the \emph{base case of an
extension} rather than a standalone observation. Grain also carries what a key
cannot. One structure under different declared grains denotes different things, so
grain determines an \emph{entity key} and a \emph{behavioral class} fixing how a
collection must be read and written, and the triple
$(\grain{R},\,\ek{R},\,\BC{R})$ is computable from the type definition
and its declarations, with no data access. A grain can
moreover be irreducible, isomorphic and still \emph{surrogate} rather than the
one its domain intends.

The grain ordering is a partial order on types, axiomatized by
Armstrong-style axioms that are sound in general and, on product types,
complete by transfer from Armstrong's; on product types the three relations
also form a bounded lattice. Every data
transformation~$h$ has a \emph{grain lift} $\gl{h}$, the grain-level content
of~$h$. For the relational algebra we prove an inference rule per
operation, chief among them a general equi-join grain theorem whose
irreducibility-and-labeling analysis refines classical key propagation: the rule
computes the \emph{exact, minimal} grain of the result, not a superkey of it.
The \emph{grain homomorphism} binds the two levels: grain projection
commutes with transformation and grain lifts compose, so grain-level reasoning
is faithful to the type level. The algorithm \emph{CalcG} therefore decides
whether a pipeline's output grain matches its target at design time, before any
code is written or any query is executed, from the schema alone, in time linear
in the pipeline. A property that could previously only be tested is decided.

The errors then follow from a single comparison. A grain the calculus computes,
set against a grain the design intends, stands in exactly one relation to it in
the grain ordering, and that position names the failure; fan traps, dropped
grain fields, ungrounded multi-version reads, wrong-grain aggregations and
behavioral-class violations are the instances of this classification. Chasm
traps, a data-instance failure the theory does \emph{not} model, are localized
rather than decided. Grain determination
coincides with functional dependency on product types, but the two solve
different problems: dependencies remove redundancy during \emph{schema design},
while grain theory takes the resulting fixed types as given and verifies the
\emph{level-of-detail correctness} of the transformations built on them. The theory is mechanized in Lean~4, with a
concrete model validating every axiom it can express.
\end{abstract}

\maketitle

\section{Introduction}
\label{sec:intro}

Consider a data engineer building a drill-across report that combines
revenue from one fact table with units sold from another, aggregated per
customer and date. The two fact tables have different \emph{grains}, Kimball's term for a
fact table's level of detail~\cite{kimball1996data,kimball2013dw}:
\texttt{SalesChannel} records one row per customer, channel, and date,
while \texttt{SalesProduct} records one row per customer, product, and date.
The natural implementation joins on the common dimensions and aggregates:

\begin{verbatim}
SELECT sc.customer_id, sc.date,
       SUM(sc.revenue)   AS total_revenue,
       SUM(sp.units_sold) AS total_units
FROM sales_channel sc
  JOIN sales_product sp
    ON sc.customer_id = sp.customer_id
   AND sc.date = sp.date
GROUP BY sc.customer_id, sc.date
\end{verbatim}

\noindent
This query compiles, runs, and returns the expected result grain
(customer $\times$ date). Yet it \emph{systematically inflates both
metrics}: the join creates a cross product over the unmatched grain
components, channels on one side and products on the other, duplicating
every row before aggregation. This is a \emph{fan trap}, the intermediate
grain (customer $\times$ channel $\times$ product $\times$ date) being
strictly finer than either input's. It is invisible to schema validation,
type checking and unit tests on small data, because with one channel and one
product per customer-date the cross product is $1 \times 1$ and the numbers
look right. The root cause is a \emph{grain inconsistency}, an unintended
change in the level of detail at which data is represented. Drill-across
reports over fact tables of different grain are a staple of dimensional
modeling, so this is a recurring source of silent metric inflation, not a
contrived example.

This is the normal condition of data engineering, not an unlucky case.
Pipelines are written by going straight to code, whether SQL, dataframe
transformations or, increasingly, code an AI assistant generates at scale,
and between the requirement and the implementation there is no design
artifact in which a correctness argument could be stated, let alone checked.
What stands in for one is testing, which carries a promise it cannot keep
here: a test exhibits behaviour on the data it is given, so it can show a
defect present but never absent, and this defect is invisible to the samples
engineers test against. Nor is it a coding slip a careful reviewer would
catch. The code is a faithful rendering of the intent; what is wrong is the
\emph{level of detail} at which two collections were combined, a property of
the transformation's structure that no assertion over a small input reveals,
the duplication factor at small scale being one. Production data does not
settle it either, since recognising that an amount is inflated takes domain
knowledge of what the figure should be, and reconciliation against a source
system reports only that a total is wrong, months later and many
transformations downstream, without locating the join that caused it.

This query is the smallest instance of a problem that grows with scale.
A production data pipeline is not a single query but a \emph{directed
acyclic graph (DAG) of hundreds or thousands of interdependent
transformations}: joins, aggregations, projections, lookups, and merges
layered from raw ingestion through to published reports. The level of
detail can shift silently at \emph{any} node, and the fan trap is only
one way it happens. Projecting away a column that formed part of the
grain leaves duplicate rows that inflate every count downstream. Reading
a table that keeps every version of an order, without pinning it to a
point in time, collapses those versions at the wrong level. A line
dropped from an order's latest extract but still present in an earlier
one is then silently resurrected, as though it were never deleted
(Section~\ref{sec:errors}; a real, anonymized production bug is dissected
in Appendix~\ref{app:case-study}). Grouping one level finer than intended
double-counts. Every one of these compiles, runs, and passes the unit
tests written against the small, well-behaved samples used in
development.

At this scale neither a reviewer nor a code-generating AI can reliably
track by inspection how the grain changes across every operation; one
misread level of detail is enough, and it does not announce itself. Such
errors surface only in production, on full data, and often far downstream
of the transformation that caused them, which makes them both hard to
detect and expensive to unwind. By the time a figure is found to be
wrong, the inflated numbers may already sit in official reports
circulated to managers across the organization, sometimes discovered
only months later. Tracing a bad total back through a thousand-node
graph to the one join that fanned out is then a costly forensic exercise. The
reliable remedy is to catch such errors \emph{before a single row of
production data is read}: to verify, from the pipeline's structure alone,
that no transformation alters the level of detail except where the
designer intends it to.

Such errors are widespread because no existing formalism treats grain as a
composable property of arbitrary algebraic data types. Kimball's dimensional
modeling~\cite{kimball1996data,kimball2013dw} introduced grain as the level of
detail of a fact table, but left it \emph{informal} (prose, not mathematics),
\emph{narrow} (fact tables only) and \emph{non-systematic} (no transformation
rules). Summarizability theory and the multidimensional models formalize
granularity in specific settings, OLAP dimension hierarchies above all
(Section~\ref{subsec:rw-dimensional}), but not as a property of an arbitrary
data type with rules for how transformations change it.
Functional dependencies and Klug's propagation
rules~\cite{armstrong1974dependency,klug1980calculating} capture
attribute-level determinacy and propagate it through relational
expressions, but reason about schema decomposition, not pipeline
correctness, are unknown to most data engineers, and are confined to
\emph{flat records}. Modern query languages, the nested relational
calculus~\cite{buneman1995principles} and monad
comprehensions~\cite{grust2004monad}, manipulate nested collections and
queries as values, for which classical dependency theory offers no
notion of identifying information; extending a schema-level identity to
\emph{every} algebraic data type is precisely what reasoning about such
programs requires. \emph{How can we formally define grain, infer how
every relational operation changes it, and verify transformation
correctness across a pipeline DAG of any size before any data is
processed, entirely at the type level?}

We introduce \emph{grain theory}, a type-theoretic framework that answers
this question at the type level. It develops in four steps: the definition
of grain, its structure, its calculus, and its decision procedure. Each of
the four is a result, rather than scaffolding for one.

\smallskip\noindent\textbf{The definition
(Section~\ref{sec:foundations}).}
We define the \emph{grain} of a type by irreducibility and isomorphism
alone, so that no functional dependency appears in the definition. The
theory is therefore self-contained at the type level and does not rest on
classical dependency theory; the correspondence with candidate keys on
product types (Remark~\ref{rem:grain-vs-pk}) is a consequence of the
definition rather than an ingredient of it. A grain is a property of the
\emph{type} and not of any instance, and it is unique up to isomorphism
(Theorem~\ref{thm:multiple-grains}). Nor is a grain schema-derivable. It is
a \emph{semantic declaration}, a lens fixing what each element of a
collection represents, arrived at by domain analysis rather than read off
the columns. One structure under different declared grains denotes different
things (Section~\ref{subsec:grain-semantics}), and a grain can be
irreducible, isomorphic and still \emph{surrogate} rather than the one its
domain intends, which separates the grains a domain admits from those it
merely permits (Definition~\ref{def:natural-grain}).

\smallskip\noindent\textbf{Its structure
(Sections~\ref{sec:relations}--\ref{sec:adt}).}
The three grain relations, equality~($\eqg$), ordering~($\leg$) and
incomparability~($\incg$), form a bounded lattice on product types, where
the ordering also has a sound and complete axiomatization conservatively
extending Armstrong's (Section~\ref{sec:relations}). Two further
constructions follow from grain alone. The \emph{entity key} is a derived
grain ($\ek{R} = \grain{E}$, with $\ek{R} \subt \grain{R}$), the identity
that two types share across levels of detail and the point at which
collections of different grain integrate. The \emph{behavioral class},
chosen from seven core classes, fixes the read and write patterns a
collection admits, which unifies table types across data modeling
paradigms (Section~\ref{sec:entity}). The triple
$(\grain{R},\,\ek{R},\,\BC{R})$ is computable from the type definition
and its declarations, with no data access. Grain finally extends compositionally to every algebraic data type,
commuting with each type constructor: products, sums, and the inductive and
coinductive fixed points, with windowing recovering a finite, verifiable
grain for infinite streams (Section~\ref{sec:adt}). Identifying information
is thereby defined for the nested and recursive types that the nested
relational calculus~\cite{buneman1995principles} and monad
comprehensions~\cite{grust2004monad} manipulate, and on which classical
dependency theory has no referent.

\smallskip\noindent\textbf{Its calculus
(Sections~\ref{sec:inference} and~\ref{sec:grain-factorization}).}
Every data transformation $h : R_1 \to R_2$ factors through $\grain{R_2}$,
and its \emph{grain lift} $\gl{h} : \grain{R_1} \to \grain{R_2}$ names the
factor (Section~\ref{sec:grain-factorization}). Which function that factor
is depends on the transformation algebra, so the inference machinery is
tied to a type discipline, and we develop it for the setting data
engineering actually uses: collections of product types under the
relational algebra. There the rules of Section~\ref{sec:inference} give an
explicit formula for $\gl{h}$ for every operation, with a general
\emph{equi-join grain inference theorem} as the headline result. Its
irreducibility-and-labeling analysis refines classical key
propagation~\cite{klug1980calculating}, computing the \emph{exact,
minimal} grain of a join rather than a superkey of it
(Theorem~\ref{thm:equijoin-grain}). What licenses composing these formulas
along a pipeline is the \emph{grain homomorphism}: grain projection
commutes with transformation
($\textit{grain} \circ h = \gl{h} \circ \textit{grain}$) and grain lifts
compose ($\gl{h_2 \circ h_1} = \gl{h_2} \circ \gl{h_1}$). It holds by
construction, in a one-line proof, and that is exactly its role. It
certifies that grain-level reasoning is faithful to the type level and that
it composes; without it, chaining per-operation results down a pipeline
would be symbol-pushing.

\smallskip\noindent\textbf{Its decision procedure
(Section~\ref{sec:calcg}).}
The algorithm \emph{CalcG} traverses a pipeline DAG, computes the grain at
every vertex from its operation's inference rule, and compares the grain
that reaches the sink against the grain the design declares there. It
therefore decides whether a pipeline's output grain matches its target,
from the schema alone, with no data access, and in time linear in the
pipeline (Theorem~\ref{thm:calcg-correctness}). Relational algebra is only
one instantiation, because CalcG needs nothing beyond a computable,
faithful per-operation grain rule over a finite DAG; any transformation
language equipped with such rules inherits design-time verification
unchanged.

\smallskip\noindent\textbf{What this buys.}
A correctness property that could previously only be tested is thereby
decided, at design time, before any code is written or any query is
executed. The same result organizes the failures. Grain-related bugs, from
incorrect aggregations to non-upsert writes against entity tables, share a
root in undefined or imprecise grain, and one comparison exposes them all:
a grain the calculus computes, set against a grain the design intends, its
position in the grain ordering naming the failure
(Section~\ref{sec:errors}). Because read and write patterns are
grain-determined (Section~\ref{sec:entity}), even a point-in-time query on
an event collection or a non-upsert write to an entity table becomes a type
error caught at compile time. The theory is explicit about its limit as
well: chasm traps are a data-instance failure, localized to the grain
ordering chains structurally vulnerable to them rather than decided.

\smallskip\noindent
\textbf{Grain theory is not functional-dependency theory in new
notation.} The two share a relation on product types and solve different
problems at different stages of the data lifecycle. Functional dependencies
are a tool of \emph{schema design}, exposing attribute-level redundancy so a
database can be normalized, an analysis complete before any pipeline exists.
By the time an engineer builds transformations the schemas, and so their
types, are \emph{fixed}. Grain theory begins there, asking what dependencies
never do: does each transformation preserve, or intentionally change, the
\emph{level of detail} at which data is represented? On product types the
grain ordering $R_1 \leg R_2$ and the dependency $R_1 \to R_2$ are the same
relation, read not as ``$R_1$ determines $R_2$'', a redundancy fact settled at
design time, but as ``$R_1$ is finer-grained than $R_2$'', a correctness fact
about transformations. The determination is the shared substrate; what grain
theory adds is the level-of-detail semantics, an ordering that holds
\emph{between} schemas (e.g., a pipeline's sources and target) rather than
inside one, a calculus for how operations change grain, and an extension to
sums and (co)inductive types where dependencies have no referent.

\smallskip\noindent\textbf{Contributions.}
\begin{itemize}
\item \textbf{A formal definition of grain.} Dimensional modeling has organized
  data warehouses around grain since the 1990s, and declaring it is the first
  step a designer takes~\cite{kimball1996data,kimball2013dw}, yet it has
  remained prose: no definition, no rules for how transformations change it, no
  account beyond a fact table's attributes. We define it by irreducibility and
  isomorphism alone (Definition~\ref{def:grain}), with no appeal to functional
  dependencies: a property of the \emph{type} and not of any instance, unique
  up to isomorphism (Theorem~\ref{thm:multiple-grains}), a candidate key on
  product types and more than one beyond them, a declaration under which one
  type denotes different things (Section~\ref{subsec:grain-semantics}),
  \emph{natural} or merely permitted by its domain
  (Definition~\ref{def:natural-grain}), and carrying an \emph{entity key} at
  which collections of different grain integrate (Section~\ref{sec:entity}).
  Earlier formalizations of granularity are confined to OLAP dimension
  hierarchies and multidimensional schemas
  (Section~\ref{subsec:rw-dimensional}); grain here is a property of an
  arbitrary data type.
\item \textbf{Grain for every algebraic data type}
  (Section~\ref{sec:adt}): a compositional operator computing grain over
  products, sums and the inductive and coinductive fixed points, proved correct
  in Theorem~\ref{thm:grain-adt}. Identifying information is thereby defined for
  the nested and recursive types manipulated by the nested relational
  calculus~\cite{buneman1995principles} and monad
  comprehensions~\cite{grust2004monad}, on which classical dependency theory has
  no referent. The correspondence with candidate keys is therefore the
  \emph{base case of an extension}, not a standalone observation.
\item \textbf{Grain ordering as a partial order on types}
  (Section~\ref{sec:relations}). Antisymmetry up to isomorphism
  (Theorem~\ref{thm:grain-partial-order}) is what a preorder of functional
  dependencies does not give, and it is what makes the error classification
  below exhaustive rather than merely suggestive. The ordering is
  \emph{axiomatized} by A1--A9 in Section~\ref{subsec:armstrong}, sound
  there and, on product types, complete by transfer from Armstrong's system
  (Proposition~\ref{prop:grain-completeness}); on product types the three
  grain relations form a bounded lattice.
\item \textbf{Grain inference for the relational algebra}
  (Section~\ref{sec:inference}), an explicit rule per operation, with a general
  equi-join grain theorem (Theorem~\ref{thm:equijoin-grain}) at its centre. An
  irreducibility-and-labeling analysis refines classical key
  propagation~\cite{klug1980calculating}: the rule computes the \emph{exact,
  minimal} grain of the result, not a superkey of it.
\item \textbf{Decidability of pipeline grain-correctness at design time, before
  any code is written or any query is executed}
  (Sections~\ref{sec:grain-factorization}--\ref{sec:calcg}). Grain projection
  commutes with transformation and grain lifts compose, so grain-level reasoning
  is faithful to the type level; the \emph{CalcG} algorithm therefore decides
  whether a pipeline's output grain matches its target, at design time, from the
  schema alone, in time linear in the pipeline
  (Theorem~\ref{thm:calcg-correctness}). A property that could previously only be
  tested is decided.
\item \textbf{A classification of transformation errors}
  (Section~\ref{sec:errors}): one comparison, of a computed grain against a
  reference grain, whose four positions in the grain ordering name the
  failures; fan traps, dropped grain fields, ungrounded multi-version reads,
  wrong-grain aggregations and behavioral-class violations are its instances.
  Chasm traps are a data-instance failure the theory does not model,
  and are localized rather than decided.
\end{itemize}

\noindent
The theory is mechanized in Lean~4~\cite{demoura2021lean4}, and
Section~\ref{sec:verification} reports the mechanization, its scope, and an
empirical validation of the equi-join inference theorem against
PostgreSQL. A companion paper takes up what follows from the decidability
result: it develops the dataflow instantiation, proves the operation-set
universality of CalcG, and builds a design-first methodology on it, reading
the triple $(\grain{R},\,\ek{R},\,\BC{R})$ as a \emph{denotation of data}
and $\gl{h}$ as the denotation of a transformation, with the homomorphism
between them licensing design in the semantic
domain~\cite{pddpaper}.

\smallskip\noindent\textbf{Roadmap.}
Section~\ref{sec:prelim} fixes notation and the type-level set operations.
Grain and its fundamental theorems come in Section~\ref{sec:foundations},
the grain relations and their lattice in Section~\ref{sec:relations}, the
entity key and behavioral classes in Section~\ref{sec:entity}, and the
extension to arbitrary algebraic data types in Section~\ref{sec:adt}. The
calculus follows: the relational grain-inference rules and the equi-join
theorem (Section~\ref{sec:inference}), the grain homomorphism
(Section~\ref{sec:grain-factorization}), the CalcG algorithm
(Section~\ref{sec:calcg}), and the transformation errors it detects
(Section~\ref{sec:errors}). Section~\ref{sec:verification} reports the Lean
and empirical validation, Section~\ref{sec:related} surveys related work,
and Section~\ref{sec:conclusion} concludes with future directions. A reader
interested primarily in the practical payoff may go from the equi-join
theorem (Section~\ref{sec:inference}) directly to the error catalog
(Section~\ref{sec:errors}); the appendices collect the proofs, a
type-theory primer, and a full production case study, and may be skipped on
a first reading.
 
\section{Preliminaries}
\label{sec:prelim}

We work in a type-theoretic setting. Data types are built from base types by
the grammar of \emph{algebraic data types}
\[
  T \;::=\; B \;\mid\; X \;\mid\; T \times T \;\mid\; T + T
           \;\mid\; \mu X.\,T \;\mid\; \nu X.\,T,
\]
where $B$ ranges over base types and $X$ over type variables: $\times$ forms
\emph{product} (record) types, $+$ forms \emph{sum} (variant) types, and
$\mu, \nu$ form inductive and coinductive types. Products and sums appear from
the outset; the recursive constructors $\mu, \nu$ are developed in
Section~\ref{sec:adt}. One instance of each, drawn from the retail domain
this paper uses throughout, fixes what the four constructors build. A
\emph{product} is a record, $\textit{Customer} = \textit{CustomerId} \times
\textit{CustomerName} \times \textit{Email}$, an inhabitant carrying one
value in every slot. A \emph{sum} is a tagged variant,
$\textit{Payment} = \textit{Card} + \textit{Transfer}$, an inhabitant being
a card payment or a bank transfer, never both, and carrying the tag that
says which. A $\mu$-type is finite recursive data, as in the line items of
an order,
\[
  \textit{List}\,\textit{OrderLine}
    \;=\; \mu X.\,\mathbf{1} + \textit{OrderLine} \times X,
\]
read as ``a list of order lines is either empty ($\mathbf{1}$) or one line
followed by another such list'', the bound variable~$X$ standing for that
tail. A $\nu$-type is the same equation read the other way, admitting
\emph{infinite} inhabitants: an unbounded feed of order events is
$\textit{Stream}\,\textit{OrderEvent} = \nu X.\,\textit{OrderEvent} \times
X$, an event followed by a stream, with no empty case to stop at. The \emph{components} of a product type
$R = R_1 \times R_2 \times \cdots \times R_n$ are its field types
$R_1, \dots, R_n$ (their ordering does not matter); more generally a
\emph{field} of $R$ is a surjection $p : R \thra R_i$ onto a component.
For a record, then, components are its fields, which is the case a
relational reader needs. Any other algebraic data type is read the same
way, a component being the type sitting in each slot: the element type of
a list or a stream, or the matching slot in each case of a variant. Types
are therefore compared slot by slot throughout, never by the values they
hold, and Section~\ref{sec:adt} carries this to the recursive
constructors. We
follow functional notation for application, $f\,a\,b$ rather than $f(a,b)$,
and assume semantically rich, domain-driven types
(e.g.\ $\textit{Customer} = \textit{CustomerId} \times
\textit{CustomerName} \times \textit{Email}$, instead of
$\textit{Int} \times \textit{String} \times \textit{String}$) rather than
uninterpreted base types.

\smallskip\noindent\textbf{Inhabitants are constrained records, not the free
product.}\quad
A type denotes a set; its \emph{elements} (or \emph{inhabitants}) are its
values, and we write \emph{row} or \emph{record} when appealing to relational
intuition. The inhabitants of a record type are, however, \emph{not} the full
cartesian product of its components' carriers. They are the records permitted
by the \emph{determinations declared on the type}, a determination being a
declared surjection (a functional dependency or key constraint, exactly as
recorded in a database catalogue) that fixes which combinations of component
values are valid. A function $f : R_1 \to R_2$ is an \emph{isomorphism}, denoted
$f : R_1 \stackrel{\cong}{\longrightarrow} R_2$, if there exists an inverse
$g : R_2 \to R_1$ with $f \circ g = id_{R_2}$ and $g \circ f = id_{R_1}$,
and we write $R_1 \cong R_2$ when such an isomorphism exists. Thus if
$\textit{CustomerId}$ is declared to determine a customer then the
inhabitants of $\textit{Customer}$ are in bijection with the valid
identifiers, and $\textit{Customer} \cong \textit{CustomerId}$ even though
$\textit{Customer}$ is written as a three-way product. Declarations are
\emph{schema constraints, not data}: they carry the same status as keys and
functional dependencies in a relational catalogue, and treating them as part
of the type is what keeps the entire development, grain, inference, and
verification, at the type level.

\smallskip\noindent\emph{Type-theoretic background (for the relational
reader).}\quad
A \emph{base type} is atomic, not built from others
($\textit{CustomerId}$, $\textit{Integer}$, $\textit{Date}$); the
\emph{carrier} of a type is its underlying set of values, forgetting how
each element is structured. Products and sums, records and
variants, and the type-level operations below are all the relational
development requires; the recursive constructors $\mu$ and $\nu$ (finite
and possibly-infinite recursive data, such as lists and streams) enter
only in Section~\ref{sec:adt}, where every result specializes to
products. Appendix~\ref{app:primer} gives a fuller primer.

Product types can be viewed as sets of field types, enabling type-level set
operations.

\begin{definition}[Type-Level Subset]
\label{def:type-subset}
Given types $A$ and $B$, we write $A \subt B$ if there exists a
surjective function $p : B \thra A$. The witness $p$ is of one of two
kinds. It is \emph{structural} when every component of $A$ appears
among the components of $B$ and $p$ is the canonical projection; this
case is decidable from the schema alone. Otherwise $p$ is
\emph{declared}, a surjection asserted by domain knowledge, such as
a foreign key, and $A$ need not be a sub-product of $B$.
\end{definition}

\begin{definition}[Type-Level Proper Subset]
\label{def:type-proper-subset}
We write $A \psub B$ if $A \subt B$ and $\neg(B \subt A)$.
\end{definition}

\begin{definition}[Structural subtype]
\label{def:structural-subtype}
We write $A \ssub B$, and say $A$ is a \emph{structural subtype} of
$B$, if every component of $A$ appears among the components of $B$;
the witness is the canonical projection $B \thra A$ forgetting the
extra components. We write $A \spsub B$ when moreover $B$ has a
component not in $A$. This relation is decidable from the schema
alone and \emph{refines} $\subt$: $A \ssub B$ implies $A \subt B$
(a projection is a surjection), but not conversely.
\end{definition}

\smallskip\noindent\textbf{Structure maps: which functions count.}\quad
Grain compares types by functions between them, so it matters which
functions are admitted as witnesses, and the answer is: only those a
schema can write down. A \emph{structure map} is built by composition
from a fixed list of primitives: keep some columns (projection); put
values side by side (pairing); pick one case of a variant, or branch on
which case is present (the injections of a sum, and case analysis on its
tag); wrap or unwrap one level of a recursive type, and apply a map to
every element of it (a fixed point's constructor and destructor, and the
fold and unfold they induce when the algebra or coalgebra is itself a
structure map); and follow a declared determination, the type-level
foreign key or functional dependency. Nothing else is admitted. Two
consequences carry the later proofs. On product types only projection,
pairing and declared determinations arise, and they are all the results
on products use. And a witness is invertible only when its inverse can
itself be written from the primitives: a projection that drops columns
determined by the ones it keeps is invertible, its inverse pairing them
back through the declarations, and a determination declared in both
directions is invertible, the reverse being a primitive and coherence
(Definition~\ref{def:coherent-declarations}) making the round trip the
identity, which is the alternate-key reading of
Remark~\ref{ex:employee-email}. A one-way declaration that happens to be
one-to-one on the data is not invertible: the reverse was never declared,
and the type level does not look at data. Types are therefore compared
through their structure and their declarations, never through accidental
bijections of their carriers such as
$\mathbb{N} \cong \mathbb{N} \times \mathbb{N}$. The restriction is
load-bearing rather than merely hygienic:
Lemma~\ref{lem:slotwise-structure-map} shows that on products carrying no
declared determination such maps cannot entangle slots, which is what the
irreducibility argument of Theorem~\ref{thm:grain-product} rests on.

\begin{definition}[Coherent declaration set]
\label{def:coherent-declarations}
A set of declared determinations is \emph{coherent} when any two chains
of declared determinations between the same types agree on the valid
inhabitants: for declared chains $d_k \circ \cdots \circ d_1$ and
$d'_m \circ \cdots \circ d'_1$, both from $A$ to $B$, the two composites
are equal on the inhabitants the declarations admit. A closed loop is the
case in which one chain is empty, so every loop of declared
determinations composes to the identity on the valid inhabitants. All
declaration sets are assumed coherent.
\end{definition}

\begin{remark}[Coherence within one type is automatic]
\label{rem:coherence-automatic}
For chains declared among the components of a single type, agreement
on the valid inhabitants is a consequence of the constrained-inhabitants
reading (Section~\ref{sec:prelim}): a valid record satisfies every
declaration at once, so two chains from one component to another read
the same value off it, and a cycle returns the value it started from.
What the reading does not supply is that every value of a component
extends to a valid record: where two chains disagree on some value, no
record carries it, and the projection onto that component stops being
onto. Coherence is therefore assumed rather than derived, and across
types nothing even suggests it: there it is the familiar alternate-key
discipline, but across types, which is what a domain intends when it
declares maps both ways.
\end{remark}

\noindent
The two subtype relations differ precisely by the admissible witness,
and the distinction is load-bearing. Structural subtyping $\ssub$ is the \emph{schema-decidable} relation,
comparing types by their components, and is what grain \emph{irreducibility}
(Definition~\ref{def:grain}) and the CalcG algorithm use. Type-level
subtyping $\subt$ is the \emph{semantic} relation, closed under
declared determinations, over which the grain \emph{ordering} of
Section~\ref{sec:relations} is defined. Because $\ssub$ implies
$\subt$ but not conversely, the two do not collapse:
$\textit{CustomerId} \spsub \textit{CustomerId} \times
\textit{CustomerName}$ is proper even though the two types are
isomorphic under the declared determination
$\textit{CustomerId} \thra \textit{CustomerName}$.

\noindent
For example, if $\textit{CustomerProfile} = \textit{CustomerId}\;\times\;
\textit{CustomerName}$, then $\textit{CustomerProfile} \ssub
\textit{Customer}$ (hence $\textit{CustomerProfile} \subt
\textit{Customer}$), witnessed by the canonical projection. A purely
\emph{declared} inclusion arises when every employee belongs to exactly
one department and every department has at least one employee: the
foreign key $\textit{EmployeeId} \thra \textit{DepartmentId}$ is a
surjection, so $\textit{DepartmentId} \subt \textit{EmployeeId}$ even
though $\textit{DepartmentId} \not\ssub \textit{EmployeeId}$
($\textit{DepartmentId}$ is not a component of $\textit{EmployeeId}$).
We define the \emph{type-level union}
$A \tun B$ as the type containing all fields of $A$ and $B$ (common fields
appearing once), the \emph{intersection} $A \tin B$ as the type containing
only fields common to both, and the \emph{difference} $A \tdiff B$ as the
fields in $A$ but not in~$B$.
Type-level union differs from the product~$\times$ exactly by
\emph{identifying} shared components:
$(X \times Y) \tun (Y \times Z) = X \times Y \times Z$ (a single copy
of~$Y$), whereas $(X \times Y) \times (Y \times Z)$ keeps two;
equivalently $A \tun B = A \times (B \tdiff A)$. It is also distinct from
the coproduct~$+$ of the grammar (a tagged variant, ``either an $A$ or a
$B$''): $\tun$ combines the \emph{fields} of two record schemas, whereas
$+$ combines the \emph{inhabitants} of two types.

\smallskip\noindent\textbf{Natural extension to any algebraic data type.}\quad
The formulas above are the product-type realization. The relation
$\subt$ is already universal by Definition~\ref{def:type-subset}
(surjection existence). The operations $\tin$ and $\tun$ return in
Section~\ref{subsec:incomparability-lattice} as the join and meet of a
bounded lattice on \emph{grains} (its lub and glb) rather than on types,
so
Section~\ref{sec:relations} develops the grain level rather than
re-defining the formulas above. The operation
$\tdiff$ is product-specific and is used only in the
relational-algebra inference rules of
Section~\ref{sec:inference}, which restricts to product types
accordingly. Table~\ref{tab:notation} collects the notation.

\begin{table}[t]
\caption{Key notation used throughout the paper.}
\label{tab:notation}
\centering
\small
\begin{tabular}{@{}ll@{}}
\toprule
\textbf{Symbol} & \textbf{Meaning} \\
\midrule
$A \ssub B$      & Structural subtype, schema-decidable (Def.~\ref{def:structural-subtype}) \\
$A \spsub B$     & Proper structural subtype (Def.~\ref{def:structural-subtype}) \\
$A \subt B$      & Type-level subtype, surjection (Def.~\ref{def:type-subset}) \\
$A \psub B$      & Proper type-level subtype (Def.~\ref{def:type-proper-subset}) \\
$A \tun B$       & Type-level union \\
$A \tin B$       & Type-level intersection \\
$A \tdiff B$     & Type-level difference \\
$R_1 \cong R_2$  & Type isomorphism \\
$\grain{R}$      & Grain of type $R$ (Def.~\ref{def:grain}) \\
$\ek{R}$         & Entity key of type $R$ (Section~\ref{sec:entity}) \\
$\eqg$           & Grain equality (Section~\ref{sec:relations}) \\
$\leg$           & Grain ordering (Section~\ref{sec:relations}) \\
$\stg$           & Strict grain ordering (Section~\ref{sec:relations}) \\
$\incg$          & Grain incomparability (Section~\ref{sec:relations}) \\
\bottomrule
\end{tabular}
\end{table}
 
\section{Grain Theory: Definitions and Fundamental Theorems}
\label{sec:foundations}

This section defines grain formally and establishes the facts the rest of
the paper rests on: every type has a grain, unique up to isomorphism; the
grain projection is injective, so a grain \emph{identifies} each element;
and grain is computed compositionally for product and sum types. On
product types the grain coincides with a minimal key, the bridge to
classical dependency theory, though the definition itself needs no such
machinery.

\subsection{Definition of Grain}
\label{subsec:grain-def}

The \emph{grain} of a data type captures its \emph{level of detail}: the
minimum information needed to uniquely identify each element. Intuitively,
the grain of $R$ is the simplest type isomorphic to $R$, carrying only the
essential identifying structure while discarding redundant fields.

\begin{definition}[Grain]
\label{def:grain}
Given a data type $R$, a type $G$ is a \emph{grain} of $R$ if it
satisfies:
\begin{enumerate}
\item \textbf{Isomorphism.} $G$ and $R$ are isomorphic:
  there exists an isomorphism
  $f_g : G \stackrel{\cong}{\longrightarrow} R$,
  called the \emph{grain function}.
\item \textbf{Irreducibility.} No proper \emph{structural} subtype of
  $G$ is isomorphic to $R$: $\neg\exists\, G' \spsub G$ with
  $G' \cong R$. The comparison must be structural, slot by slot
  (Definition~\ref{def:structural-subtype}; Section~\ref{sec:adt} for sums
  and the recursive constructors): under the semantic relation $\subt$ the
  clause is vacuous, since any $T \cong R$
  gives both $T \subt G$ and $G \subt T$, so $G$ has no proper
  $\subt$-subtype isomorphic to~$R$. Every type isomorphic to~$R$ would
  then qualify, and $\textit{CustomerId} \times \textit{CustomerName}$
  would be a grain of $\textit{Customer}$ alongside
  $\textit{CustomerId}$.
\end{enumerate}
\end{definition}

\smallskip\noindent\textbf{Notation.}\quad
Definition~\ref{def:grain} is a \emph{relation}: ``$G$ is a grain of~$R$''
holds when $G$ satisfies clauses~(1)--(2). Write $\textit{Grains}(R)$ for its
\emph{fiber} over~$R$, which Lemma~\ref{lem:grain-existence} makes non-empty
and Theorem~\ref{thm:multiple-grains} contained in a single isomorphism
class. We write
$\grain{R}$ for a member: a claim about $\grain{R}$ holds of every member, and
$\grain{R} = T$ asserts that the concrete type~$T$ is one. That assertion is
strictly stronger than the grain-equality $T \eqg R$ of
Section~\ref{sec:relations}, which requires only $T \cong R$, and the
equi-join inference of Section~\ref{sec:inference} turns on the gap. Since
$\eqg$ and $\leg$ already compare grains, $A \eqg \grain{B}$ means $A \eqg B$;
we write the shorter form.

\smallskip\noindent
Which member is \emph{displayed} is no schema-derivable quantity but a
declaration of domain analysis, which the catalogue records. It is immaterial
wherever only the grain's identifying power is used, which is most of this
paper, and decisive wherever structure is read off the representative: the
natural grain (Definition~\ref{def:natural-grain}) and the behavioral class it
fixes (Definition~\ref{def:behavioral-class}). What no result here supplies is
a \emph{construction}, a way to compute a grain of~$R$ from~$R$ itself;
Definition~\ref{def:compositional-grain} supplies one, by recursion on the type
expression and the declared determinations, and Theorem~\ref{thm:grain-adt}
proves it sound.

\noindent
The inverse $f_g^{-1}$, written $\textit{grain} : R \thra \grain{R}$, is the
\emph{grain projection}, sending each element $r : R$ to its identifying
component, written $\textit{grain}\ r$.

\smallskip\noindent\textbf{Example.}\quad
Let $\textit{Customer} = \textit{CustomerId} \times \textit{CustomerName}
\times \textit{Email}$. Then $\grain{\textit{Customer}} = \textit{CustomerId}$,
an \emph{internal} grain: its fields are part of the type. A grain may also be
\emph{external}, as $\textit{BalanceDate}$ is for a collection
$\textit{Coll}\ \textit{MonthlyBalance}$, identifying each monthly snapshot
without being a field of the collection type. Irreducibility is in both cases a
property of $\grain{R}$ itself, that no proper structural subtype of it is
isomorphic to~$R$: internally this says no smaller subset of $R$'s fields
suffices, externally that no smaller type determines~$R$. Isomorphism and the
grain ordering, by contrast, are governed by $\subt$
(Definition~\ref{def:type-subset}), so Lemma~\ref{thm:grain-subset} applies to
both alike.

\begin{lemma}[Existence of a grain]
\label{lem:grain-existence}
Every data type has a grain.
\end{lemma}

\begin{proof}
For a product type $R$, the structural subtypes of $R$ that are
isomorphic to $R$ form a finite, nonempty set: $R$ has finitely many
structural subtypes (one per subset of its components), and the set is nonempty because it contains $R$
itself, witnessed by the identity and requiring no declared
determination. Let $G$ be a $\ssub$-minimal
element. Then $G \cong R$ gives clause~(1) of
Definition~\ref{def:grain}, and any $G' \spsub G$ with $G' \cong R$
would be a strictly smaller element of the set, contradicting
minimality, so clause~(2) holds and $G$ is a grain. For an arbitrary
algebraic data type a grain is exhibited by the compositional
construction of Section~\ref{sec:adt}
(Definition~\ref{def:compositional-grain},
Theorem~\ref{thm:grain-adt}); an external grain is supplied by the
declaration that specifies the type.
\end{proof}

\begin{remark}[Existence is unconditional; determinations reduce]
\label{rem:grain-existence-declared}
Existence needs no functional dependency: $R$ is a structural subtype
of itself and $R \cong R$, so a grain always exists, and absent any
declared determination $R$ is already irreducible with $\grain{R} = R$.
A declared determination, a functional dependency or key that comes from
domain knowledge rather than schema structure, is precisely what makes a
\emph{proper}
structural subtype of $R$ isomorphic to $R$, reducing the grain
below~$R$. Hence \emph{that} a type has a grain is automatic;
\emph{which} grain it has is fixed by its declared determinations.
\end{remark}

\begin{remark}[Grain as Irreducible Core]
\label{rem:grain-prime}
The grain is the canonical minimal representative of the isomorphism class
$[R] = \{T : T \cong R\}$: no proper structural subtype of it is isomorphic
to~$R$, so it is the simplest type carrying $R$'s full identifying structure,
the atom on which the rest of the type's information is built.
\end{remark}

This irreducibility has a striking consequence: the grain of the
co-domain determines any function targeting it.

\begin{proposition}[Grain Factorization]
\label{prop:grain-factorization-codomain}
For any function $h : R_1 \to R_2$ there is a \emph{unique}
$e : R_1 \to \grain{R_2}$ with $h = f_{g_{R_2}} \circ e$, namely
$e = \textit{grain}_{R_2} \circ h$. Equivalently, the map
$e \mapsto f_{g_{R_2}} \circ e$ is a bijection from the \emph{set} of functions
$R_1 \to \grain{R_2}$ onto the \emph{set} of functions $R_1 \to R_2$.
\end{proposition}

\begin{proof}
For existence, take $e = \textit{grain}_{R_2} \circ h$, where
$\textit{grain}_{R_2} = f_{g_{R_2}}^{-1}$; then
$f_{g_{R_2}} \circ e = h$. For uniqueness, $f_{g_{R_2}}$ is an isomorphism
(Definition~\ref{def:grain}) and so injective, whence
$f_{g_{R_2}} \circ e = f_{g_{R_2}} \circ e'$ gives $e = e'$. The assignments
$e \mapsto f_{g_{R_2}} \circ e$ and $h \mapsto \textit{grain}_{R_2} \circ h$
are therefore mutually inverse: composed in either order they give the
identity, since
\[
  f_{g_{R_2}} \circ (\textit{grain}_{R_2} \circ h)
    = \mathrm{id}_{R_2} \circ h = h
  \quad\text{and}\quad
  \textit{grain}_{R_2} \circ (f_{g_{R_2}} \circ e)
    = \mathrm{id}_{\grain{R_2}} \circ e = e.
\]
The first is the existence step. The second is uniqueness, computed
directly: $\textit{grain}_{R_2}$ is a left inverse of $f_{g_{R_2}}$, which
is why $f_{g_{R_2}}$ is injective.
\end{proof}

\noindent
Uniqueness is what the proposition asserts; existence alone would say
nothing. A witness $e$ can be manufactured from~$h$ for \emph{any}
isomorphism into~$R_2$, so exhibiting one establishes no factorization at
all. What has content is that the witness is \emph{the} one: every
transformation into~$R_2$, whatever its domain, corresponds to exactly one
function into $\grain{R_2}$, so the two are interchangeable and nothing is
lost by working with the second. All information flowing into a type is
channeled through its grain. (The function sets here are those of the
metatheory, in which we reason about transformations; the grammar of
Section~\ref{sec:prelim} has no function types, which is what makes its fixed
points exist.)
Section~\ref{subsec:grain-factorization} refines~$e$ further using the
\emph{grain lift}, establishing that grain projection commutes with
transformation.

\subsection{Fundamental Theorems}
\label{subsec:fund-theorems}

Proofs for the results below are in Appendix~\ref{sec:proofs}. The two
short arguments above, for Lemma~\ref{lem:grain-existence} and
Proposition~\ref{prop:grain-factorization-codomain}, are given in place
because each is a few lines and each is used immediately.

\begin{theorem}[Multiple Grains Isomorphism]
\label{thm:multiple-grains}
If $G_1$ and $G_2$ are both grains of $R$, then $G_1 \cong G_2$.
\end{theorem}

\noindent\textbf{Example.}\quad
A $\textit{WeatherObservation}$ type may admit two grains:
$G_1 = \textit{ObservationId}$ and
$G_2 = \textit{ObservationDateTime} \times \textit{WeatherStationId}$.
The theorem gives $G_1 \cong G_2$: identifiers correspond one-to-one with
(datetime, station) pairs. For everything that reasons about \emph{level of
detail} the two are interchangeable, and that is most of this paper: the grain
relations of Section~\ref{sec:relations} compare grains only up to
isomorphism, and the inference of Sections~\ref{sec:inference}
to~\ref{sec:calcg} carries them through a pipeline on that basis.

Interchangeable is not equally informative, and the difference deserves a
name. Call the \emph{denotational power} of a grain the extent to which its
components display what the data \emph{means} rather than merely telling its
elements apart. Two things carry it: the \emph{entity key} $\ek{R}$, saying
which entity an element is about (Definition~\ref{def:entity-key}), and, where
the elements stand in a causal or other order, the component recording it
(Definition~\ref{def:behavioral-class}). Taken
together these fix what Section~\ref{sec:entity} calls the \emph{behavior}
of a data type, and behavior in turn fixes the canonical way a collection
of that type must be read and written: current state by upsert, an event
log append-only, a versioned collection only through a point-in-time
query. A grain that displays them is therefore the most useful thing one
can hand to whoever writes the transformation, engineer or coding agent,
since it says not only how finely the data is cut but what may legitimately
be done with it.

We call such a grain the \emph{natural} or \emph{business grain} of a
type (Definition~\ref{def:natural-grain}), by direct analogy with the
natural and business keys of a source
system~\cite{kimball1996data,linstedt2015dv}. In the example, $G_1$ is a
surrogate: $\textit{ObservationId}$ identifies, and says nothing else.
$G_2$ is the natural grain: it identifies equally well, and in doing so
names the station an observation is about and the instant at which it
happened, from which Section~\ref{sec:entity} reads the entity key
$\ek{R} = \textit{WeatherStationId}$ and the behavioral class
\textbf{IsEvent}. Both are correct grains; one says more. And because the
entity is not recoverable from the type's structure but is a declaration
made by domain analysis (Definition~\ref{def:entity}), so is the natural
grain: this is a further reason grain is \emph{declared} rather than
inferred (Remark~\ref{rem:grain-existence-declared}).

The converse does not hold: $G_2 \cong G_1$ alone does not make $G_2$ a
grain of~$R$; irreducibility must also be verified independently.

\begin{theorem}[Grain Uniqueness]
\label{thm:grain-uniqueness}
The grain projection $\textit{grain} : R \thra \grain{R}$ is injective:
for any $r_1, r_2 : R$, if $\textit{grain}\ r_1 = \textit{grain}\ r_2$
then $r_1 = r_2$. Here $\textit{grain}\ r_i$ is the projection applied
to the element $r_i$ (functional notation for $\textit{grain}(r_i)$);
the grain itself is a property of the \emph{type}~$R$, not of an
individual element.
\end{theorem}

\begin{remark}[Grain and keys]
\label{rem:grain-vs-pk}
On a product type the grain and the candidate key are the same thing, in
both directions. Theorem~\ref{thm:grain-uniqueness} supplies
identification, so a grain determines every component; irreducibility
(Definition~\ref{def:grain}) supplies minimality, since no proper
structural subtype of it still does; and identification with minimality is
candidacy. Conversely a candidate key determines the whole relation and is
minimal in doing so, hence is isomorphic to it and irreducible, hence a
grain. A relation with several candidate keys is therefore a type with
several grains, which is Theorem~\ref{thm:multiple-grains} seen from the
relational side. Both notions live at the schema level, and neither is a
statement about any particular collection. Grain implication accordingly
coincides with functional-dependency implication
(Corollary~\ref{cor:grain-implication}); Section~\ref{subsec:rw-fd} gives
the full correspondence, with axioms A1--A9 recovering Armstrong's and
completeness transferring (Proposition~\ref{prop:grain-completeness}). We claim no new dependency theory, since on
relations grain \emph{is} minimal-key reasoning. It
differs from a declared key in three respects. First, a key is a
property of a base table in a catalogue, whereas a pipeline continually produces collections on the
fly, such as intermediate results and in-memory data frames, that live in
no catalogue and carry no declared key; the
grain, being fixed by the \emph{type}, is defined for all of them. That is
what lets a transformation be reasoned about at the type level, at design
time, before any of its data exists.
Second, grain is defined for \emph{every} algebraic data type, sums,
and inductive and coinductive types such as streams $\nu X.\,R \times
X$, where a candidate key has no referent at all
(Theorem~\ref{thm:grain-adt}). Third, the grain is not merely an
identifier but a \emph{lens} on the data: the very same type, under different declared grains, denotes
fundamentally different things, a current entity, a versioned record,
an event, as the next subsection
(Section~\ref{subsec:grain-semantics}) shows. A precise grain is thus a
business declaration, not a schema fact; it fixes the entity key and
behavioral class that a candidate key cannot express, formalized as the
denotation $(\grain{R},\, \ek{R},\, \BC{R})$ in
Section~\ref{sec:entity}. On a product type grain \emph{is} a candidate
key, and where a relation has several, any of them. On a sum or a
recursive type there is no candidate key for it to be, and grain is
defined there regardless. In both cases it carries what a key cannot.

A related identifier must be kept distinct from the grain, and the
distinction is one of \emph{kind}, not of tightness. The grain is a
property of the type; an \emph{instance-level} identifier is a fact about
one collection at one moment, and may be strictly smaller. A point-in-time
read of an \textbf{IsMultiVersion} collection (Section~\ref{sec:entity})
leaves the grain $\ek{R} \times \textit{FromDtm}$ unchanged, the type being
unchanged, yet each entity then occurs once, so $\ek{R}$ alone identifies.
Pushed to its limit the notion says nothing about the type at all: in a
collection of one row, every attribute identifies. The next theorem locates
the difference exactly, in a quantifier.
\end{remark}

\begin{theorem}[Type grain versus collection key]
\label{thm:type-grain-vs-collection-key}
Let $C\,R$ denote a bag of elements of~$R$, and let $k : R \thra K$ be a
field-set projection, so that $K$ retains a subset of $R$'s components. The
following are equivalent:
\begin{enumerate}
\item $k$ is injective on the \emph{type}~$R$;
\item $k$ is a superkey of \emph{every} collection $c : C\,R$, that is,
  injective on the elements of each;
\item $k$ is itself an isomorphism $R \cong K$, whence $R \eqg K$.
\end{enumerate}
Moreover $K$ is a \emph{grain} of~$R$ exactly when no proper sub-projection
of~$k$ satisfies these conditions.
\end{theorem}

\begin{proof}
(1)$\Rightarrow$(2) is immediate: injectivity on~$R$ gives injectivity on
the elements of any collection of $R$-elements. For (2)$\Rightarrow$(1),
contrapositively, a collision $k\,r_1 = k\,r_2$ with $r_1 \neq r_2$ makes
$k$ non-injective on the two-element collection holding both, so~$k$ is a
superkey of not every collection. For
(1)$\Leftrightarrow$(3), $k$ is surjective by construction and a surjection
is injective exactly when it is a bijection, which is
Definition~\ref{def:grain}'s isomorphism clause with~$k$ as its witness;
$R \eqg K$ follows. The final clause is Definition~\ref{def:grain}'s
irreducibility clause read on~$K$.
\end{proof}

\begin{remark}[The quantifier is the whole difference]
\label{rem:key-quantifier}
A declared key holds of the one collection its base table holds.
Clause~(2) quantifies over \emph{every} collection of $R$-elements, and that
quantifier cannot be checked against data: any finite collection can be
filtered until a proper sub-projection becomes injective on it, which is why
relational systems \emph{enforce} declared keys rather than infer them. The
grain discharges the quantifier at the type level instead, by
Theorem~\ref{thm:grain-uniqueness}. That is what makes it defined for the
intermediate results and in-memory frames a pipeline produces on the fly,
which live in no catalogue and carry no declared key, and it is the first of
the three respects of Remark~\ref{rem:grain-vs-pk} made precise.
\end{remark}

\begin{theorem}[Grain of a Grain, Idempotency]
\label{thm:grain-idempotent}
$\grain{\grain{R}} = \grain{R}$.
\end{theorem}

\begin{lemma}[Structure maps between determination-free products act slotwise]
\label{lem:slotwise-structure-map}
Let $R = R_1 \times \cdots \times R_n$ be a product in which no
determination is declared among the components. Let
$\mathcal{I} \subseteq \{1, \dots, n\}$ and let
$Q = \prod_{i \in \mathcal{I}} Q_i$ be a product whose slots are indexed
by~$\mathcal{I}$, with $Q_i \ssub \grain{R_i}$ for each $i \in \mathcal{I}$.
Then every isomorphism $\tau : Q \stackrel{\cong}{\longrightarrow} R$ in the
structure-map class of Section~\ref{sec:prelim} has
$\mathcal{I} = \{1, \dots, n\}$ and decomposes as
$\tau = \tau_1 \times \cdots \times \tau_n$ for isomorphisms
$\tau_i : Q_i \stackrel{\cong}{\longrightarrow} R_i$.
\end{lemma}

\begin{theorem}[Grain of products and sums]
\label{thm:grain-product}
\label{thm:grain-sum}
Let $R_1, \dots, R_n$ be data types.
\begin{enumerate}
\item If no determination is declared among the components of a product, so
  that no $R_i$ is determined by the others, then
  $\grain{R_1 \times \cdots \times R_n}
     = \grain{R_1} \times \cdots \times \grain{R_n}$.
\item For a sum,
  $\grain{R_1 + \cdots + R_n} = \grain{R_1} + \cdots + \grain{R_n}$,
  unconditionally.
\end{enumerate}
\end{theorem}

\noindent
The condition separates the legs: a product can hide a determination across
its components, a sum cannot. Where one holds the product is reducible and the
Grain Inference theorem computes its grain, so
$\grain{\textit{CustomerId} \times \textit{CustomerName}} =
\textit{CustomerId}$ when the first determines the second. In~(2), $+$ is the
coproduct, a tagged variant and not the field-set union~$\tun$: for
$\textit{Vehicle} = \textit{Car} + \textit{Motorcycle}$ we obtain
$\grain{\textit{Vehicle}}
   = \grain{\textit{Car}} + \grain{\textit{Motorcycle}}$.

\subsection{Grain Determines Data Behavioral Semantics}
\label{subsec:grain-semantics}

The same field list under different grain declarations yields
fundamentally different semantics, and, the declaration being part of the
type (Section~\ref{sec:prelim}), different types, just as
\textnormal{CREATE TABLE} statements that differ only in their primary-key
clause are different schemas. Consider $\textit{Customer} =
\textit{CustomerId} \times \textit{Name} \times \textit{Email} \times
\textit{Address} \times \textit{EffectiveFrom} \times \textit{CreatedOn}$
under four grain declarations:
\begin{itemize}
\item $\grain{\textit{Customer}} = \textit{CustomerId}$
  \,---\, \emph{entity}: each row corresponds to a customer's current state.
\item $\grain{\textit{Customer}} = \textit{CustomerId} \times
  \textit{EffectiveFrom}$ \,---\, \emph{versioned}: each row corresponds to a time-stamped version of a customer's state.
\item $\grain{\textit{Customer}} = \textit{CustomerId} \times
  \textit{Address} \times \textit{EffectiveFrom}$ \,---\,
  \emph{versioned (compound entity)}: each row corresponds to a version
  of a customer-address pair, allowing multiple addresses per customer
  and independent version histories per pair.
\item $\grain{\textit{Customer}} = \textit{CustomerId} \times
  \textit{CreatedOn}$ \,---\, \emph{event}: each row corresponds to a customer creation event at a given date.
\end{itemize}

\noindent
Without a precise grain declaration the data cannot be interpreted at all:
one field list, four grain declarations, four types, four read/write
disciplines; the four grains are not alternative grains of one type in the
sense of Theorem~\ref{thm:multiple-grains}, which would be isomorphic, but
the grains of four types. In practice engineers
resolve the ambiguity by guesswork, from column names, sample data or tribal
knowledge, which is the root of the errors of Section~\ref{sec:intro}. These
readings return as the \emph{behavioral classes} of
Section~\ref{sec:entity}, where the declared constraints settle which of the
four grains is the \emph{natural} one (Definition~\ref{def:natural-grain}).

\section{Grain Relations}
\label{sec:relations}

Having defined grain as a type-level property, we now establish three fundamental relations on grains, namely equality,
ordering and incomparability, and show that on product types they form
a bounded lattice.

\noindent
Proofs for this section are in Appendix~\ref{sec:proofs}.

\subsection{Grain Equality}
\label{subsec:grain-equality}

\begin{definition}[Grain Equality]
\label{def:grain-eq}
Given two data types $R_1$ and $R_2$, we say they have \emph{equal grain}
and write $R_1 \eqg R_2$ if and only if there exists an isomorphism
$f : \grain{R_1} \stackrel{\cong}{\longrightarrow} \grain{R_2}$ between
their grains.
\end{definition}

\begin{theorem}[Grain Equality Theorem]
\label{thm:grain-equality}
$R_1 \eqg R_2$ if and only if $R_1$ and $R_2$ are isomorphic:
$R_1 \eqg R_2 \;\Leftrightarrow\; R_1 \cong R_2$.
\end{theorem}

\noindent
Grain equality is an equivalence relation (reflexive, symmetric,
transitive).

\smallskip\noindent\textbf{Example.}\quad
If $\textit{Customer} \eqg \textit{WebCustomer}$, then each customer
corresponds to exactly one web account, and vice versa, since the two types
represent essentially the same entities at the same level of detail but from a different perspective.

\subsection{Grain Ordering}
\label{subsec:grain-ordering}

A natural generalization asks not just when two types have the same grain,
but when one is \emph{finer} than the other.

\begin{definition}[Grain Ordering]
\label{def:grain-ord}
Given two data types $R_1$ and $R_2$, we say $R_1$ has \emph{lower grain}
(finer granularity) than $R_2$, and write $R_1 \leg R_2$, if and only if
there exists a surjective function
$f : \grain{R_1} \thra \grain{R_2}$ between their grains.
All functions in grain theory are total.
\end{definition}

\noindent
We write $R_1 \stg R_2$ for the strict order, abbreviating
$R_1 \leg R_2 \wedge \neg(R_2 \leg R_1)$.

\noindent
The relations $\eqg$ and $\leg$ compare the \emph{grains} of types. Like
$\subt$ they are stated for arbitrary data types, nothing is redefined
for a narrower setting, but they live one level above the type
relations of Section~\ref{sec:prelim}: $\ssub$ and $\subt$ compare
\emph{types}, $\eqg$ and $\leg$ compare their \emph{grains}. The two
levels are related, not identified; Lemma~\ref{thm:grain-subset} below
makes the bridge precise (and shows the direction reverses).

\begin{corollary}[Type-Level Reading of Grain Ordering]
\label{thm:grain-ordering}
Since $R \cong \grain{R}$, grain ordering lifts from the grains to the
types themselves: $R_1 \leg R_2$ if and only if there exists a surjective
function $h : R_1 \thra R_2$, a many-to-one correspondence in which each
$R_2$ element has at least one $R_1$ preimage. (Compose the grain
surjection of Definition~\ref{def:grain-ord} with the grain
isomorphisms; Appendix~\ref{sec:proofs}.)
\end{corollary}

\noindent
Grain ordering (Definition~\ref{def:grain-ord}) requires only the
existence of a surjection between grains. The witness comes in one
of two forms: (1)~a \emph{projection} from a shared product
structure, for instance $\textit{OrderDetail} \leg \textit{Order}$ via
$(\textit{orderId}, \textit{lineItemId}) \mapsto \textit{orderId}$,
which is verifiable from the schema alone; or (2)~a \emph{declared
determination} from domain knowledge such as a foreign-key
reference, for instance $\textit{Employee} \leg \textit{Department}$ via
$\textit{EmployeeId} \thra \textit{DepartmentId}$. The Armstrong-style
axioms of Section~\ref{subsec:armstrong} reason about both:
reflexivity yields schema-verifiable orderings, augmentation and
transitivity propagate declared ones. Surjectivity is type-level:
it asserts that every coarser-grained value \emph{can} be determined
by some finer-grained value, not that the determination holds in
every instance.

\begin{remark}[Grain vs.\ Functional Dependencies]
\label{rem:grain-vs-fd}
Grain theory and functional-dependency theory share a relation but not a
problem, and they act in sequence, not competition. Functional
dependencies are a tool of \emph{schema design}: over loose attributes
they expose redundancy so a database can be normalized. That work is
finished before a pipeline is built. By then the source and target
schemas, and hence their \emph{types}, are fixed. Grain theory begins there,
taking the fixed types as given and reasoning about the \emph{level of
detail} of data and its transformations. This is why grain relates
\emph{whole types}, not ad-hoc attribute sets. In the pipeline phase the
fixed relations are the atomic units. A determinant such as
$\{\textit{StoreId}, \textit{Season}\}$ is a bag of attributes drawn from
two different types, so it is not a grain-level object at all.

On product types the two coincide \emph{as a relation}. The ordering
$R_1 \leg R_2$ is a surjection $\grain{R_1} \thra \grain{R_2}$ between
grain types. It holds precisely when $\grain{R_1}$ functionally
determines $\grain{R_2}$, with the left side a \emph{whole grain}. Grain states this between
\emph{types}, an FD between \emph{attributes}; but the data engineer
reads it as ``$R_1$ is finer-grained than $R_2$'' (what governs
transformation correctness), not ``$R_1$ determines $R_2$'' (a redundancy
fact already settled at design time). The full correspondence, covering the level
distinction, the Armstrong-axiom recovery and the completeness
transfer, is established in Section~\ref{subsec:rw-fd}.
\end{remark}

\begin{lemma}[Grain Subset--Ordering Equivalence]
\label{thm:grain-subset}
$\grain{R_1} \subt \grain{R_2}$ if and only if $R_2 \leg R_1$.
(One unfolding of the definitions; note the contravariance. We state it
as a named bridge because it is used repeatedly and the reversal is easy
to get wrong.)
\end{lemma}

\begin{corollary}[Grain Determines All Type Subsets]
\label{cor:grain-all-subsets}
If $G = \grain{R}$, then for any $R'$ such that $R' \subt R$, we have
$G \leg R'$. The grain of a type determines any subset of that
type's fields.
\end{corollary}

\noindent
\begin{theorem}[Grain Ordering is a Partial Order]
\label{thm:grain-partial-order}
$\leg$ is a partial order (up to isomorphism):
\begin{itemize}
\item \textbf{Reflexivity.} $R \leg R$ (identity is surjective).
\item \textbf{Antisymmetry.} $R_1 \leg R_2 \wedge R_2 \leg R_1
  \Rightarrow R_1 \eqg R_2$. Established on product types and on sums of
  products; on fixed points, for witnesses induced from surjections
  between the bodies, the general fixed-point case being open
  (Appendix~\ref{sec:proofs}).
\item \textbf{Transitivity.} $R_1 \leg R_2 \wedge R_2 \leg R_3
  \Rightarrow R_1 \leg R_3$ (composition of surjections is surjective).
\end{itemize}
\end{theorem}

\noindent\textit{Proof of antisymmetry.}\quad
On product types by attribute closure (Lemma~\ref{lem:sat-closure}), the
two component sets having the same closure and each being isomorphic to
the product on it; on sums by matching summands; given coherent
declarations (Definition~\ref{def:coherent-declarations}).
Appendix~\ref{sec:proofs}.

\begin{remark}[Antisymmetry across schemas: Employee and EmailAccount]
\label{ex:employee-email}
Let $\textit{Employee}$ carry
$\grain{\textit{Employee}} = \textit{EmployeeId}$ and
$\textit{EmailAccount}$ carry
$\grain{\textit{EmailAccount}} = \textit{EmailAddr}$, with two
declared determinations, both onto: every employee has a mailbox
($\textit{EmployeeId} \thra \textit{EmailAddr}$) and every mailbox
has an employee owner
($\textit{EmailAddr} \thra \textit{EmployeeId}$). The two types
share no components, so neither surjection is structural; antisymmetry
nonetheless yields $\textit{Employee} \eqg \textit{EmailAccount}$:
the two populations correspond one to one, $\textit{EmployeeId}$ and
$\textit{EmailAddr}$ being alternate identifiers of one entity
(Section~\ref{sec:entity}). The declared two-cycle also shows
coherence (Definition~\ref{def:coherent-declarations}) at work: the
owner of an employee's mailbox is that employee. On finite populations
mutual onto maps are bijective by counting alone, but the inverse the
type level admits is the one the declarations build; coherence supplies
it, on finite and infinite populations alike, and it is what the domain
means.
\end{remark}

\smallskip\noindent
Antisymmetry holds up to isomorphism, not strict type equality: types
with isomorphic grains are interchangeable for all correctness
reasoning, even across system boundaries where naming conventions
differ. This distinguishes grain ordering from functional dependencies,
which form only a preorder (Section~\ref{subsec:rw-fd}).

\begin{remark}[Relational Intuition]
\label{rem:relational-interp}
For readers familiar with relational databases: $R_1 \leg R_2$
captures the structural pattern of a foreign key (many-to-one), and
$R_1 \eqg R_2$ a one-to-one correspondence.
\end{remark}

\subsection{Grain Theorems for Type Operations}
\label{subsec:grain-type-ops}

The grain ordering enables several results about how grain interacts with
type-level set operations.

\begin{theorem}[Grain Inference]
\label{thm:grain-inference-sufficient}
For any data type $R$, if there exists a type $G$ such that
(i)~$G \subt R$, (ii)~$G \leg R$, and (iii)~$\grain{G} = G$, then
$\grain{R} = G$: that is, $G$ satisfies Definition~\ref{def:grain} for $R$
($G \cong R$ and no proper sub-type of $G$ is isomorphic to $R$).
\end{theorem}

\noindent
Conditions~(i) and~(ii) run in \emph{opposite} directions and neither
implies the other. (i)~$G \subt R$ is a surjection $R \thra G$, so $G$ is
\emph{recoverable} from~$R$; (ii)~$G \leg R$ is a surjection
$\grain{G} \thra \grain{R}$, essentially $G \thra R$, so $G$
\emph{determines}~$R$. For instance $\textit{CustomerName} \subt
\textit{Customer}$ satisfies~(i) but not~(ii). Adding
(iii)~($G$ irreducible) then pins $\grain{R}$ to exactly~$G$.

\noindent
This theorem provides sufficient conditions for identifying the grain of a
type without direct computation, and it is a key tool in
Section~\ref{sec:inference}. Condition~(iii) is essential: it ensures
that $G$ is irreducible (no proper sub-type of $G$ is isomorphic to $R$).

\subsection{Armstrong Axioms for Grain Ordering}
\label{subsec:armstrong}

The grain ordering admits a complete axiomatization analogous to Armstrong's
for functional dependencies~\cite{armstrong1974dependency}. A1, A2 and A4
restate properties established above; A3, Augmentation, is what completes the
system. A2's premise is the structural relation $\ssub$, decidable from the
schema, so that a derivation is a syntactic object; the semantic $\subt$
would let one application of A2 assert any ordering that happens to hold.

\smallskip\noindent\textbf{Basic axioms.}

\begin{enumerate}
\item[\textbf{A1.}] \textbf{Self-determination.} $R \leg R$ for any type $R$.
\item[\textbf{A2.}] \textbf{Reflexivity.} If $R_1 \ssub R_2$, then
  $R_2 \leg R_1$.
\item[\textbf{A3.}] \textbf{Augmentation.} If $R_1 \leg R_2$, then
  $R_1 \tun R_3 \leg R_2 \tun R_3$ for any type $R_3$.
\item[\textbf{A4.}] \textbf{Transitivity.} If $R_1 \leg R_2$ and
  $R_2 \leg R_3$, then $R_1 \leg R_3$.
\end{enumerate}

\smallskip\noindent\textbf{Derived rules.}

\begin{enumerate}
\item[\textbf{A5.}] \textbf{Union.} If $R_1 \leg R_2$ and $R_1 \leg R_3$,
  then $R_1 \leg (R_2 \tun R_3)$.
\item[\textbf{A6.}] \textbf{Decomposition.} If $R_1 \leg (R_2 \tun R_3)$,
  then $R_1 \leg R_2$ and $R_1 \leg R_3$.
\item[\textbf{A7.}] \textbf{Composition.} If $R_1 \leg R_2$ and
  $R_3 \leg R_4$, then $(R_1 \tun R_3) \leg (R_2 \tun R_4)$.
\item[\textbf{A8.}] \textbf{Pseudotransitivity.} If $R_1 \leg R_2$ and
  $(R_2 \tun R_4) \leg R_3$, then $(R_1 \tun R_4) \leg R_3$.
\item[\textbf{A9.}] \textbf{Darwen's Theorem.} If $R_1 \leg R_2$ and
  $R_3 \leg R_4$, then
  $(R_1 \tun (R_3 \tdiff R_2)) \leg (R_2 \tun R_4)$.
\end{enumerate}

\smallskip\noindent\textbf{What the axioms are sound and complete for.}\quad
Soundness and completeness are claims relative to a semantics, and the
paper's semantics has been fixed since Section~\ref{sec:prelim}: a grain
relation holds when a structure map witnesses it, and the structure-map
class is generated by the declared determinations. An interpretation of
the base types is presupposed whenever a map is called onto, but none is
ever fixed: everything below holds for every interpretation in which the
declared surjections are coherent
(Definition~\ref{def:coherent-declarations}), because which maps exist is
settled by the declarations alone, so no carriers appear in the
definitions. What the completeness statement needs is only that the
declaration set be named, so that ``holds'' and ``entails'' have a
referent.

\begin{definition}[Grain assertions under a declaration set]
\label{def:grain-holds}
A \emph{grain assertion} is one of $R_1 \leg R_2$, $R_1 \eqg R_2$ or
$\grain{R} = T$, for types built from a set of base types. Write
$\mathrm{comps}(T)$ for the set of components of a product type~$T$.
\begin{enumerate}
\item A \emph{declaration set} $\Delta$ is a finite set of \emph{declared
  determinations}, each a pair $(U, V)$ of finite sets of base types; the
  surjections realizing them are assumed coherent
  (Definition~\ref{def:coherent-declarations}) in every interpretation
  considered. Under
  $\Delta$, a product type is inhabited by the records satisfying every
  determination $\Delta$ derives among its components, the constrained
  reading of Section~\ref{sec:prelim} with $\Delta$ as the catalogue, and
  the \emph{structure-map class generated by}~$\Delta$ is the class of
  Section~\ref{sec:prelim} whose declared generators are one surjection
  from the product of~$U$ onto the product of~$V$ per $(U, V) \in \Delta$.
\item An assertion \emph{holds under}~$\Delta$ when its witness lies in
  that class: a surjection $\grain{R_1} \thra \grain{R_2}$ for
  $R_1 \leg R_2$, an isomorphism $\grain{R_1} \cong \grain{R_2}$ for
  $R_1 \eqg R_2$, and for $\grain{R} = T$ the isomorphism and the
  irreducibility clause of Definition~\ref{def:grain}, both taken in that
  class.
\item A set $\Gamma$ of assertions on product types \emph{entails} an
  assertion~$\alpha$, written $\Gamma \models \alpha$, when $\alpha$ holds
  under every declaration set under which every member of $\Gamma$ holds.
  Within entailment an assertion $\grain{R} = T$ is read as $T \eqg R$:
  its irreducibility clause is a negative condition that no closure
  expresses, and the reduction of Corollary~\ref{cor:grain-implication}
  drops it in the same way. The declaration set \emph{named by}~$\Gamma$,
  written $\Gamma^{\flat}$, has one determination
  $(\mathrm{comps}(R_1), \mathrm{comps}(R_2))$ for each $R_1 \leg R_2$
  in~$\Gamma$, both directions for each $R_1 \eqg R_2$, and both
  directions between $\mathrm{comps}(T)$ and $\mathrm{comps}(R)$ for each
  $\grain{R} = T$. It is built from the types' own components, which are
  fixed by the syntax, and not from their grains, which are fixed only once
  a declaration set is.
\end{enumerate}
\end{definition}

\noindent
A declaration set is the type-level counterpart of a catalogue's
dependencies, and entailment is the counterpart of dependency implication.
On product types the counterpart is exact: whether a surjection exists in
the generated class is a question about which components can be reached
from which, and that is attribute closure.

\begin{lemma}[Grain ordering on product types is attribute closure]
\label{lem:sat-closure}
Let $\Delta$ be a declaration set and let $R_1, R_2$ be product types over
its base types. Then $R_1 \leg R_2$ holds under~$\Delta$ if and only if
every component of $R_2$ lies in the closure of the components of $R_1$
under~$\Delta$. The same holds with either type replaced by any of its
grains under~$\Delta$, a grain having the closure of its type.
\end{lemma}

\noindent
The lemma is the whole content of the transfer below: it identifies holding
under~$\Delta$ with FD satisfaction against~$\Delta$, so that entailment
among grain assertions is FD implication and Armstrong's closure theorem
carries over unchanged.

\smallskip\noindent
Each rule preserves the existence of a structure map between the grains
under any declaration set (Appendix~\ref{sec:proofs}), so the system is
sound; with the Grain Inference
Theorem~(\ref{thm:grain-inference-sufficient}) it infers grain relationships
across transformations, and the next proposition adds completeness in the
form that also makes the ordering decidable.

\begin{proposition}[Soundness and completeness of the grain axioms on product types]
\label{prop:grain-completeness}
Let $\Gamma$ be a finite set of grain assertions on product types and let
$R_1, R_2$ be product types, an $\eqg$-assertion in $\Gamma$ standing for
its two orderings and $\grain{R} = T$ for $T \eqg R$ as in
Definition~\ref{def:grain-holds}. The following are equivalent:
\begin{enumerate}
\item[(a)] $\Gamma \models R_1 \leg R_2$ (Definition~\ref{def:grain-holds});
\item[(b)] every component of $R_2$ lies in the closure of the
  components of $R_1$ under $\Gamma^{\flat}$;
\item[(c)] $R_1 \leg R_2$ is derivable from $\Gamma$ by A1--A9.
\end{enumerate}
In particular A1--A9 are sound, (c)$\Rightarrow$(a), and complete,
(a)$\Rightarrow$(c), for grain ordering on product types.
\end{proposition}

\noindent\textit{Proof sketch.}\quad
(a)$\Leftrightarrow$(b) is Lemma~\ref{lem:sat-closure} applied to
$\Gamma^{\flat}$, the least up to closure of the declaration sets under
which $\Gamma$ holds.
(b)$\Leftrightarrow$(c) is Armstrong's closure theorem, that an FD is
derivable exactly when its right side lies in the closure of its left,
transported through the correspondence of A1--A4 with Armstrong's
rules~\cite{armstrong1974dependency}. Full proof in
Appendix~\ref{sec:proofs}.

\smallskip\noindent\textbf{Example.}\quad
Consider collections $C\;R_1$ and $C\;R_2$ where
$R_1 = \textit{Account} \times \textit{Customer}$ with
$\grain{R_1} = \textit{Account}$, and
$R_2 = \textit{Customer} \times \textit{Address}$ with
$\grain{R_2} = \textit{Customer} \times \textit{Address}$.
After joining on $\textit{Customer}$, we obtain
$R_3 = \textit{Account} \times \textit{Customer} \times \textit{Address}$.
To infer $\grain{R_3}$:
(1)~Since $\grain{R_1} = \textit{Account}$, we have
$\textit{Account} \leg \textit{Customer}$.
(2)~By Augmentation~(A3) with $\textit{Address}$:
$\textit{Account} \times \textit{Address} \leg
\textit{Customer} \times \textit{Address}$.
(3)~By Augmentation~(A3) with $\textit{Account}$:
$\textit{Account} \times \textit{Address} \leg R_3$.
(4)~Since $\textit{Account} \times \textit{Address} \subt R_3$,
$\textit{Account} \times \textit{Address} \leg R_3$, and
$\grain{\textit{Account} \times \textit{Address}} =
\textit{Account} \times \textit{Address}$ (as $\textit{Account}$ is
irreducible), the Grain Inference
Theorem~(\ref{thm:grain-inference-sufficient}) yields
$\grain{R_3} = \textit{Account} \times \textit{Address}$.

\subsection{Grain Incomparability and the Grain Lattice}
\label{subsec:incomparability-lattice}

When neither equality nor ordering holds between two types, they are
grain-incomparable.

\begin{definition}[Grain Incomparability]
\label{def:grain-inc}
Given two data types $R_1$ and $R_2$, we say they are
\emph{grain-incomparable} and write $R_1 \incg R_2$ if and only if
$\neg(R_1 \eqg R_2)$, $\neg(R_1 \leg R_2)$, and $\neg(R_2 \leg R_1)$.
\end{definition}

\noindent
Grain incomparability is irreflexive (no type is incomparable with itself)
and symmetric, but \emph{not} transitive.

\smallskip\noindent\textbf{Example.}\quad
The Section~\ref{sec:intro} fact tables satisfy
$\textit{SalesChannel} \incg \textit{SalesProduct}$: neither grain is a
subset of the other. This incomparability is the source of the fan trap.

\smallskip\noindent\textbf{Lattice structure.}\quad
The ordering itself is the partial order of
Theorem~\ref{thm:grain-partial-order}. On product types over a finite
attribute universe it is a bounded lattice (Figure~\ref{fig:lattice}),
with bounds taken in the
closure lattice of Lemma~\ref{lem:sat-closure}: every pair $R_1, R_2$
has a least upper bound, the product on
$\mathrm{comps}(\grain{R_1})^{+} \cap \mathrm{comps}(\grain{R_2})^{+}$
(finest common coarsening), and a greatest lower bound $R_1 \tun R_2$
(coarsest common refinement), whose closure is
$(\mathrm{comps}(\grain{R_1}) \cup \mathrm{comps}(\grain{R_2}))^{+}$.
When no declared determination relates the two grains, the closures are
the component sets themselves and both bounds are the field-set
operations $\tin$ and $\tun$ of Section~\ref{sec:prelim}; with
$A \to B$ declared, by contrast, $A$ and $B$ are comparable and their
least upper bound is $B$, not the empty intersection. Beyond product
types the bounds are not established here, for a
structural reason: a partial order need not have meets and joins, and
the field-set operations that realize them presuppose a component set
to intersect or identify, which sums and the fixed points do not
supply. Their general interpretation calls for a categorical
reformulation, the least upper bound as a colimit and the greatest
lower bound as a limit in the category of types and structure maps,
and is left as future work.

The direction repays attention. In this order \emph{up} means coarser, so
by the contravariance of Lemma~\ref{thm:grain-subset} it is, in the
determination-free case, the field-set \emph{intersection} that gives the
upper bound and the \emph{union} that gives the lower one: dropping fields coarsens a grain and moves it up,
adding fields refines it and moves it down. The fact tables of
Section~\ref{sec:intro} show both bounds at once. Their grains are
$\textit{CustomerId} \times \textit{ChannelId} \times \textit{Date}$ and
$\textit{CustomerId} \times \textit{ProductId} \times \textit{Date}$, so,
no declaration relating them, their intersection
$\textit{CustomerId} \times \textit{Date}$ is the grain
the drill-across report wants, sitting above both, while their union
$\textit{CustomerId} \times \textit{ChannelId} \times \textit{ProductId}
\times \textit{Date}$ is the grain the join actually produces, sitting
below both. The fan trap of Section~\ref{sec:errors} is exactly the join
landing on the greatest lower bound where the report needed the least
upper one.

When $R_1$ and $R_2$ are grain-comparable the two bounds reduce to the
coarser and the finer of the pair respectively, immediately from the
definitions of least upper and greatest lower bound: in any partial
order, $b \leq a$ makes $a$ the join (lub) of the two and $b$ the
meet (glb).
Nothing about $\leg$ beyond Theorem~\ref{thm:grain-partial-order} is
involved.

\begin{figure}[t]
\centering
\begin{tikzcd}[column sep=small, row sep=large]
  & R_1 \tin R_2 & \\
  R_1 \arrow[ur, "\leg"]
  & & R_2 \arrow[ul, "\leg"'] \\
  & R_1 \tun R_2 \arrow[ul, "\leg"] \arrow[ur, "\leg"'] &
\end{tikzcd}
\caption{The grain lattice for two incomparable types $R_1, R_2$. The
  greatest lower bound (bottom) is $R_1 \tun R_2$ (finest); the least
  upper bound (top) is $R_1 \tin R_2$ (coarsest).}
\label{fig:lattice}
\end{figure}

\noindent
The lattice is bounded: the finest grain (union of all grains) at bottom,
the coarsest (intersection, or unit type) at top. This structure is central
to the join inference rules (Section~\ref{sec:inference}), where
incomparable grain components are resolved through lattice operations.
 
\section{Entity Key, a Grain Within the Grain, and Grain-Determined Semantics}
\label{sec:entity}

Grain fixes \emph{at what level of detail} each row is recorded, but
tables at different levels of detail can still be about the \emph{same
thing}, and no two of them need share a grain. To join, compare, or
reconcile records across them we need a notion of the stable
\emph{identity} the data is about, independent of its level of
detail. That identity is itself a grain, the grain of the entity, and it
is therefore a \emph{grain within the grain}: the \emph{entity key}, an
integration and reconciliation point between types of different detail.
Grain carries a second kind of
information as well: it dictates \emph{how} a table is meant to be
read and written, and using the wrong pattern is a correctness
bug. This
section makes both precise, building on the grain relations of
Section~\ref{sec:relations}: the \emph{entity} and \emph{entity key} as
derived grains (Section~\ref{subsec:entity-def}), and the
\emph{behavioral class} that fixes a table's read/write discipline
(Section~\ref{subsec:behavioral-semantics}), together completing the
denotation triple $(\grain{R},\,\ek{R},\,\BC{R})$.

\subsection{Entity and Entity Projection}
\label{subsec:entity-def}

\begin{definition}[Entity]
\label{def:entity}
Given a data type $R$, an \emph{entity} for $R$ is a type $E$ with a
surjective function $\textit{entity} : R \thra E$ returning, for each
element, the \emph{subject of the information} it conveys: who or what
exists, or the subject (or recipient) of the event that occurred. The
entity is a \emph{semantic} property fixed by domain analysis
(answering ``what is this data about?''), not computed from~$R$'s
structure. It is \emph{unique}: a data type has a single subject of
information, fixed by the declaration and not derived from the grain.
Whereas a grain is unique only \emph{up to isomorphism}, one level of
detail, possibly expressed by several isomorphic types (the
``alternative grains'' of Theorem~\ref{thm:multiple-grains}), an
entity has no alternatives at all.
\end{definition}

\smallskip\noindent
Grain and entity answer different questions and are orthogonal: the
grain asks \emph{at what level of detail} each element is recorded,
the entity asks \emph{who or what} it is about. ``Does the target
hold the same customers as the source?'' is an \emph{entity}
question, unanswerable from the grain alone.

\noindent
Whereas the grain projection extracts the identifying component, the
entity projection extracts the \emph{subject of information}. The
\emph{entity of data} is broader than the entity-relationship
model's notion~\cite{chen1976entity}: it is a universal property of
any data type, including events (whose entity is the event's subject
or recipient, not the event itself), making the entity key a
consistent integration and reconciliation point across grains and
behavioral classes (Remark~\ref{rem:data-integration}).

\smallskip\noindent\textbf{Example.}\quad
For $R = \textit{WeatherObservation}$ with fields $(\textit{StationId},
\textit{Date}, \textit{Temperature}, \textit{Humidity})$, the entity is
$E = \textit{WeatherStation}$, the subject of each
observation, and the grain is
$\grain{R} = \textit{StationId} \times \textit{Date}$.
Each observation describes a particular station (entity) at a particular
level of detail (grain), and every collection derived from these
observations, at whatever grain, shares the station as its entity key.

\subsection{Entity Key as Derived Grain}
\label{subsec:entity-key}

\begin{definition}[Entity Key]
\label{def:entity-key}
Given a data type $R$ with entity $E$ and grain $\grain{R}$, the
\emph{entity key} of~$R$ is $\ek{R} \equiv \grain{E}$, the grain of the
entity type. The entity-key projection admits two equivalent derivations:
\begin{enumerate}
\item \emph{Semantic:}\; $\textit{ek} = \textit{grain}_E \circ
  \textit{entity}$ \,---\, extract the entity, then its grain.
\item \emph{Practical:}\; $\textit{ek} = \textit{proj}_{EK} \circ
  \textit{grain}$ \,---\, extract the grain, then project to the
  entity-key component.
\end{enumerate}
The entity key is always a (type-level) subset of the grain
($\ek{R} \subt \grain{R}$); this is not an assumption but a consequence
of the definitions (Theorem~\ref{thm:ek-hierarchy}).
\end{definition}

\noindent
Recall that the grain is the irreducible (``atomic'') core of a
type, no proper subset suffices to identify its elements. Yet the grain
can be meaningfully ``squeezed'': projection
$g_{ek} : \grain{R} \to \ek{R}$ discards non-entity components and
yields the grain of another type, the entity~$E$. The relationship is captured by the
commutative square in Figure~\ref{fig:ek-square}:

\begin{figure}[t]
\centering
\begin{tikzcd}[column sep=huge, row sep=large]
  \grain{R} \arrow[r, "f_g", "\cong"']
            \arrow[d, "g_{ek}"']
            \arrow[dr, "g_e" description]
  & R \arrow[d, "\textit{entity}"] \\
  \ek{R} = \grain{E} \arrow[r, "f_{gE}"', "\cong"]
  & E
\end{tikzcd}
\caption{Entity-key commutative square. The two paths from $\grain{R}$
  to~$E$ commute: $\textit{entity} \circ f_g = f_{gE} \circ g_{ek}$.
  Squeezing the grain via $g_{ek}$ yields the entity key.}
\label{fig:ek-square}
\end{figure}

\noindent
The one formal result of this section is proved in
Appendix~\ref{sec:proofs}.

\begin{theorem}[EK--Grain Hierarchy]
\label{thm:ek-hierarchy}
For any data type $R$ with entity $E$:
$\ek{R} \subt \grain{R}$.
\end{theorem}

\begin{corollary}[Entity is the coarsest grain]
\label{cor:entity-coarsest}
For any data type $R$ with entity $E$, $R \leg E$; and $R \eqg E$ exactly
when $\grain{R} \cong \ek{R}$.
\end{corollary}

\begin{proof}
$\ek{R} = \grain{E} \subt \grain{R}$ by Theorem~\ref{thm:ek-hierarchy}, which
Lemma~\ref{thm:grain-subset} reads, contravariantly, as $R \leg E$. The second
clause is Definition~\ref{def:grain-eq} with $\grain{E} = \ek{R}$.
\end{proof}

\noindent
In the \textit{WeatherObservation} example:
$\textit{StationId} \subt (\textit{StationId} \times \textit{Date})$.
The entity key identifies the subject (weather station); the grain
adds the temporal component that distinguishes individual observations.

\begin{definition}[Natural grain]
\label{def:natural-grain}
Let $R$ have entity $E$ and entity key $\ek{R}$, and let its domain
analysis declare, besides the entity, a finite set $\mathcal{C}$ of
\emph{cross-element constraints}, each a triple $c = (I_c, \varphi_c, \psi_c)$
with $I_c \ssub R$ and $\varphi_c$, $\psi_c$ $k$-ary predicates on elements
of~$R$, $\varphi_c$ mentioning inhabitants of~$R$ only through $\pi_{I_c}$,
read as $\forall r_1 \ldots r_k : R.\; \varphi_c(\pi_{I_c} r_1, \ldots,
\pi_{I_c} r_k) \Rightarrow \psi_c(r_1, \ldots, r_k)$: $\varphi_c$ says
\emph{which} elements the constraint relates and in what order, $\psi_c$
what must then hold of them. A grain $G$ of~$R$ is \emph{natural}, or a
\emph{business grain}, if (1)~$\ek{R} \ssub G$, and (2)~$I_c \ssub G$ for
every $c \in \mathcal{C}$.
Fields a constraint merely \emph{compares}, those read by $\psi_c$ alone, are
not required, and are not admissible on that ground: they are what grain
drops.
A grain that is not natural is a \emph{surrogate} grain.
\end{definition}

\begin{remark}[Why naturalness is structural, and what it decides]
\label{rem:natural-grain}
Theorem~\ref{thm:ek-hierarchy} gives $\ek{R} \subt \grain{R}$ for
every grain of~$R$, surrogates included, so clause~(1) is not that
theorem restated: it strengthens $\subt$ to $\ssub$, and the
difference is whether the entity key can be read off the type
expression or must be recovered through an isomorphism. Clause~(2)
does the same work for the constraints. Take $\textit{Customer}$ of
Section~\ref{subsec:grain-semantics}, and suppose it also carries a
$\textit{LoadDtm}$ recording when each version was loaded, distinct
per version. Declare that it records versions of a customer's state,
ordered by $\textit{EffectiveFrom}$, consecutive versions differing.
Then $\textit{CustomerId} \times \textit{EffectiveFrom}$ is natural:
the constraint identifies elements by $\textit{CustomerId}$ and
orders them by $\textit{EffectiveFrom}$, and both are components.
$\textit{CustomerId} \times \textit{LoadDtm}$ is a correct grain,
identifying versions equally well and irreducible, but the version
order lives on a field it does not display, so it is surrogate.
($\textit{CustomerId} \times \textit{CreatedOn}$ is no grain of a
versioned table at all: $\textit{CustomerId}$ determines
$\textit{CreatedOn}$, so it cannot tell versions apart.) And
$\textit{Name}$, $\textit{Email}$ and $\textit{Address}$, the fields
the constraint requires to \emph{differ}, belong to neither: they
are exactly what grain drops, and admitting them would cost
irreducibility.

No semantic reading of clause~(2) is available, and the reason is
worth recording. A grain is isomorphic to its type
(Definition~\ref{def:grain}), so $g : R \to G$ is a bijection and
every condition of the form ``the constraint factors through~$g$'' or
``elements agreeing on $G$ agree on the constraint'' holds of every
grain, the surrogate included. A bijection forgets nothing, so there
is no semantic content to test, and the distinction lives entirely in
which components are displayed. This is
Definition~\ref{def:grain} a second time, where irreducibility likewise
had to be read structurally rather than semantically, and for the same
reason: under the semantic reading the clause goes vacuous, every type
isomorphic to $R$ qualifying.

The order of dependence is therefore declaration first: domain
analysis fixes the entity and the constraints, those fix which grains
are natural, and the behavioral class is read off a natural one
(Definition~\ref{def:behavioral-class}). Nothing above mentions a
class.
\end{remark}

\begin{remark}[Data Integration and Reconciliation at the Entity Level]
\label{rem:data-integration}
When $R_1$ and $R_2$ share entity $E$, both
$\grain{R_1} \leg \ek{R_1}$ and $\grain{R_2} \leg \ek{R_2}$, and
$\ek{R_1} = \ek{R_2} = \grain{E}$: their grains converge on the shared
entity key, enabling integration at the entity level even when grains differ.
For example, a customer entity table ($\grain{} = \textit{CustomerId}$),
a customer version history
($\grain{} = \textit{CustomerId} \times \textit{EffectiveFrom}$), and
a customer event log
($\grain{} = \textit{CustomerId} \times \textit{EventDtm}$) all share
$\ek{} = \textit{CustomerId}$, three different grains and behavioral
classes, yet a single integration point, and the very customer readings
introduced in Section~\ref{sec:foundations}.
Beyond integration, the entity key enables \emph{reconciliation} across
collections of different grains and behavioral classes: two collections
$c_1$ and $c_2$ satisfy \emph{entity-key equality}
($c_1 =_{ek} c_2$) when they track the same set of entities, regardless
of their differing grains. This is strictly more informative than
cardinality comparison ($|c_1| = |c_2|$), which reveals nothing about
\emph{which} entities are present or missing.
\end{remark}

\begin{definition}[Entity preservation]
\label{def:entity-preservation}
Write $\textit{ekset}(c) = \{\,\textit{ek}\,e \mid e \in c\,\}$ for the
set of entity-key values of a collection~$c$, where $\textit{ek}$ is the
entity-key projection (Definition~\ref{def:entity-key}). Entity-key
equality and inclusion of Remark~\ref{rem:data-integration} are then
$c_1 =_{ek} c_2 \iff \textit{ekset}(c_1) = \textit{ekset}(c_2)$ and
$c_1 \subseteq_{ek} c_2 \iff \textit{ekset}(c_1) \subseteq
\textit{ekset}(c_2)$ (up to the entity-key isomorphism). A data
collection transformation $h : C\,R_1 \to C\,R_2$ with
$\ek{R_2} = \ek{R_1}$ is \emph{entity-preserving} if $h\,c =_{ek} c$ for
every input~$c$, and \emph{entity-including} if $h\,c \subseteq_{ek} c$.
\end{definition}

\noindent
Entity preservation is a design-time correctness invariant distinct from
\emph{element} preservation (row-for-row equality, $h\,c =_c c$), and
decidable from the grain denotation. For instance, a point-in-time read
of an SCD2 customer dimension
($\grain{} = \textit{CustomerId} \times \textit{EffectiveFrom}$) that
returns the latest version of each customer as current state
($\grain{} = \textit{CustomerId}$, dropping $\textit{EffectiveFrom}$ and
so yielding a new type) is entity-preserving, since the output holds exactly the customers of the
input, one row each, yet it is not element-preserving,
since the versioned input has more rows than the one-per-customer output.
The full family of collection relations, grading correspondence from
elements ($=_c$) through grain ($=_g$) to entity ($=_{ek}$), is developed
in the companion~\cite{pddpaper}.

\subsection{Grain Determines Behavioral Semantics}
\label{subsec:behavioral-semantics}

The \emph{behavioral semantics} (or \emph{behavior}) of a data
type~$R$ are the structural characteristics that determine the
proper access patterns for reading and writing collections of
elements of~$R$, independent of their domain meaning. A
multi-version table is read and written with the same pattern
whether the versions describe customers, contracts, or any other
entity.

\begin{definition}[Behavioral class]
\label{def:behavioral-class}
Let $R$ have entity key $\ek{R}$ and a declared natural grain
(Definition~\ref{def:natural-grain}). A \emph{behavioral class}
is a type-level predicate on~$R$ given by
\begin{enumerate}
\item a \emph{grain condition}: an equation fixing $\grain{R}$ in
  terms of $\ek{R}$ and, where the class has one, a
  \emph{distinguishing component};
\item an \emph{order} on that component, where there is one; and
\item a \emph{universally quantified invariant}, where noted.
\end{enumerate}
The grain condition is an equation on the \emph{natural} grain:
$\grain{R}$ denotes the declared natural representative, and
$\ek{R}$ and~$D$ are its components. Neither need be a base
type: $\ek{R}$ is the grain of the entity type, so it is any
type expression Section~\ref{sec:adt} admits, a product where
the entity has a compound key, a sum where it is a variant; and
$D$ is any type carrying the stated order. The class constrains
those components and nothing else: $R$ may carry further
fields, ordering ones included, on which it imposes no
condition.

Grain conditions are read \emph{structurally}, slot by slot, as
in Definition~\ref{def:grain}. Read up to isomorphism they
collapse on infinite types, where $\ek{R} \times D \cong D$
whenever both are countably infinite, so that \textbf{IsEvent}
and \textbf{IsSeqEvent} would hold at once
(Remark~\ref{rem:structural-irreducibility}).

The $\textit{Customer}$ of Remark~\ref{rem:natural-grain} is the
three clauses at once: its grain condition is
$\grain{R} = \ek{R} \times \textit{EffectiveFrom}$, its order is
the version order on that field, and its invariant is that
consecutive versions differ. That is the shape
Definition~\ref{def:core-class} calls \textbf{IsMultiVersion};
the $\textit{LoadDtm}$ beside it is a field no clause mentions.
\end{definition}

\begin{definition}[Core classes]
\label{def:core-class}
A class is \emph{core} when its distinguishing component is a
single \emph{ordering dimension}~$D$. For one such dimension the
three ways it can relate to the entity key are exhaustive and
exclusive, and give the seven core classes:
\begin{itemize}
\item \emph{$D$ absent}, $\grain{R} = \ek{R}$: \textbf{IsEntity}.
\item \emph{$D$ appended}, $\grain{R} = \ek{R} \times D$ with $D$
  totally ordered~($\leq$): \textbf{IsOrdered}, whose order is a
  rank or priority; \textbf{IsEvent}, refining it with a causal
  order; and \textbf{IsMultiVersion}, refining \textbf{IsEvent}
  with consecutive same-entity versions differing in payload
  (Remark~\ref{rem:bc-mv-constraint}).
\item \emph{$D$ is the key}, $\grain{R} = \ek{R} = D$ with $D$
  strictly totally ordered~($<$): \textbf{IsSeqOrdered};
  \textbf{IsSeqEvent}, refining it with a causal order; and
  \textbf{IsSnapshot}, refining \textbf{IsSeqEvent} with complete
  entity coverage (Remark~\ref{rem:bc-snap-constraint}).
\end{itemize}
\end{definition}

\begin{remark}[The classes refine; $\BC{\cdot}$ is a function]
\label{rem:bc-refinement}
The core classes refine one another rather than partition, with
$\Rightarrow$ read as implication, so that the leftmost class
in each chain is the most specific:
\[
  \textbf{IsMultiVersion} \Rightarrow \textbf{IsEvent} \Rightarrow
  \textbf{IsOrdered},
  \qquad
  \textbf{IsSnapshot} \Rightarrow \textbf{IsSeqEvent} \Rightarrow
  \textbf{IsSeqOrdered},
\]
each step to the right dropping an invariant or weakening a
causal order to a general one, with \textbf{IsEntity} standing
alone. $\BC{\cdot}$
returns the \emph{most specific} core class a type satisfies.
Given a declared natural grain, order and invariants that class
is determined, so $\BC{\cdot}$ is a \emph{function}, and into the
classes rather than into types, unlike $\grain{\cdot}$ and
$\ek{\cdot}$, which name a type only up to isomorphism.

The catalogue is not thereby exhaustive: a type may carry more
than one dimension, and other relations among elements yield
further classes. Two appear in Table~\ref{tab:bc-formal} already
and are extensions rather than core, their distinguishing
component being no ordering dimension: \textbf{IsNMRelation},
whose grain condition relates $\grain{R}$ to the grains of the
two types it associates, and \textbf{HasParent}, which has no
grain condition at all and is fixed by a declared relation. Both
are orthogonal rather than alternative, each saying how $R$
relates to another type rather than how its own elements relate,
so a type carries either alongside its core class: a Data~Vault
Transactional Link is \textbf{IsNMRelation} with
\textbf{IsEvent}.
\end{remark}

\begin{remark}[The ordering dimension and its order]
\label{rem:bc-order}
The ordering dimension ($\textit{EventDtm}$, $\textit{FromDtm}$,
$\textit{SnapshotDtm}$, $\textit{OrdType}$) is a \emph{generic
type parameter} carrying the stated order, not a fixed type. Its
order is \emph{total on each entity-key cohort}, which is what
makes ``the latest version'' well defined and so what fixes the
read and write patterns (Remark~\ref{rem:bc-mv-constraint}).
Strictness is not a further choice but follows from the shape:
where the dimension \emph{is} the key its values are unique and
no ties are possible.

A \emph{causal} order is not a different order but a condition on
this one. Happened-before is in general a \emph{partial} order,
concurrent elements being incomparable, and requiring $r_1$
happened-before $r_2$ to imply $\textit{dtm}\,r_1 <
\textit{dtm}\,r_2$ says the dimension's order \emph{linearises}
it. That is the whole difference from the general classes, which
read the order as a rank with no causal reading: \textbf{IsEvent}
is exactly \textbf{IsOrdered} with a causal order. All classes
are mechanized as Agda record types
(Section~\ref{sec:verification}).
\end{remark}

\begin{table}[t]
\centering
\small
\caption{Formal definition of the behavioral classes. Each is a type
  class (constraint) given by a grain condition, its order, and any
  integrity constraint. The order is \emph{causal} (a happened-before
  relation) for the event classes and \emph{general} for the ordering
  classes; \textbf{HasParent} carries a relation rather than a grain
  condition. $P_1, P_2$ denote the two parties of an N:M relation. Rows are
  grouped by family (the dimension absent, appended, or the key
  itself, then the two relational extensions); within a family
  each class implies the one above it.}
\label{tab:bc-formal}
\begin{tabular}{@{}llll@{}}
\toprule
\textbf{Class} & \textbf{Grain Condition} & \textbf{Order} & \textbf{Additional Constraint} \\
\midrule
IsEntity
  & $\grain{R} = \ek{R}$
  & --- & --- \\
\addlinespace
IsOrdered
  & $\grain{R} = \ek{R} \times \textit{OrdType}$
  & general,~$\leq$
  & ties allowed \\
IsEvent
  & $\grain{R} = \ek{R} \times \textit{EventDtm}$
  & causal,~$\leq$
  & --- \\
IsMultiVersion
  & $\grain{R} = \ek{R} \times \textit{FromDtm}$
  & causal,~$\leq$
  & consecutive versions \\
  & & & differ in payload (Rem.~\ref{rem:bc-mv-constraint}) \\
\addlinespace
IsSeqOrdered
  & $\grain{R} = \ek{R} = \textit{OrdType}$
  & general,~$<$
  & --- \\
IsSeqEvent
  & $\grain{R} = \ek{R} = \textit{EventDtm}$
  & causal,~$<$
  & --- \\
IsSnapshot
  & $\grain{R} = \ek{R} = \textit{SnapshotDtm}$
  & causal,~$<$
  & complete entity coverage \\
  & & & (Rem.~\ref{rem:bc-snap-constraint}) \\
\addlinespace
IsNMRelation
  & $\grain{R} = \grain{P_1} \times \grain{P_2}$
  & ---
  & N:M edge ($\grain{R} = \ek{R}$) \\
HasParent
  & \emph{(none)}
  & ---
  & $\textit{parent} : R \to \mathrm{Maybe}\,\grain{P}$ (1:N) \\
\bottomrule
\end{tabular}
\end{table}

\noindent
\textbf{IsMultiVersion} is a class, for instance, because it is
exactly the shape of a versioned entity, an SCD~Type~2 dimension
or Data~Vault Satellite (Table~\ref{tab:bc-paradigms}), whose
read and write patterns are fixed by that grain condition.

Order among elements, however, is only one relation a type can
carry. Classifying by the \emph{relation among elements} (identity,
time, order, hierarchy, or association) yields further classes on
the same footing: a \emph{one-to-many} hierarchy
(\textbf{HasParent}), for parent-child structures, and a
\emph{many-to-many} association (\textbf{IsNMRelation}), for bridge
tables and graphs (grain conditions in Table~\ref{tab:bc-formal},
realizations in Table~\ref{tab:bc-paradigms}). The catalogue is thus
extensible along each axis; the seven core classes above are its
rigorously developed core.

\begin{remark}[IsMultiVersion: Consecutive Versions Differ in Payload]
\label{rem:bc-mv-constraint}
For $R$ with $\grain{R} = \ek{R} \times \textit{FromDtm}$,
IsMultiVersion additionally requires
\[
\forall\, r_1, r_2 : R.\; \text{ek}(r_1) = \text{ek}(r_2)
\wedge \text{adj}(r_1, r_2)
\;\Rightarrow\; \text{payload}(r_1) \neq \text{payload}(r_2),
\]
where $\text{ek}(r)$ and $\text{payload}(r)$ are the entity-key and
payload projections from~$R$, and $\text{adj}(r_1, r_2)$ holds when
$r_2$ immediately follows $r_1$ in the \textit{FromDtm} total order
on $r_1$'s entity-key cohort.
\end{remark}

\begin{remark}[IsSnapshot: Complete Entity Coverage]
\label{rem:bc-snap-constraint}
For $R$ with $\grain{R} = \ek{R} = \textit{SnapshotDtm}$, each
element of~$R$ records a snapshot of the state of a population of
entities at one timestamp. The entity of~$R$ is therefore itself a
collection, $\textit{Coll}\,E$, where $E$ is the underlying
entity-state type; combined with
$\grain{R} = \ek{R} = \textit{SnapshotDtm}$, the entity-grain
function of Section~\ref{sec:entity} specializes to
\[
f_{gE} : \textit{SnapshotDtm} \to \textit{Coll}\,E,
\]
mapping a snapshot timestamp to the complete entity population at
that time. IsSnapshot additionally requires each snapshot to
record this complete population:
\[
\forall\, r : R.\; f_{gE}(\textit{snapDtm}(r))
\;=_{coll}\; \textit{snapshot}(r),
\]
where $\textit{snapDtm} : R \to \textit{SnapshotDtm}$ extracts the
snapshot timestamp, $\textit{snapshot} : R \to \textit{Coll}\,E$
extracts the population recorded in~$r$, and $=_{coll}$ denotes
collection equality.
\end{remark}

\begin{remark}[The class invariants as cross-element constraints]
\label{rem:bc-as-constraints}
Both invariants are cross-element constraints in the sense of
Definition~\ref{def:natural-grain}. For \textbf{IsMultiVersion},
$I = \ek{R} \times \textit{FromDtm}$, $\varphi(r_1, r_2)$ is
$\text{ek}(r_1) = \text{ek}(r_2) \wedge \text{adj}(r_1, r_2)$, where
$\text{adj}$ binds a third inhabitant ($\neg\exists\, r_3 : R$ between $r_1$
and $r_2$ in the cohort's \textit{FromDtm} order) and reads it, like $r_1$
and $r_2$, only through~$I$; and $\psi(r_1, r_2)$ is
$\text{payload}(r_1) \neq \text{payload}(r_2)$. For \textbf{IsSnapshot},
$k = 1$ with a vacuous antecedent: $I = \emptyset$, $\varphi \equiv \top$,
and $\psi(r)$ is the coverage equation; $\textit{SnapshotDtm}$ enters the
grain through clause~(1), since $\ek{R} = \textit{SnapshotDtm}$, not through
the constraint. In both cases clause~(2) asks for exactly the fields the
antecedent reads, and for none of those the consequent compares.
\end{remark}

\smallskip\noindent\textbf{Example.}\quad
The $\textit{Customer}$ example from Section~\ref{subsec:grain-semantics}
illustrates: with $\ek{\textit{Customer}} = \textit{CustomerId}$ fixed,
the four grain declarations yield \textbf{IsEntity} (entity),
\textbf{IsMultiVersion} (versioned), \textbf{IsMultiVersion} with
$\ek{} = \textit{CustomerId} \times \textit{Address}$ (versioned
customer-address pair), and \textbf{IsEvent} (event) semantics,
respectively.

\begin{remark}[Behavioral Class Propagation]
\label{rem:bc-propagation}
Since grain inference (Section~\ref{sec:inference}) computes the grain
of every RA result, the behavioral class propagates compositionally:
the result grain's structure determines the result's BC. In a pipeline,
the target type's grain and BC are declared; CalcG
(Section~\ref{sec:calcg}) verifies that the computed grain of the
result matches the target grain, which guarantees that the target's BC
is satisfied. For example, in a 1-to-many join
(Proposition~\ref{prop:join-special-cases}, Case~2),
$\textit{Orders}$ (\textbf{IsEvent}) $\bowtie$
$\textit{Customer}$ (\textbf{IsEntity}) on $\textit{CustomerId}$
gives $\grain{Res} = \grain{\textit{Orders}}$, hence
\textbf{IsEvent}.
\end{remark}

\begin{remark}[Paradigm Independence]
\label{rem:paradigm-independence}
Behavioral classes apply independently of domain meaning \emph{and}
data modeling paradigm. Table~\ref{tab:bc-paradigms} maps the classes to well-known table types from Kimball's dimensional
modeling~\cite{kimball2013dw} and
Data~Vault~\cite{linstedt2015dv}: a Kimball SCD~Type~2 dimension and
a Data~Vault Satellite, for example, are both instances of
\textbf{IsMultiVersion}, same grain condition, same read/write
patterns, letting data engineers reason about a small number of
grain-determined patterns rather than the specifics of each paradigm.
\end{remark}

\begin{table}[t]
\centering
\small
\caption{Behavioral classes mapped to data modeling paradigms
  (representative examples).}
\label{tab:bc-paradigms}
\begin{tabular}{@{}lll@{}}
\toprule
\textbf{Class} & \textbf{Grain / relation} & \textbf{Representative Table Types} \\
\midrule
IsEntity & $\ek{R}$
  & SCD Type~1 (Kimball), Hub (Data~Vault) \\
IsEvent & $\ek{R} \times \textit{EventDtm}$
  & Transaction Fact (Kimball), Time-Series \\
IsMultiVersion & $\ek{R} \times \textit{FromDtm}$
  & SCD Type~2 (Kimball), Satellite (Data~Vault) \\
IsSeqEvent & $\ek{R} = \textit{EventDtm}$
  & CDC Stream, Audit Log \\
IsSnapshot & $\ek{R} = \textit{SnapshotDtm}$
  & Periodic Snapshot Fact (Kimball) \\
\addlinespace
IsOrdered & $\ek{R} \times \textit{OrdType}$
  & Ranking, Leaderboard, Queue position \\
IsSeqOrdered & $\ek{R} = \textit{OrdType}$
  & Workflow Steps, Project Phases \\
HasParent & $\textit{parent} \to \grain{P}$ (1:N)
  & FK / self-referential hierarchy (org chart, category tree) \\
IsNMRelation & $\grain{P_1} \times \grain{P_2}$
  & Bridge/junction table, property graph \\
\bottomrule
\end{tabular}
\end{table}

\noindent
This grain-determines-behavior link is the essential point of behavioral
classes. Each class fixes not only a grain condition but a
\emph{canonical} way to read and write its collections: an
\textbf{IsEntity} table is read as current state and written by upsert;
an \textbf{IsEvent} table is append-only and read over time windows; an
\textbf{IsMultiVersion} table is written one version per change and
\emph{read only through a point-in-time query}. Declared as a type-class
constraint on the element type, these access patterns are enforced by
the type checker: reading a multi-version collection without a
point-in-time, or a non-upsert write to an entity table, is a compile-time \emph{type error}, unless the type is explicitly
re-declared (e.g.\ viewing the multi-version data as an
\textbf{IsEntity} current state). Purely structural constraints such as
functional dependencies cannot even distinguish an entity from a
versioned entity (Section~\ref{subsec:rw-fd}), let alone prescribe its
access. Two of the three clauses of
Definition~\ref{def:behavioral-class} are \emph{schema-decidable} and
one is not. A grain condition is an equation between type expressions,
so CalcG checks it from the declarations alone
(Section~\ref{sec:calcg}); so is the order, being a declaration on a
component. An invariant is a sentence quantified over elements,
\textbf{IsMultiVersion}'s consecutive versions differing
(Remark~\ref{rem:bc-mv-constraint}) or \textbf{IsSnapshot}'s coverage
(Remark~\ref{rem:bc-snap-constraint}), and reduces to no equation
between types, so nothing here discharges it: it is declared, and
enforced by whatever loads the collection. This is not a
type-level against instance-level split, since a declared key equally
needs enforcing or the table acquires duplicates; it is the line
between what this paper's machinery checks and what it assumes, and a
companion paper closes it by giving each operation a
$(\textsf{Pre},\textsf{Intra},\textsf{Post})$ contract that carries
such invariants along a pipeline~\cite{pddpaper}. The catalogue is deliberately \emph{representative and
extensible}, not exhaustive: its value is a single, paradigm-independent
view, one set of grain-determined read/write patterns spanning identity, time, order, hierarchy and association across
Kimball, Data~Vault, graph models and beyond
(Table~\ref{tab:bc-paradigms}), baked into the type system, rather than
the number of classes it lists. The classes it does list are the
ones we have found to recur in warehouse and lake practice, each
abstracting a table type of an established paradigm
(Table~\ref{tab:bc-paradigms}).

\begin{definition}[The Type-Level Denotation of Data]
\label{rem:type-level-denotation}
Together, grain, entity key, and behavioral class yield a triple
$(\grain{R},\, \ek{R},\, \BC{R})$ that we call the
\emph{type-level denotation} of $R$, capturing the essential
semantic content of any algebraic data type: \emph{what} each
record uniquely represents ($\grain{R}$, its level of detail),
\emph{which entity} it is about ($\ek{R}$, invariant to that detail),
and \emph{how} it must be read and written ($\BC{R}$, the core class of
Definition~\ref{def:core-class} and the canonical access
pattern its declared semantics fix,
Section~\ref{subsec:grain-semantics}; a type also carrying
\textbf{HasParent} or \textbf{IsNMRelation} records those
separately). The remaining sections show that this denotation is
\emph{computable} (Sections~\ref{sec:inference},
\ref{sec:grain-factorization}) and \emph{verifiable at design time}
via CalcG (Section~\ref{sec:calcg}).
\end{definition}
 
\section{Grain for Arbitrary Algebraic Data Types}
\label{sec:adt}

Three data types, all ordinary data-engineering inputs. An order arrives
as a nested document, its line items an array inside it. A change log
arrives as one unbounded sequence per entity, a Kafka topic partition. A
sensor feed arrives as a stream with no last element:
\begin{align*}
  \textit{Order} &= \textit{OrderId} \times \textit{CustomerId}
                    \times \textit{List}\,\textit{OrderLine}, \\
  \textit{ChangeLog} &= \textit{Stream}\,\textit{CDC-Event}, \qquad
  \textit{SensorFeed} = \textit{Stream}\,\textit{WeatherObservation}.
\end{align*}
Before writing a transformation over any of them an engineer must settle
what a single element represents and what irreducibly identifies it. Get it
wrong and every count downstream is wrong.

A candidate key answers for one of the three, and only at its outermost
level. It is a \emph{minimal subset of a fixed, finite set of attributes},
and that presupposes a product: a superkey is a subset, minimality is over
subsets, and both need a fixed set to range over. $\textit{Order}$ is such
a product, so once $\textit{OrderId}$ is declared $\textit{Order}$ has one,
and no normal form stands in the way: current SQL is content with an
array-valued column and a primary key beside it. What a key cannot do is
look \emph{inside}. A list type is no product, its inhabitants being
sequences of different lengths, so there is no attribute set over which to
be minimal, and ``what identifies an order line, or a list of them?''
cannot be posed at all. A stream is not even a relation: its inhabitants
are unbounded, and there is no product anywhere to start from. The question is not hard for these types, it is \emph{unaskable} in the
relational vocabulary, and unaskable exactly where a pipeline needs it: at
every level of a nested type, since that is where a pipeline unnests,
explodes and aggregates. Making it askable there, and answerable from the
declarations a catalogue already holds, is what Theorem~\ref{thm:grain-adt}
establishes.

Definition~\ref{def:grain} is already free of that limitation: it fixes
the grain of a type by isomorphism and irreducibility alone, with no
appeal to the type's shape, so it applies to all three as written. What it
does not yet give is a way to \emph{find} the grain, since it says what a
grain is, not how to obtain one from a type. This section supplies that,
and it is where grain parts company with the candidate key for good. We
define a grain \emph{type-expression operator} $\gop{\cdot}$ by
recursion on type structure (Section~\ref{subsec:compositional-grain}), prove that what it
returns is always a grain in the sense of Definition~\ref{def:grain}
(Section~\ref{subsec:grain-adt-theorem}), and then come back and answer
the three types above
(Section~\ref{subsec:windowed-collections}).

Two symbols, and the difference between them is the point of the
section. $\grain{R}$ is the notation of
Section~\ref{sec:foundations}: Definition~\ref{def:grain} is a
\emph{relation}, $\grain{R}$ names any member of
$\textit{Grains}(R)$, its fiber over~$R$, and
Lemma~\ref{lem:grain-existence} with
Theorem~\ref{thm:multiple-grains} makes that fiber non-empty and contained
in a single isomorphism class. What those results withhold is a
\emph{construction}, a way to compute a grain of~$R$ from~$R$;
supplying one is this section's business.

That relation is not inert, and it is worth being exact about what
Definition~\ref{def:compositional-grain} adds to it. On a product type the
grain is \emph{already} computed from the declared determinations, by
Theorem~\ref{thm:grain-product} with the Grain Inference theorem
(Theorem~\ref{thm:grain-inference-sufficient}) and the Armstrong-style
axioms of Section~\ref{subsec:armstrong}: that computation is attribute
closure, the passage from a relation to a candidate key. What
$\gop{\cdot}$ adds is therefore not computation where there was none, but
\emph{reach}. Closure runs over the components of a product and stops
there, and at the inner levels of a nested or recursive type there are no
components for it to run over, which is the obstruction this section
opened with. $\gop{\cdot}$ is the recursion on type structure that carries
that same closure past products: its product clause calls the closure and
adds nothing to it, and its clauses for $+$, $\mu$ and $\nu$ are where
grain goes beyond what a key can express. The same reading applies:
where a product admits several candidate keys the recursion may
return any of them. Theorem~\ref{thm:grain-adt} is its soundness: for a
\emph{closed} type expression~$T$, the type $\gop{T}$ is a grain of~$T$ in
the sense of Definition~\ref{def:grain}, so everything
Section~\ref{sec:foundations} proves about $\grain{T}$ holds of it
unchanged.

Keeping the two apart matters because the two kinds of equation have
different standing. A clause of Definition~\ref{def:compositional-grain}
\emph{says what $\gop{\cdot}$ is}: it is a decision about a symbol we are
introducing, so it holds by definition and there is nothing in it to
verify. An equation about $\grain{\cdot}$ describes an object
Section~\ref{sec:foundations} has already defined independently, so it can
be false and has to be proved. Thus
$\gop{\mu X.\,F} = \mu X.\,\gop{F}$ holds by definition, whereas
$\grain{\mu X.\,F} = \mu X.\,\grain{F}$ is a theorem, and not an obvious
one, since it is exactly what Theorem~\ref{thm:grain-adt} has to
establish. $\gop{X} = X$ holds by definition; $\grain{X} = X$ is not even
well-formed, Definition~\ref{def:grain} speaking of types where $X$ is a
variable. And where the naive equation is \emph{false}, the definition has
to say so, which is why the product clause below is not an equation
between $\gop{\cdot}$-terms at all but a call to the closure:
$\grain{A_1 \times \cdots \times A_n} = \grain{A_1} \times \cdots \times
\grain{A_n}$ fails when a determination crosses the components
(Theorem~\ref{thm:grain-product}).

Recall from Section~\ref{sec:prelim} that an algebraic data type is built
from base types by the grammar
\[
  T \;::=\; B \;\mid\; X \;\mid\; T \times T \;\mid\; T + T
            \;\mid\; \mu X.\,T \;\mid\; \nu X.\,T,
\]
with $B$ a base type and $X$ a type variable: products and sums give
records and variants, $\mu$ the inductive types such as the finite list
$\textit{List}\,R = \mu X.\,\mathbf{1} + R \times X$, where $X$ stands for
the tail, and $\nu$ the coinductive ones such as the stream
$\textit{Stream}\,R = \nu X.\,R \times X$. In every such definition~$X$
occurs only where the type \emph{stores} data, as a field of a record or a
summand of a variant, and never as the argument of a function. That
property has a name, \emph{strict positivity}, and without it a recursive
definition need not describe any type at all: a $T$ satisfying
$T \cong (T \to \textit{Bool})$ would stand in bijection with its own
powerset, which Cantor's theorem forbids. The grammar above has no
function types, so every body is strictly positive automatically and the
fixed points exist. A binder's body~$F$ therefore denotes a
\emph{polynomial functor}, a type expression with a hole that acts on
functions as well as on types~\cite{birddemoor1997};
Appendix~\ref{app:primer} collects this and the rest of the
type-theoretic vocabulary in database terms.
The two central results of this section, Theorem~\ref{thm:grain-adt} and
the lemma it rests on, are proved in place: they are what the extension to
recursive types amounts to, and a reader should not have to take them on
trust.

\subsection{The Order Document}
\label{subsec:adt-running-example}

Of the three types this section opened with, the order document is the one
that carries the development from here on. It is a nested record holding a
list of line items, each line referring to a catalogue product:
\begin{align*}
  \textit{Product}   &= \textit{ProductId} \times \textit{ProductName}
                        \times \textit{Price}, \\
  \textit{OrderLine} &= \textit{Product} \times \textit{Qty}
                        \times \textit{Discount}, \\
  \textit{Order}     &= \textit{OrderId} \times \textit{CustomerId}
                        \times \textit{List}\,\textit{OrderLine}.
\end{align*}
Two determinations are declared for these types, both of the kind a
catalogue records: $\textit{ProductId}$ determines the rest of a
$\textit{Product}$, and $\textit{OrderId}$ determines the rest of an
$\textit{Order}$, its customer and its lines alike. Nothing here is
exotic; it is an ordinary document, of the kind a JSON column or an
array-valued column holds.

The grains of these types are the following, stated here so that the rest
of the section has a target to hit:
\[
  \grain{\textit{Product}} = \textit{ProductId}, \quad
  \grain{\textit{OrderLine}} = \textit{ProductId} \times \textit{Qty}
                               \times \textit{Discount}, \quad
  \grain{\textit{Order}} = \textit{OrderId},
\]
with $\grain{\textit{List}\,\textit{OrderLine}} =
\textit{List}\,\grain{\textit{OrderLine}}$ in between. Read aloud: a
product is identified by its id, a line by what was bought, how much of
it and at what discount, and a whole order document by its order id.

These are claims in the sense of Definition~\ref{def:grain}, and the first
three can be checked against it on the spot. $\textit{ProductId}$ is
isomorphic to $\textit{Product}$, by the declared determination, and is a
base type, so nothing smaller is; likewise $\textit{OrderId}$ for a whole
order. $\textit{ProductId} \times \textit{Qty} \times \textit{Discount}$
is isomorphic to $\textit{OrderLine}$ and irreducible, no component being
determined by the others. The fourth is the interesting one, since
Definition~\ref{def:grain} gives no way to reach inside a
$\textit{List}$, and it is what Theorem~\ref{thm:grain-adt} will
establish. What the rest of this section supplies is not these four
answers but a procedure that produces them, for these types and for every
other.

\subsection{The Grain Type-Expression Operator}
\label{subsec:compositional-grain}

Structural containment (Definition~\ref{def:structural-subtype}) reaches the
rest of the grammar componentwise, which is what lets irreducibility be read
off a sum or a fixed point at all: $A \ssub B$ on a sum when each summand
of~$A$ is a structural subtype of the corresponding summand of~$B$, and on
$\mu X.\,F$ or $\nu X.\,F$ when it holds of the bodies with the recursion
variable held fixed. As on products, $A \spsub B$ when moreover the
containment is somewhere strict.

Every clause below but one is pure structural recursion, reducing each
slot of a type independently of its neighbours. Declared determinations,
by contrast, relate slots to one another, and they enter at exactly one
place: the clause for products.

\begin{definition}[Grain type-expression operator]
\label{def:compositional-grain}
The grain type-expression operator extends from base types to every
algebraic data type by recursion on type structure:
\begin{align*}
  \gop{B} &= \grain{B}, &
  \gop{X} &= X, \\
  \gop{A_1 \times \cdots \times A_n}
    &= \textstyle\prod_{i \in S} \gop{A_i}, &
  \gop{A_1 + \cdots + A_n} &= \gop{A_1} + \cdots + \gop{A_n}, \\
  \gop{\mu X.\,F} &= \mu X.\,\gop{F}, &
  \gop{\nu X.\,F} &= \nu X.\,\gop{F},
\end{align*}
where $\gop{F}$ applies the operator to the parameter
sub-expressions of~$F$ and leaves the bound recursion variable fixed. The
base clause defers to Definition~\ref{def:grain} rather than asserting an
answer: an atomic base type is its own grain, since it has no proper
structural subtype (Definition~\ref{def:structural-subtype}), so there
is nothing a determination could reduce it to. In the product clause
$S \subseteq \{1, \dots, n\}$ is a minimal set of components whose closure
under the declared determinations is $\{1, \dots, n\}$, obtained by the
Grain Inference theorem
(Theorem~\ref{thm:grain-inference-sufficient}) and the axioms of
Section~\ref{subsec:armstrong}; where several such~$S$ exist the clause
may return any of them, as several candidate keys may. That clause is
Theorem~\ref{thm:grain-product} placed where the recursion can reach it. The
operator is total on every algebraic data type.
\end{definition}

\noindent
Read the clauses one at a time on the running example. At a \emph{base
type} there is nothing to do, $\gop{\textit{Qty}} = \textit{Qty}$. A
\emph{product} reduces each slot first, then drops any slot the others
determine. The inner product goes first,
$\gop{\textit{Product}} = \textit{ProductId}$, because
$\textit{ProductId}$ was declared to determine the rest of a
$\textit{Product}$, and the outer one follows:
\[
  \gop{\textit{OrderLine}}
    = \textit{ProductId} \times \textit{Qty} \times \textit{Discount}.
\]
Nothing was dropped from the three surviving slots, because no
determination was declared \emph{between} them: a line
is identified by what was bought, how much of it and at what discount, and
dropping any one of the three loses a line. A \emph{sum} reduces
summand by summand: for a payment method
$\textit{Payment} = \textit{Card} + \textit{Transfer}$ with
$\gop{\textit{Card}} = \textit{AuthCode}$ and
$\gop{\textit{Transfer}} = \textit{TransferRef}$, we get
$\gop{\textit{Payment}} = \textit{AuthCode} + \textit{TransferRef}$:
identifying a payment means saying which kind of payment it is and giving
that kind's identifier. A \emph{$\mu$-type} reduces its payload and keeps
its shape:
\[
  \gop{\textit{List}\,\textit{OrderLine}}
    = \textit{List}\,(\textit{ProductId} \times \textit{Qty}
                      \times \textit{Discount}),
\]
a list of the same length and in the same order as the original, each
element stripped of the product name and price that its
$\textit{ProductId}$ recovers. A \emph{$\nu$-type} behaves identically:
for a stream of order events
$\textit{OrderEvent} = \textit{OrderId} \times \textit{EventDtm} \times
\textit{EventType}$, in which $\textit{OrderId} \times \textit{EventDtm}$
identifies an event, $\gop{\textit{Stream}\,\textit{OrderEvent}}
= \textit{Stream}\,(\textit{OrderId} \times \textit{EventDtm})$, an
infinite sequence again.

\begin{remark}[Where declared determinations enter]
\label{rem:compositional-independence}
The structural clauses
reduce \emph{within} a slot and are pure syntax, computable from the type
expression alone. Products are the one place declared determinations act,
and they act \emph{across} slots, dropping a slot exactly when the others
determine it (Theorem~\ref{thm:grain-product}). Sums impose no such
condition, their summands being disjoint alternatives.

Two consequences matter later. Nothing is dropped \emph{under a binder}: a
schema declares determinations between named types, never on a recursion
variable, so no slot of a product mentioning~$X$ can be determined by the
others. That is why the operator passes through $\mu$ and $\nu$ untouched,
and why $\gop{X} = X$. And because products are the only place a declaration
is consulted, the operator stays decidable from the schema: every other
clause is a traversal of the type expression, and that one a closure over a
finite set of declarations.

The cost is the closure's. Attribute closure over a finite set of declared
determinations is linear in their size~\cite{beeri1979computational}, and a
\emph{minimal}~$S$ costs at most $n$ of them, so $\gop{\cdot}$ is polynomial
in the type expression and the declarations together. Minimality is what keeps
it so: a \emph{minimum}~$S$ is a least-cardinality candidate key, and deciding
whether one of a given size exists is
NP-complete~\cite{lucchesi1978candidate}.
Definition~\ref{def:compositional-grain} never asks for one, and by
Theorem~\ref{thm:multiple-grains} it could not matter which minimal~$S$ is
returned.

What is polynomial is the propagation, not the grain. The
determinations are input, declared by domain analysis rather than
read off a schema (Remark~\ref{rem:grain-existence-declared}), and
which of the grains so obtained is the \emph{natural} one, the one
displaying the entity key and the order the data is read by, is
declared too (Definition~\ref{def:natural-grain}).
\end{remark}

\noindent
One clause of Definition~\ref{def:compositional-grain} carries the
whole extension and should not pass unexamined: $\gop{X} = X$. Every
other clause discards something, a base type giving way to its grain
and a product to its determined components; this one discards
nothing. Since it is a clause of a definition, the question is not
whether it is \emph{true} but why it is \emph{right}, and the short
answer is that there is nothing at~$X$ to reduce. A reduction removes
a component the survivors determine, and~$X$ is not a component: it
marks the position at which the type refers back to itself, holds no
fields of its own, and carries no declaration, a schema relating named
types and never a recursive occurrence. Irreducibility
(Definition~\ref{def:grain}) is satisfied at~$X$ vacuously.

What makes this more than bookkeeping is what the alternatives do. A
list is $\textit{List}\,R = \mu X.\,\mathbf{1} + R \times X$, in which
$X$ marks the tail. Set $\gop{X} = \mathbf{1}$, discarding the tail on
the plausible ground that it holds no fields of its own, and the
$\mu$-clause gives
\[
  \gop{\textit{List}\,\textit{OrderLine}}
    \;=\; \mu X.\,\mathbf{1} + \gop{\textit{OrderLine}} \times \mathbf{1}
    \;=\; \mathbf{1} + \gop{\textit{OrderLine}},
\]
``either nothing or a single line'': every order has silently lost all
but one of its line items, and the operator is rewriting the data
rather than identifying it. Setting $\gop{X} = X \times X$ fails the
other way, turning the list into a binary tree with no order among its
nodes. Only the identity leaves the shape alone, and the shape is what
the grain of a recursive type has to preserve. That the identity is
also what the recursion \emph{needs}, and what carries the isomorphism
into the fixed point, is settled with
Theorem~\ref{thm:grain-adt} in Section~\ref{subsec:grain-adt-theorem}.

\noindent
$\textit{Order}$ is the case where a declaration does bite. The other
clauses return
$\textit{OrderId} \times \textit{CustomerId} \times
\textit{List}\,\gop{\textit{OrderLine}}$, isomorphic to
$\textit{Order}$ but \emph{not} irreducible, because $\textit{OrderId}$
was declared to determine the other two components. The product clause
drops them, $\textit{OrderId}$ being the smallest sub-product from which a
whole order is recoverable, so
$\gop{\textit{Order}} = \textit{OrderId}$. A nested order document is
identified by its order id, and the operator says so from the declared
determinations alone.

\subsection{The Operator Computes the Grain}
\label{subsec:grain-adt-theorem}

\noindent
What the construction owes is that it is \emph{right}: that what it returns
satisfies both clauses of Definition~\ref{def:grain} for every algebraic data
type, the recursive ones included. One induction establishes it.

\begin{theorem}[Correctness of the grain type-expression operator]
\label{thm:grain-adt}
For every \emph{closed} algebraic data type~$T$, one with no holes in
it, $\textit{Order}$ or $\textit{List}\,\textit{OrderLine}$ rather than
the body $\mathbf{1} + \textit{OrderLine} \times X$ from which the
latter is built (Appendix~\ref{app:primer}), $\gop{T}$ is a grain of~$T$ in the sense of
Definition~\ref{def:grain}: it is isomorphic to~$T$ and irreducible. The
operator therefore \emph{computes} what Definition~\ref{def:grain} only
specifies: from the type expression and the declared determinations it
produces a grain of~$T$, and by Theorem~\ref{thm:multiple-grains} it does
not matter which one. It is moreover idempotent,
$\gop{\gop{T}} = \gop{T}$, which is the operator's counterpart of
Theorem~\ref{thm:grain-idempotent}, that being a statement about the grain
and this one about the syntax that computes it; and it commutes with $+$,
$\mu$ and $\nu$ outright and with $\times$ once determined slots are
dropped. The $\mu$ and $\nu$ cases assume two properties of the fixed-point
structure: isomorphic bodies have isomorphic fixed points, and
irreducibility passes through them.
\end{theorem}

\smallskip
\noindent\emph{Intuition.}\quad Grain simplifies the payload at each
node but never touches how the nodes are wired together, the
recursion, so building the recursive type from simplified parts yields
an equivalent structure. This discharges what the opening promised, and it is what licenses
collapsing the notation: for a \emph{closed} type we may write
$\grain{T}$ for $\gop{T}$, and do so from here on. We keep $\gop{\cdot}$ in two places where the collapse
would lose something: where the discussion is about the operator's
clauses rather than about the grain they compute, as in
Section~\ref{subsec:compositional-grain}, and on open expressions, where
Definition~\ref{def:grain} has nothing to say at all.

\begin{proof}
The mechanism is that grain rewrites payloads and never the wiring, so
a recursive type built from reduced parts is isomorphic to the one
built from the originals. The induction that establishes this passes
under binders, where the body is an open expression and
Definition~\ref{def:grain} has nothing to say, so it is carried out on
open expressions, as Lemma~\ref{lem:grain-open} below; the theorem is
that lemma's closed case together with irreducibility, which the lemma
does not address.

\emph{Isomorphism and idempotency.}\quad For a closed~$T$ there are no
holes to fill, so the lemma's second half reads $\gop{T} \cong T$ and
its first half reads $\gop{\gop{T}} = \gop{T}$.

\emph{Irreducibility.}\quad By induction on~$T$, with the isomorphisms
of the lemma in hand. Base types are Definition~\ref{def:grain}; a sum
is Theorem~\ref{thm:grain-sum}\,(2). In the product case, writing $P$
for $\gop{A_1} \times \cdots \times \gop{A_n}$ and $K$ for
$\prod_{i \in S} \gop{A_i}$, $K$ is irreducible because no further slot
of it is determined by the rest: any proper structural subtype of it
fails to recover~$P$, hence fails to be isomorphic to it. When no
determination crosses the components $K = P$ and the case is
Theorem~\ref{thm:grain-product}. For $\mu X.\,F$ and $\nu X.\,F$ the
statement's second hypothesis is invoked, that irreducibility passes
through the fixed point: $\gop{F}$ is irreducible in every base position
by the induction, and the hypothesis carries this to $\mu X.\,\gop{F}$
and $\nu X.\,\gop{F}$. This hypothesis and its companion, that
isomorphic bodies have isomorphic fixed points, which the lemma's
fixed-point case uses, are the two the statement names; they are
assumed, not proved.

\emph{Commutation.}\quad The operator commutes with $+$, $\mu$ and
$\nu$ because its clauses for those constructors apply the constructor
to the reduced parts, and with $\times$ once determined slots are
dropped, which is what the product clause does.
\end{proof}

\noindent
The isomorphism is concrete, not an existence claim. For
$\textit{List}\,\textit{OrderLine}$ it maps a list of lines to the list of
their grains element by element, and back by looking each
$\textit{ProductId}$ up in the catalogue: same length, same order, which is
the sense in which the recursion is left untouched. For $\textit{Order}$ it
maps an order to its $\textit{OrderId}$ and back through the declared
determination. For a stream it maps an infinite sequence of events to the
infinite sequence of their identifying pairs, both sides infinite, which
Section~\ref{subsec:adt-calcg} takes up.

\begin{remark}[Why irreducibility counts slots, not values]
\label{rem:structural-irreducibility}
Calling the grain the \emph{simplest} type isomorphic to~$R$
(Section~\ref{sec:foundations}) invites a reading
Definition~\ref{def:grain} does not take: that the grain is the one
with the fewest values. On finite products the two agree. On the types
of this section they part company, and it is worth seeing how far.

Take an order's line list. Its inhabitants are the finite lists over a
countably infinite element type, so a G\"{o}del-style encoding gives
$\textit{List}\,\textit{OrderLine} \cong \mathbb{N}$, and the
fewest-values reading would make $\mathbb{N}$ its grain: isomorphic,
with no smaller carrier to descend to. That answer says nothing about
what identifies a line item; it is the \emph{same} answer for every
countably infinite type, since
$\mathbb{N} \cong \mathbb{N} \times \mathbb{N} \cong
\textit{List}\,\mathbb{N}$; and it destroys compositionality, the grain
of a list no longer being a list of grains. Streams fail the same way
one cardinal higher.

Counting slots gives the intended answer instead, and gives it
mechanically: $\grain{\textit{List}\,\textit{OrderLine}} =
\textit{List}\,\grain{\textit{OrderLine}}$ is irreducible because no
slot of it reduces, the element slot already holding a grain and the
recursion slot a variable. $\mathbb{N}$ is not a competitor at all,
being no structural subtype of it
(Definition~\ref{def:structural-subtype}), so the clause never
considers it. The check is the fixed-point test of
Theorem~\ref{thm:grain-adt}, apply the operator again and see whether
anything changes, which is decidable; ``does some smaller carrier
exist?'' is not.
\end{remark}

\noindent
The general reason is bookkeeping. Recursing on syntax, the operator reaches
the body of $\mu X.\,F$ and meets~$X$, a placeholder for the very type it is
computing. It cannot look that answer up, so it assumes the least it can:
whatever $X$ finally stands for is \emph{already a grain}. The assumption is
self-fulfilling, since what $X$ stands for is $\mu X.\,\gop{F}$, an output of
the operator; granted it, leaving $X$ alone is the only option, reducing
something already reduced doing nothing. The lemma states that once and for
all. Write $F[\sigma]$ for the type obtained by
replacing each free variable of~$F$ with the type that the assignment
$\sigma$ gives it, and call $\sigma$ \emph{grain-saturated} when every
type it gives is its own grain, $\gop{\sigma(X)} = \sigma(X)$.

\begin{lemma}[Grain on open expressions]
\label{lem:grain-open}
Let $F$ be a type expression, possibly with holes. Then
$\gop{\gop{F}} = \gop{F}$ as expressions; and for every
substitution~$\sigma$ filling every hole of~$F$ with a type that is
already its own grain, $\gop{F}[\sigma] \cong F[\sigma]$.
\end{lemma}

\begin{proof}
Both halves are by induction on the structure of~$F$.

\emph{Idempotency.}\quad Every base position of $\gop{F}$ already
carries a grain, and $\grain{\grain{B}} = \grain{B}$
(Theorem~\ref{thm:grain-idempotent}); the variable clause $\gop{X} = X$
is idempotent outright; the sum and fixed-point clauses apply the
constructor to reduced parts, which a second pass leaves as they are;
and the product clause drops determined slots, of which none remain
after one pass. So $\gop{\gop{F}} = \gop{F}$ as expressions.

\emph{Isomorphism.}\quad Let $\sigma$ be grain-saturated. The base case
is Definition~\ref{def:grain}. The variable case holds because
$\gop{X}[\sigma]$ and $X[\sigma]$ are the same type, $\sigma(X)$, not
merely isomorphic ones; this is the clause under discussion. The sum
case is Theorem~\ref{thm:grain-sum}\,(2). In the product case, writing
$P$ for $\gop{A_1} \times \cdots \times \gop{A_n}$ and $K$ for
$\prod_{i \in S} \gop{A_i}$, $K \ssub P$ and $P \subt K$ are mutual
admissible surjections, so $K \cong P$ by antisymmetry
(Theorem~\ref{thm:grain-partial-order}, product case). For
$\mu X.\,F$, let $\hat F$ and $\hat G$ denote the polynomial functors
$Y \mapsto F[Y/X]$ and $Y \mapsto \gop{F}[Y/X]$, so that
$\mu X.\,F = \mu\hat F$ and $\gop{\mu X.\,F} = \mu\hat G$. By the
induction hypothesis the grain isomorphisms hold on every parameter
sub-expression of~$F$; since $\gop{X}=X$ leaves the recursion variable
untouched, since no determination is declared on a recursion variable and
so no slot is dropped from any product of~$F$ mentioning~$X$, and since
$\times$, $+$ preserve isomorphism and naturality, these
assemble into a \emph{natural} isomorphism of functors
$\hat F \cong \hat G$. That naturally isomorphic polynomial functors
have isomorphic initial algebras is the first hypothesis
Theorem~\ref{thm:grain-adt} names, taken from~\cite{birddemoor1997};
the isomorphism it supplies is the fold of one constructor along the
other, a structure map (Section~\ref{sec:prelim}), hence
$\mu X.\,F \cong \mu X.\,\gop{F} = \gop{\mu X.\,F}$. The coinductive
case is identical, with final coalgebras in place of initial algebras,
unfolds in place of folds, and $\nu$ in place of~$\mu$.
\end{proof}

\noindent
The second half is the soundness a binder needs: reduce the body first and
fill in already-reduced types afterwards, and you land on a type isomorphic to
the one you would have got by filling in first. The first says the operator is
idempotent even on open expressions, which makes the grain-saturation
hypothesis available again at the \emph{next} binder, so nested recursive
types -- a list of trees, a stream of orders -- are covered without further
argument.
On the running example, $F$ is $\mathbf{1} + \textit{OrderLine} \times X$:
reducing it gives $\mathbf{1} + \gop{\textit{OrderLine}} \times X$, and
filling~$X$ afterwards with a type already reduced lands on the same
list, up to isomorphism, as filling first and reducing after.
Its variable case is the clause under discussion, and it goes through
because $\gop{X}[\sigma]$ and $X[\sigma]$ are the same type,
$\sigma(X)$, not merely isomorphic ones. Send $X$ anywhere else and this
case fails immediately, taking the isomorphism with it in the body of
every binder and so in the fixed point, exactly as the collapsed list
above shows. The hypothesis on $\sigma$ is never an obligation in
practice, because it discharges itself at the binder: there $X$ is
instantiated by $\mu X.\,\gop{F}$, which is its own grain by the
idempotence half of the lemma,
$\gop{\mu X.\,\gop{F}} = \mu X.\,\gop{\gop{F}} = \mu
X.\,\gop{F}$.

\begin{remark}[What $\gop{X} = X$ rules out]
\label{rem:gx-scope}
The clause fixes the theory's scope. Because grain never touches the
recursion, it does not look for determinations running \emph{across}
recursive occurrences, that each element's key is one more than its
predecessor's, say. Such facts are invariants of the data rather than
statements about slots of a type, and are excluded by design.
Excluding them costs no soundness: $\gop{T} \cong T$ still holds
(Theorem~\ref{thm:grain-adt}), and $\gop{T}$ remains irreducible
relative to the declarations, which is what
Definition~\ref{def:grain} asks.
\end{remark}

\subsection{The Three Types, Answered}
\label{subsec:windowed-collections}

The operator answers the first of the three types outright:
$\grain{\textit{Order}} = \textit{OrderId}$, computed above from the
type expression and its declared determinations, with no relation
and no attribute set in sight. The change
log and the sensor feed need one more step, and naming it precisely is
half the answer. They are bare \emph{types}, and a bare stream type has no
finite grain. What a pipeline actually reads is never the bare stream but
a \emph{partitioned collection} of its elements, grouped by a declared
key (Definition~\ref{def:partitioned-collection} below), and that is the
step a stream processor takes anyway.

\subsubsection{An Infinite Grain, and Why It Is Not a Defect}
\label{subsec:adt-calcg}

For a coinductive type Theorem~\ref{thm:grain-adt} yields an
\emph{infinite} grain. A stream is
$\textit{Stream}\,R = \nu X.\,R \times X$, so
$\grain{\textit{Stream}\,R} = \nu X.\,\grain{R} \times X
= \textit{Stream}\,\grain{R}$, whose inhabitants are infinite
sequences. This is not a defect of the operator but a fact about the
type: the bare coinductive datatype has \emph{no finite grain}, because a
grain is isomorphic to its type and no finite type is isomorphic to an
infinite stream. The grain is nonetheless finitely \emph{presented}: as a
type expression, $\textit{Stream}\,\grain{R}$ is the finite syntax
$\nu X.\,\grain{R} \times X$. The distinction matters, and the next
paragraph says why.

This is the point at which to state precisely what CalcG
(Section~\ref{sec:calcg}) requires. CalcG is a topological traversal
of a finite pipeline DAG that, at each vertex, applies a per-operation
\emph{grain-inference function} to the grains of the vertex's
predecessors. It needs only that this function be \emph{computable}
and the DAG finite. Grains are, in general, algebraic-data-type
expressions, finite syntax, $\mu$ and $\nu$ included, so an infinite
grain is no obstacle to CalcG \emph{as such}. The finite field sets
and the rules of Table~\ref{tab:ra-inference} are the instantiation of
grain inference for the \emph{relational algebra}; they are a property
of that operation set, not of CalcG. What an unbounded stream lacks is
not a representable grain but a settled, computable grain-inference
function for raw stream operators, which is exactly what windowing
supplies.

\subsubsection{Windowing Restores a Finite Grain}
\label{subsec:windowed-collections-sub}

A stream acquires a finite, verifiable grain the moment it is
\emph{windowed} or \emph{partitioned}, the same move a stream
processor~\cite{carbone2015flink,akidau2015dataflow,armbrust2018structured}
must perform to evaluate an unbounded input at all. Call a type
\emph{grain-irreducible} when it is its own grain, $\grain{K} = K$.

\begin{definition}[Partitioned collection]
\label{def:partitioned-collection}
Let $R$ be an element type and $K \ssub \grain{R}$ a grain-irreducible
\emph{structural} subtype of its grain, the \emph{partition key}. The $K$-partitioned collection
$\textit{Coll}_K\,R$ is the type whose inhabitants are the
\emph{complete fibers of the $K$-projection}: one inhabitant per
$k : K$, holding every $R$-element whose $K$-projection is~$k$. Each
inhabitant is determined by, and unique to, its key, so
$\textit{Coll}_K\,R \cong K$.
\end{definition}

\noindent
The containment must be structural: $\grain{R} \tdiff K$ of
Proposition~\ref{prop:partitioned-grain} is a difference of fields
(Section~\ref{sec:prelim}), so $K$'s components must be among
$\grain{R}$'s, and under $\subt$ they need not be.

\noindent
Concretely, an inhabitant of $\textit{Coll}_K\,R$ is one whole group: for
$K = \textit{OrderId}$ over order events, a single inhabitant is the
entire event history of one order, and there is exactly one inhabitant
per order.

\begin{proposition}[Grain of a partitioned collection]
\label{prop:partitioned-grain}
$\grain{\textit{Coll}_K\,R} = K$, a finite product type. Writing
$V = \grain{R} \tdiff K$, the element grain decomposes as
$\grain{R} = K \times V$: the key~$K$ is shared by all elements within
one collection inhabitant, and $V$ identifies an element inside its
fiber. Because $K$ is a finite product type, CalcG
(Section~\ref{sec:calcg}) applies to
$\textit{Coll}_K\,R$ exactly as to a relational collection, whether
each fiber is a finite list or an unbounded stream.
\end{proposition}

\begin{proof}
$\textit{Coll}_K\,R \cong K$ by
Definition~\ref{def:partitioned-collection}, and $K$ is
grain-irreducible by that definition; by Definition~\ref{def:grain},
$\grain{\textit{Coll}_K\,R} = K$. The decomposition
$\grain{R} = K \times V$ is the type-difference split of the product
$\grain{R}$ into~$K$ and its complement $V = \grain{R} \tdiff K$.
\end{proof}

\noindent
The partition key~$K$ is a \emph{declared} choice, part of the
collection's specification, not a function of~$R$, and is typically
fixed by the collection's behavioral class (Section~\ref{sec:entity});
the completeness of the fibers is precisely the coverage invariant of
the \textbf{IsSnapshot} class.

\smallskip
This settles the two types left open at the start of the section.
\emph{The change log.} For $\textit{CDC-Event}$,
$\grain{\textit{CDC-Event}} = \textit{EntityId} \times
\textit{ChangeSeqNo}$. The log partitions by $K = \textit{EntityId}$, one
inhabitant per entity holding the ordered sequence of its changes, which
is exactly a Kafka topic partition. Each fiber is an \emph{unbounded}
stream, yet $\grain{\textit{Coll}} = \textit{EntityId}$ is finite: the
grain is the partition key, not the fiber content.
\emph{The sensor feed.} For $\textit{WeatherObservation}$,
$\grain{\textit{WeatherObservation}} = \textit{StationId} \times
\textit{ObservationDtm}$, and the bare
$\textit{Stream}\,\textit{WeatherObservation}$ has only the infinite grain
of Section~\ref{subsec:adt-calcg}. Different windowings yield different
finite grains, each its own collection type. By station,
$K = \textit{StationId}$, already a component of the grain. By day or
by tumbling window the key is \emph{derived} first, an
$\textit{ObservationDate}$ by truncating
$\textit{ObservationDtm}$ or a $\textit{WindowId}$ by bucketing it,
after which the derived field is a component of the element grain and
$K = \textit{ObservationDate}$, or
$K = \textit{StationId} \times \textit{WindowId}$, satisfies
Definition~\ref{def:partitioned-collection}. The engineer's choice
of window \emph{is} the choice of grain, and the operator makes that
visible in the type.
\emph{A periodic snapshot}, for good measure. For
$\textit{ProductBalance}$, one record per product per date,
$\grain{\textit{ProductBalance}} = \textit{ProductId} \times
\textit{BalanceDate}$. The snapshot collection partitions by
$K = \textit{BalanceDate}$: one inhabitant per date, holding that date's
complete set of product balances, so
$\grain{\textit{Coll}} = \textit{BalanceDate}$ and
$V = \textit{ProductId}$ locates a balance within a snapshot.
This is Kimball's periodic snapshot fact
table~\cite{kimball2013dw}, its snapshot-date grain declared by
practitioners for three decades;
Definition~\ref{def:partitioned-collection} is the formal reason it is a
grain at all.

\smallskip
Grain theory thus covers every algebraic data type. For finite
data (products, sums, inductive collections) the grain is finite and
CalcG verifies grain-correctness statically, from the declarations
alone. For
infinite data the bare coinductive grain is infinite, but every
windowed or partitioned collection has the finite grain~$K$ and is
verified as an ordinary finite collection. A streaming pipeline is
therefore grain-correct at design time exactly when each of its
windowed stages is, and a windowed stage is the only form in which a
stream processor ever evaluates a stream.

\subsection{Grain as a Minimal Identifying Family}
\label{subsec:identifying-family}

This section records a second way to say what a grain is, and it earns
a section because it is not a restatement. Definition~\ref{def:grain}
asks two things, that a grain be isomorphic to its type and that it be
irreducible, and it treats internal and external grains separately, the
first a structural subtype of the type and the second not
(Section~\ref{sec:foundations}). The formulation below asks one thing,
that a family of observations be minimal among those that identify, and
the two cases collapse into it: a field, a declared projection, an
entity map and a snapshot date are observations alike. The isomorphism
is then not assumed but obtained, the combined map of an identifying
family being injective and onto its image.

It also reaches further, which is the reason to be careful with it. An
observation need not project onto a component, so a truncated timestamp
or a window index identifies as well as a field does, and the derived
keys of Definition~\ref{def:partitioned-collection} would need no
deriving. That freedom is what
Remark~\ref{rem:structural-irreducibility} warns about: admit arbitrary
observations and a single G\"{o}del encoding is a minimal identifying
family, returning the useless grain. The formulation therefore needs a
restriction on which observations are admissible, which the present
theory gets for free from structural containment. Supplying that
restriction is what makes this a direction rather than the definition
of Section~\ref{sec:foundations}.

Grain generalizes the relational \emph{candidate key}. On a flat
relation a candidate key is a minimal set of attributes whose values
jointly identify each row; on product types that is exactly the grain
(Section~\ref{subsec:rw-fd}). The same idea lifts to \emph{any}
algebraic data type once ``attribute'' is read as \emph{observation}, a
projection out of the type.

\begin{definition}[Identifying family]
\label{def:identifying-family}
Let $R$ carry a family of observations $\{p_i : R \to A_i\}_{i \in I}$:
its fields, together with any declared projections, an entity map
$\textit{entity} : R \to E$ (Section~\ref{sec:entity}) or the snapshot
date of a partitioned collection (Section~\ref{sec:adt}). The family is
\emph{identifying} if agreeing on every $p_i$ forces equality,
$(\forall i.\ p_i\,r = p_i\,r') \Rightarrow r = r'$. A \emph{minimal
identifying family} is one no proper subfamily of which is identifying.
\end{definition}

\noindent
The grain is, equivalently, a minimal identifying family, presented as
the image of its combined map. On product types the observations are the
fields, and identifying families are exactly keys.

\begin{proposition}[Identifying families on product types]
\label{prop:identifying-product}
On a product type $R$ whose observations are its fields, a family is
identifying iff it is a superkey; a \emph{minimal} identifying family is
therefore a candidate (minimal) key, and the image of its combined map
is the grain $\grain{R}$ of Definition~\ref{def:grain}.
\end{proposition}

\begin{proof}
A field family $\{p_i\}_{i \in S}$ is identifying iff
$p_i\,r = p_i\,r'$ for all $i \in S$ implies $r = r'$, which is exactly the condition that the fields in $S$
jointly identify each row, i.e.\ that
$S$ is a superkey; minimality of $S$ is then candidacy. For minimal
$S$, let $G$ be the image of
$\langle p_i \rangle_{i \in S} : R \to \prod_{i \in S} A_i$. This
combined map is injective (identifying) and onto its image, so
$G \cong R$; and no proper sub-family, equivalently no proper
structural subtype of $G$, is identifying, so $G$ is irreducible.
Hence $G$ satisfies Definition~\ref{def:grain}, and $G = \grain{R}$.
\end{proof}

\noindent
Observations need not be fields, and that is where the generalization
earns its keep. For $\textit{Order}$ the single observation
$\textit{orderId} : \textit{Order} \to \textit{OrderId}$ is already
identifying, and no smaller family is, which recovers
$\grain{\textit{Order}} = \textit{OrderId}$ without ever asking what the
``attributes'' of a nested document are. A declared date or entity map is
an observation like any other, so \emph{external} grains
(Section~\ref{sec:foundations}) are minimal identifying families that come
``for free,'' and on sums and (co)inductive types the observations are
constructor positions, so the view agrees with the compositional grain
(Definition~\ref{def:compositional-grain}) throughout. Categorically the
family is a \emph{jointly-monic cone}, with minimality in the
\emph{index set} rather than the carrier, which is what separates
$\textit{CustomerId}$ from
$\textit{CustomerId} \times \textit{CustomerName}$, isomorphic yet
distinct as key-families.

Nothing in the development above depends on this formulation: the
candidate-key correspondence is exact either way
(Section~\ref{subsec:rw-fd}), and Definition~\ref{def:grain} covers
internal and external grains as they stand. What taking it as
\emph{primitive} would buy is the economy described at the head of this
section, index-set minimality replacing irreducibility and the
isomorphism falling out; what it costs is the admissibility question.
We pursue both in future work.

\section{Grain Inference for Relational Operations}
\label{sec:inference}

We now turn to the paper's computational core: systematic inference
of grain through data transformations. Given the grains of input types,
what is the grain of the result? We answer this question first for
equi-joins, the most complex case, and then for the full relational
algebra.

\smallskip\noindent\textbf{Scope.}\quad
The relational algebra operates on product types: relations are sets of
tuples, each a product of attribute values. Throughout this section the
inclusions between field sets are therefore \emph{structural}: one field
set's components sit among another's. We write them with the structural
subtype $\ssub$ (Definition~\ref{def:structural-subtype}), and $\tin$,
$\tun$, $\tdiff$ denote the corresponding component-set operations. Since $A \ssub B$ implies $A \subt B$, every such inclusion is
an instance of the general grain relations of
Sections~\ref{sec:foundations},~\ref{sec:relations}, which remain
universal; the inference rules of this section specialize to RA-specific
operations because that is the algebra under study. The general
surjection-based reading of $\subt$ is not needed for the relational
rules, and we set it aside here. Under that reading a field set already
\emph{includes} any grain it merely determines, such as an alternative
key. We likewise set aside the extension to richer type disciplines, such
as sum and recursive types, and to transformation algebras beyond the
relational one.

\noindent
Proofs for this section are in Appendix~\ref{sec:proofs}, except
Proposition~\ref{prop:grain-reduction-fd}, whose argument is short and is
given in place. The independence hypothesis of
Theorem~\ref{thm:equijoin-grain} is discussed at length in
Remark~\ref{rem:info-independence}, in
Appendix~\ref{sec:deferred-remarks}.

\subsection{Grain Inference for Equi-Joins}
\label{subsec:equijoin-inference}

We assume an abstract collection type $C$ such that $C\;R$ denotes a
finite bag of elements of type~$R$. We do not assume set semantics.
By Theorem~\ref{thm:grain-uniqueness}, the grain projection is injective,
so distinct elements of $C\;R$ are uniquely identified by their grain values.
A bag may still repeat an element, so call $C\,R$ \emph{grain-unique} when
no two of its elements share a grain value.

\begin{theorem}[Equi-Join Candidate Grains]
\label{thm:equijoin-candidates}
For any two collections $C\;R_1$ and $C\;R_2$ sharing common fields
of type $J_k$ (i.e., $J_k \ssub R_1$ and $J_k \ssub R_2$), the
equi-join on~$J_k$ produces a result $C\;Res$ where
$Res = (R_1 \tdiff J_k) \times (R_2 \tdiff J_k) \times J_k$.
For either labeling $(i,j) \in \{(1,2),(2,1)\}$, the candidate
\[
G_{cand}(i,j) \;=\; \grain{R_i} \tun (\grain{R_j} \tdiff J_k)
\]
satisfies $G_{cand}(i,j) \eqg Res$: each candidate is
isomorphic to the grain of the result and determines every element
of~$Res$. The lattice bounds
\[
(\grain{R_1} \times \grain{R_2}) \leg Res \leg
(\grain{R_1} \tin \grain{R_2})
\]
hold for both candidates. The lower bound is achieved when $J_k$
shares no fields with either grain (the cartesian-product case);
the upper bound when the grains are equal and $J_k$ covers the
grain.
\end{theorem}

\noindent
Theorem~\ref{thm:equijoin-candidates} gives two types isomorphic to
$Res$. The grain is the unique-up-to-isomorphism \emph{irreducible}
one (Definition~\ref{def:grain}); the next theorem identifies it by a
schema-level comparison of the inputs' join-key grain portions.

\begin{theorem}[Grain of an Equi-Join]
\label{thm:equijoin-grain}
Assume the inputs are \emph{informationally independent} outside the
join key: no field of $\grain{R_2} \tdiff J_k$ is determined by
$\grain{R_1}$, and none of $\grain{R_1}$ by $\grain{R_2} \tdiff J_k$
(Remark~\ref{rem:info-independence}). Let
$G_i^{J_k} = \grain{R_i} \tin J_k$ denote the $J_k$-portion of
each input grain. The grain $\grain{Res}$ is the irreducible
candidate of Theorem~\ref{thm:equijoin-candidates}, selected by
$\ssub$-comparing $G_1^{J_k}$ and $G_2^{J_k}$:
\begin{itemize}
\item If $G_1^{J_k} \ssub G_2^{J_k}$ (the \emph{canonical}
  labeling), then
  \[
    \grain{Res} \;=\; G_{cand}(1,2) \;=\;
    \grain{R_1} \tun (\grain{R_2} \tdiff J_k).
  \]
  When the inclusion is strict, the reverse candidate $G_{cand}(2,1)$
  properly contains $G_{cand}(1,2)$, so it fails the
  no-proper-sub-type clause of Definition~\ref{def:grain} and is
  reducible; when $G_1^{J_k} = G_2^{J_k}$ the two candidates coincide.
\item If $G_1^{J_k}$ and $G_2^{J_k}$ are $\ssub$-incomparable, both
  candidates are irreducible: each is a grain of $Res$, and the two
  are isomorphic (Theorem~\ref{thm:multiple-grains}).
\end{itemize}
\end{theorem}

\begin{remark}[Grain-level determination]
\label{rem:grain-level-determination}
In Theorem~\ref{thm:equijoin-grain} \emph{determined by} is grain-level: the
determinant is a \emph{whole grain} (Remark~\ref{rem:grain-vs-fd}). An
arbitrary declared functional dependency whose left side is no grain, such as
a composite key spanning both inputs, is not seen at this level. It is applied
afterwards by the reduction of Proposition~\ref{prop:grain-reduction-fd}.
\end{remark}

\begin{remark}[FD Correspondence]
\label{rem:equijoin-fd}
On product types, Theorem~\ref{thm:equijoin-candidates} specializes
Klug's key-propagation rule for
equi-joins~\cite{klug1980calculating}; the irreducibility refinement
of Theorem~\ref{thm:equijoin-grain}, gated on the $\ssub$-comparability
of the join-key grain portions $G_1^{J_k}$ and $G_2^{J_k}$, does not
appear as a named result in the key-propagation literature
(Section~\ref{subsec:rw-fd}).
\end{remark}

\noindent\textit{Proof sketch.}\quad
Theorem~\ref{thm:equijoin-candidates} gives $G_{cand}(i,j) \eqg Res$
for either labeling, by \emph{bootstrapping}: $\grain{R_i}$
determines $r_i$, hence $r_i.J_k$; the join equality propagates this
to $r_j.J_k$, and with $\grain{R_j} \tdiff J_k$ in the candidate all
of $\grain{R_j}$ is recovered. Theorem~\ref{thm:equijoin-grain} then
asks which candidate is \emph{irreducible}. By independence a candidate's only removable fields are its $J_k$-fields,
since a non-$J_k$ field of one input's grain is by hypothesis not recovered
from the other, and such a field is redundant only if the opposite input grain
recovers it through the join. Under
the canonical labeling $G_1^{J_k} \ssub G_2^{J_k}$, every $J_k$-field
of $G_{cand}(1,2)$ lies in $\grain{R_2}$, so removing it also breaks
the $\grain{R_2}$ route: $G_{cand}(1,2)$ is irreducible. Under the
strict reverse labeling $G_{cand}(2,1)$ carries surplus $J_k$-fields
that $\grain{R_1}$ still recovers, so it is reducible. Full proofs in
Appendix~\ref{sec:proofs}. \qed

\smallskip\noindent
Note that $J_k$ is a \emph{type}, the common fields where
$J_k \ssub R_1$ and $J_k \ssub R_2$, not a join predicate. Because
the formula operates entirely at the type level, it applies uniformly
to inner and outer joins. Semi-joins and anti-joins return only~$R_1$,
so $\grain{Res} = \grain{R_1}$ (Table~\ref{tab:ra-inference}).

\smallskip\noindent\textbf{Example (canonical vs.\ reverse labeling).}\quad
Using the Section~\ref{subsec:armstrong} types: join
$R_1 = \textit{Account} \times \textit{Customer}$
($\grain{R_1} = \textit{Account}$) with
$R_2 = \textit{Customer} \times \textit{Address}$
($\grain{R_2} = \textit{Customer} \times \textit{Address}$) on
$J_k = \textit{Customer}$. The $J_k$-portions are
$G_1^{J_k} = \textit{Account} \tin \textit{Customer} = \emptyset$
and $G_2^{J_k} = (\textit{Customer} \times \textit{Address})
\tin \textit{Customer} = \textit{Customer}$, so
$G_1^{J_k} \spsub G_2^{J_k}$ (comparable, with the canonical labeling
already on input~1). The two candidates of
Theorem~\ref{thm:equijoin-candidates}:
\begin{itemize}
\item $G_{cand}(1,2) = \textit{Account} \tun \textit{Address}
  = \textit{Account} \times \textit{Address}$, which is irreducible
  (\textit{Customer} is eliminated by bootstrapping: determined by
  \textit{Account} through the join). By
  Theorem~\ref{thm:equijoin-grain},
  $\grain{Res} = \textit{Account} \times \textit{Address}$.
\item $G_{cand}(2,1) = \textit{Customer} \times \textit{Address}
  \times \textit{Account}$, not irreducible: \textit{Customer} is
  determined within the candidate by \textit{Account} (via
  $\grain{R_1} \to \textit{Customer}$ in $R_1$, propagated through
  the join equality). Hence
  $G_{cand}(2,1) \eqg Res$ but
  $G_{cand}(2,1) \neq \grain{Res}$.
\end{itemize}

\subsection{Special Cases}
\label{subsec:special-cases}

\begin{proposition}[Join Special Cases]
\label{prop:join-special-cases}
Under the independence hypothesis of
Theorem~\ref{thm:equijoin-grain}, that theorem simplifies for the three
main grain relationships between the inputs. A fourth case addresses the
join-key pattern of natural joins. (In cases~1--2 the hypothesis is
vacuous, since $\grain{R_2} \tdiff J_k = \emptyset$; in cases~3--4 it is
operative.) The formulas are the theorem's under its canonical labeling,
and they fail under the strict reverse labeling
$G_2^{J_k} \spsub G_1^{J_k}$; cases~2 and~4 carry that condition
explicitly, since in case~2 the labels are fixed by which input is finer
and in case~4 the symmetric join key hides that the formula is not
symmetric.
\begin{enumerate}
\item \textbf{Equal grains} ($R_1 \eqg R_2$): if
  the join key covers both grains ($\grain{R_1} \ssub J_k$ and
  $\grain{R_2} \ssub J_k$), then $\grain{Res} = \grain{R_1}$
  (equivalently $\grain{R_2}$, since $R_1 \eqg R_2$ makes the two
  isomorphic grains of $Res$); otherwise the general formula applies.
\item \textbf{Ordered grains} ($R_1 \leg R_2$, $R_1$ finer) with
  $\grain{R_2} \ssub J_k$, provided $\grain{R_2}$ is not a proper
  structural subtype of $G_1^{J_k}$, which is automatic when
  $J_k = \grain{R_2}$ (the 1-to-many pattern, e.g.,
  $\textit{OrderDetail} \bowtie \textit{Order}$ on
  $\textit{OrderId}$): $\grain{Res} = \grain{R_1}$, preserving the
  finer grain. If instead $\grain{R_2} \spsub G_1^{J_k}$, the coarser
  grain determines join-key fields of the finer one and
  $\grain{Res} = \grain{R_2} \tun (\grain{R_1} \tdiff J_k)$.
\item \textbf{Incomparable grains} ($R_1 \incg R_2$):
  if neither grain is contained in~$J_k$, the result grain is strictly
  finer than both, the classic source of fan traps
  (Section~\ref{sec:errors}).
\item \textbf{Natural join} ($J_k = R_1 \tin R_2$), with the inputs
  labeled so that $G_2^{J_k}$ is not a proper structural subtype of
  $G_1^{J_k}$ (the canonical labeling when the two are comparable,
  either labeling otherwise):
  $\grain{Res} = \grain{R_1} \tun (\grain{R_2} \tdiff R_1)$, that is,
  all of $\grain{R_1}$ plus the $R_2$-unique grain fields.
\end{enumerate}
\end{proposition}

\subsection{Generalized Equi-Join}
\label{subsec:generalized-join}

In practice, join keys often have different structures across systems
(e.g., a composite key $\textit{Year} \times \textit{Month} \times
\textit{Day}$ in one source vs.\ a single $\textit{Date}$ type in
another). Theorem~\ref{thm:equijoin-candidates} requires identical
join keys ($J_k \ssub R_1$ and $J_k \ssub R_2$); the following
generalization relaxes identity to isomorphism while preserving the
physical-presence requirement.

\begin{theorem}[Generalized Grain Inference for Equi-Joins]
\label{thm:generalized-equijoin}
For collections $C\;R_1$ and $C\;R_2$ with join key types
$J_{k1} \ssub R_1$ and $J_{k2} \ssub R_2$ respectively, where
$J_{k1} \eqg J_{k2}$ via isomorphism
$\phi : J_{k1} \stackrel{\cong}{\longrightarrow} J_{k2}$,
the equi-join matching rows where
$\phi(\pi_1(r_1)) = \pi_2(r_2)$ produces a result with grain:
\[
Res \eqg \grain{R_1} \tun (\grain{R_2} \tdiff J_{k2})
\]
where the type-level operations are applied modulo~$\phi$.
The physical-presence premises ($J_{k1} \ssub R_1$,
$J_{k2} \ssub R_2$) are unchanged from
Theorem~\ref{thm:equijoin-candidates}: each input must contain its own
join key fields for projection and bootstrapping to work; only the
matching criterion is relaxed from equality to equality up to~$\phi$.
\end{theorem}

\noindent
The proof reduces to Theorem~\ref{thm:equijoin-candidates} via the
isomorphism (Appendix~\ref{sec:proofs}).

\smallskip\noindent\textbf{Example.}\quad
Joining $R_1$ keyed on $\textit{Year} \times \textit{Month} \times
\textit{Day}$ with $R_2$ keyed on $\textit{Date}$ via the
date-composition isomorphism $\phi$ yields
$Res \eqg R_1 \eqg R_2$.

\subsection{Grain Reduction under Declared Functional Dependencies}
\label{subsec:grain-reduction-fd}

Theorem~\ref{thm:equijoin-grain} computes the grain of an equi-join from
the \emph{grains} of its inputs, that is, from determinations whose left
side is a whole grain (Remark~\ref{rem:grain-vs-fd}). A schema may additionally
declare a functional dependency whose left side is an \emph{arbitrary}
attribute set, possibly spanning both inputs and matching no grain. Such
a dependency is invisible to the grain-level rule, yet it can shrink the
grain further. A classical key minimization closes the gap.

\begin{proposition}[Grain Reduction under Declared FDs]
\label{prop:grain-reduction-fd}
Let $K$ be the type Theorem~\ref{thm:equijoin-grain} returns for $Res$
under its grain-level independence hypothesis, and let $F$ be a set of
functional dependencies declared over the attributes of the inputs. A
grain is relative to the determinations in force
(Remark~\ref{rem:grain-existence-declared}), so $K$ is a grain of $Res$
under the determinations that theorem assumes, and declaring $F$ can
shrink it: under $F$ the type $K$ is still a superkey of $Res$ but need
not be irreducible. The grain $\grain{Res}$ under $F$ is obtained from
$K$ by deleting any attribute functionally determined, via $F^{+}$, by
the remaining attributes, iterated to a fixpoint. It is unique up to
isomorphism (Theorem~\ref{thm:multiple-grains}), and the computation is
polynomial in the number of attributes and $|F|$.
\end{proposition}

\begin{proof}
By Theorem~\ref{thm:equijoin-grain}, $K \cong Res$, so the projection
restricted to $K$ identifies every element of $Res$; hence $K$ is a
superkey of $Res$, and it stays one under any~$F$, since declaring
dependencies only adds constraints and never removes the
identification. For minimality we use the standard
superkey-to-candidate-key reduction: an attribute $a \in K$ is
\emph{extraneous} exactly when $(K \setminus \{a\}) \to a$ lies
in~$F^{+}$, and then $K \setminus \{a\}$ is still a superkey, so
$a$ may be dropped without losing identification. When no extraneous
attribute remains, the set is a candidate key, a minimal type
identifying $Res$, hence a grain of $Res$ (Definition~\ref{def:grain}).
Any two candidate keys are grains of $Res$ and therefore isomorphic
(Theorem~\ref{thm:multiple-grains}), so the fixpoint is well-defined up to
isomorphism regardless of deletion order. Each extraneity test is one
attribute-closure computation ($O(|F|)$), and at most $|K|$
attributes are removed, giving the polynomial bound. This is the
product-type instance of grain implication, which coincides with
functional-dependency implication
(Corollary~\ref{cor:grain-implication}).
\end{proof}

\smallskip\noindent\textbf{Example (a composite cross-input FD).}\quad
Join $R_1 = \textit{RegionId} \times \textit{StoreId}$
($\grain{R_1} = \textit{RegionId} \times \textit{StoreId}$) with
$R_2 = \textit{RegionId} \times \textit{Season} \times \textit{Discount}$
($\grain{R_2} = \textit{RegionId} \times \textit{Season} \times
\textit{Discount}$) on $J_k = \textit{RegionId}$. No whole grain
determines a field of the other, so the independence hypothesis of
Theorem~\ref{thm:equijoin-grain} holds and it returns
$K = \textit{RegionId} \times \textit{StoreId} \times
\textit{Season} \times \textit{Discount}$. Suppose the schema also
declares $\{\textit{StoreId}, \textit{Season}\} \to \textit{Discount}$,
meaning each store sets its own seasonal discount. That determinant is
\emph{neither} input's grain, so no grain-level rule sees it. By
Proposition~\ref{prop:grain-reduction-fd}, $\textit{Discount}$ is
determined by the remaining attributes and is removed, yielding
$\grain{Res} = \textit{RegionId} \times \textit{StoreId} \times
\textit{Season}$.
The grain-level rule is correct as far as \emph{grains} determine. The
declared attribute-level FD refines it. This is the two-layer picture of
Remark~\ref{rem:grain-vs-fd} at work.

\subsection{Complete Inference Rules for the Relational Algebra}
\label{subsec:ra-rules}

Beyond equi-joins, grain inference applies to the full relational algebra.
The governing principle is: if an operation preserves all grain fields in
the result type, the grain is unchanged. Table~\ref{tab:ra-inference}
summarizes the rules for all standard operations.

\begin{table}[t]
\centering
\small
\begin{tabular}{@{}llp{0.42\linewidth}p{0.2\linewidth}@{}}
\toprule
\textbf{Operation} & \textbf{$Res$} & \textbf{$\grain{Res}$} & \textbf{Why} \\
\midrule
\multicolumn{4}{@{}l}{\textit{Unary}} \\
Selection $\sigma$ & $R$ & $\grain{R}$ & type unchanged \\
Projection $\pi_S$ & $S$ & $G'$, a grain of $R$ with $G' \ssub S$, when $S$ determines $R$ ($\mathrm{comps}(R) \subseteq \mathrm{comps}(S)^{+}$) & else no rule; flagged (Prop.~\ref{prop:projection-dup}) \\
Extension & $R \times D$ & $\grain{R}$ & $D$ dependent \\
Rename $\rho$ & $R'$ & $R' \eqg R$ & structure preserved \\
Grouping $\gamma_{G_c}$ & $G_c \times Agg$ & $\grain{G_c}$ & one row/group \\
\midrule
\multicolumn{4}{@{}l}{\textit{Set operations (union-compatible types)}} \\
$\cup$, $\cap$, $-$ & $R$ & $\grain{R}$ & type unchanged \\
\midrule
\multicolumn{4}{@{}l}{\textit{Joins}} \\
Equi-join & Thm~\ref{thm:equijoin-grain} & $\grain{R_1} \tun (\grain{R_2} \tdiff J_k)$ & Thm~\ref{thm:equijoin-grain} \\
Cross/$\theta$-join & $R_1 \times R_2$ & $\grain{R_1} \times \grain{R_2}$ if $\theta$ non-functional & Thm~\ref{thm:grain-product} \\
Semi $\ltimes$ & $R_1$ & $\grain{R_1}$ & filter only \\
Anti $\rhd$ & $R_1$ & $\grain{R_1}$ & filter only \\
\bottomrule
\end{tabular}
\caption{Grain inference rules for relational algebra operations.
  Proofs for all rules appear in Appendix~\ref{sec:proofs}.}
\label{tab:ra-inference}
\end{table}

\noindent
Four operations deserve attention. \textbf{Projection} may
\emph{coarsen} the grain when $S$ does not determine~$R$ (the closure
test of Proposition~\ref{prop:projection-dup}), losing grain
fields and producing duplicates, a common bug detectable at design
time. \textbf{Extension} preserves grain (the derived column is
functionally dependent on~$R$; window functions are typical
examples). \textbf{Grouping} \emph{determines} a new grain
$\grain{G_c}$ from the grouping columns.

\textbf{Cross and $\theta$-joins} inherit the independence hypothesis of
Theorem~\ref{thm:grain-product}: the rule holds when no determination
arises across the inputs in the result
(Remark~\ref{rem:compositional-independence}). A declared foreign key
between $R_1$ and $R_2$ does not by itself break it, because $Res$ pairs each
$r_1$ with \emph{every} $r_2$ passing~$\theta$, so a determination
holding between \emph{matching} elements does not hold of arbitrary
pairs. What breaks it is a \emph{functional} predicate, one admitting at
most one $r_2$ per $r_1$: then $\grain{R_1}$ determines the matched $\grain{R_2}$, so the product of
grains is reducible and the rule over-reports. The common case in practice is a range join against a
non-overlapping validity interval, a SQL \textnormal{BETWEEN}, as when
each order is priced by the rate row whose validity period contains its
date, which is functional without being an equality. Such joins are
equi-joins in disguise: extend each input with the matching key, then
apply Theorem~\ref{thm:equijoin-grain}, which yields the correct coarser
grain (Section~\ref{subsec:generalized-join}). The same condition
appears in Theorem~\ref{thm:equijoin-grain} itself: an equi-join's
candidate grain is irreducible only when the inputs carry no
determination across them outside the join key.

These rules compose into the CalcG verification algorithm
(Section~\ref{sec:calcg}).

\section{The Grain Homomorphism and Design-Time Verification}
\label{sec:grain-factorization}

This section establishes the semantic foundation for design-time
pipeline reasoning. Proposition~\ref{prop:grain-factorization-codomain}
showed that every transformation~$h : R_1 \to R_2$ factors through
$\grain{R_2}$ via some $e : R_1 \to \grain{R_2}$. Here we refine~$e$
by introducing the \emph{grain lift} $\gl{h}$, the grain-to-grain
content of~$h$, and prove that grain projection commutes with
transformation (Theorem~\ref{thm:grain-factorization}). The grain lift
acts as a \emph{semantic function}: it maps each concrete data operation
to its abstract grain-level meaning, discarding all payload attributes.
For relational algebra operations, the grain lift is concretely computed
by the inference rules of Table~\ref{tab:ra-inference}
(Section~\ref{sec:inference}). Compositionality
(Corollary~\ref{cor:grain-compositionality}) follows: grain lifts
chain correctly across pipeline stages, so the verification of an
$n$-step pipeline reduces to $n$ independent grain-level checks.

\noindent
Proofs for this section are in Appendix~\ref{sec:proofs}.

\subsection{The Grain Homomorphism}
\label{subsec:grain-factorization}

\begin{definition}[Grain Lift]
\label{def:grain-lift}
For any function $h : R_1 \to R_2$, the \emph{grain lift} of~$h$ is
\[
  \gl{h} \;=\; \textit{grain}_{R_2} \circ h \circ f_{g_{R_1}}
  \;:\; \grain{R_1} \to \grain{R_2}
\]
where $\textit{grain}_{R_i} = f_{g_{R_i}}^{-1} : R_i \to \grain{R_i}$
is the grain projection (inverse of the grain function).
The grain lift is always constructible as the composition of three
total functions. It captures the grain-to-grain content of~$h$, discarding
payload attributes and implementation details. For relational algebra
operations, the concrete form of~$\gl{h}$ is given by the inference
rules of Table~\ref{tab:ra-inference}.
\end{definition}

\noindent
The lift carries the grain-level content of~$h$ and nothing else, and it is
what the inference rules of Section~\ref{sec:inference} compute, operation
by operation.

\noindent
The grain lift refines the factorization of
Proposition~\ref{prop:grain-factorization-codomain}: the unique factor
$e = \textit{grain}_{R_2} \circ h$ itself decomposes as
$e = \gl{h} \circ \textit{grain}_{R_1}$, giving the three-step
decomposition
\[
  h \;=\; f_{g_{R_2}} \circ \gl{h} \circ \textit{grain}_{R_1}.
\]
And as in the proposition, $\gl{h}$ is the \emph{unique} map
$\grain{R_1} \to \grain{R_2}$ for which this holds: both grain functions are
isomorphisms, so composing on either side cancels. Hence $h \mapsto \gl{h}$ is
a bijection from the set of transformations $R_1 \to R_2$ onto the set of
grain-level maps $\grain{R_1} \to \grain{R_2}$. That is the precise sense in
which the
denotation is \emph{faithful}: nothing about~$h$ is lost in passing to
$\gl{h}$, and every grain-level map is the lift of exactly one
transformation.

\begin{theorem}[Grain Homomorphism]
\label{thm:grain-factorization}
For any function $h : R_1 \to R_2$, grain projection commutes with
transformation:
\[
  \textit{grain}_{R_2} \circ h
  \;=\;
  \gl{h} \circ \textit{grain}_{R_1}
\]
That is, transforming and then projecting to grain yields the same
result as projecting to grain and then applying the grain lift.
\end{theorem}

\noindent
Figure~\ref{fig:grain-homomorphism} draws it.

\noindent
The commutativity is by construction (one-line proof in
Appendix~\ref{sec:proofs}); the grain lift is built to make
grain-level reasoning \emph{faithful} to the full-type level. The
substance lies in the concrete instantiation: the inference rules of
Section~\ref{sec:inference} give the explicit grain-level denotation
of each RA operation, making design-time reasoning feasible.

\begin{figure}[t]
\centering
\begin{tikzcd}[column sep=huge, row sep=large]
  \grain{R_1} \arrow[r, "\gl{h}"]
  & \grain{R_2} \\
  R_1 \arrow[u, "\textit{grain}_{R_1}"]
      \arrow[r, "h"']
  & R_2 \arrow[u, "\textit{grain}_{R_2}"']
\end{tikzcd}
\caption{The grain homomorphism: grain projection commutes with
  transformation. The bottom row operates on full types; the top row
  operates on grains. The two paths from $R_1$ to $\grain{R_2}$
  are equal: $\textit{grain}_{R_2} \circ h = \gl{h} \circ
  \textit{grain}_{R_1}$.}
\label{fig:grain-homomorphism}
\end{figure}

\begin{remark}[Connection to Grain Ordering]
\label{rem:factorization-ordering}
$h$ is surjective if and only if $\gl{h}$ is surjective (since
$\textit{grain}_{R_2}$ and $f_{g_{R_1}}$ are bijections).
Therefore, $R_1 \leg R_2$ precisely when the grain
lift~$\gl{h}$ is a surjection.
\end{remark}

\begin{corollary}[Grain Compositionality]
\label{cor:grain-compositionality}
For composable transformations $h_1 : R_1 \to R_2$ and
$h_2 : R_2 \to R_3$, the grain lifts compose:
\[
  \gl{h_2 \circ h_1} = \gl{h_2} \circ \gl{h_1}
\]
This is the algebraic homomorphism condition: the grain lift of a
composed transformation equals the composition of the individual
grain lifts.
\end{corollary}

\noindent
Figure~\ref{fig:grain-compositionality} draws it.

\begin{figure}[t]
\centering
\begin{tikzcd}[column sep=large, row sep=large]
  \grain{R_1}
    \arrow[r, "\gl{h_1}"]
    \arrow[rr, bend left=25, "\gl{h_2 \circ h_1}" description]
  & \grain{R_2}
    \arrow[r, "\gl{h_2}"]
  & \grain{R_3} \\
  R_1 \arrow[u, "\textit{grain}_{R_1}"]
      \arrow[r, "h_1"']
  & R_2 \arrow[u, "\textit{grain}_{R_2}"]
        \arrow[r, "h_2"']
  & R_3 \arrow[u, "\textit{grain}_{R_3}"']
\end{tikzcd}
\caption{Grain compositionality: the top row composes independently of
  the bottom row. The grain lift of the composition equals the
  composition of the grain lifts:
  $\gl{h_2 \circ h_1} = \gl{h_2} \circ \gl{h_1}$. Data engineers can
  reason about the entire pipeline at the grain level (top row) without
  touching the full types (bottom row).}
\label{fig:grain-compositionality}
\end{figure}

\noindent
Together, the grain lift, the homomorphism theorem and compositionality are
what make a pipeline's grain behaviour computable before the pipeline is
built: each operation's effect on the level of detail is given by the
inference rules of Section~\ref{sec:inference}, and the homomorphism is why
those per-operation results may be composed along the DAG, which is what
CalcG does (Section~\ref{sec:calcg}).

\begin{remark}[Operation Independence]
\label{rem:operation-independence}
The grain homomorphism holds for \emph{any} function
$h : R_1 \to R_2$, not just relational algebra. Table~\ref{tab:ra-inference}
instantiates it for RA operations; any algebra with computable grain
inference rules supports the same design-time verification. The
methodological value is that after each operation, the type-level
denotation $(\grain{R},\, \ek{R},\, \BC{R})$ is recomputable from the
schema and its declarations alone, with no data access, making pipeline
correctness verifiable step by step.
\end{remark}

\subsection{Grain as a Correctness Constraint for Data Transformations}
\label{subsec:grain-correctness-constraint}

Grain relations serve as \emph{mapping invariants} for data
transformations: a transformation targeting a collection with declared
grain~$\grain{R_2}$ is correct only if its output grain matches the
target's declared grain. The CalcG algorithm
(Section~\ref{sec:calcg}) decides this condition entirely at the type
level, without accessing any data.

This type-level correctness constraint is what makes grain theory
useful for data engineering. Rather than discovering grain
violations at runtime, after data has been processed, loaded, and
potentially consumed downstream, grain violations are caught at
design time through schema analysis alone.
 
\section{The CalcG Verification Algorithm}
\label{sec:calcg}

The inference rules of Section~\ref{sec:inference} give the grain of a
\emph{single} operation; this section composes them into \emph{CalcG}, an
algorithm that propagates grain across an entire pipeline DAG and checks
the output against the declared target grain, a data-independent proof of
grain-correctness computed from structure alone, before any query is
executed.

\noindent
Proofs for this section are in Appendix~\ref{sec:proofs};
Theorem~\ref{thm:calcg-correctness} is sketched in place.

\begin{definition}[Pipeline]
\label{def:pipeline}
A \emph{pipeline} is a finite directed acyclic graph $P = (V, E)$ whose
vertices are \emph{operations}, each with a computable grain-inference rule,
and whose edges carry the output of one as an input of another. Only two
grains are \emph{declared}: those of the \emph{source} vertices, fixing the
inputs, and that of the \emph{target} sink, which the pipeline is verified
against. Every intermediate grain is \emph{computed}.
\end{definition}

\noindent
This paper instantiates the relational algebra
(Section~\ref{sec:inference}); the general operation-set case is the
companion's~\cite{pddpaper}.

\begin{definition}[CalcG]
\label{def:calcg}
Given a pipeline $P = (V, E)$ (Definition~\ref{def:pipeline}),
\emph{CalcG} traverses $P$ in topological order (every operation is
processed after all the operations it depends on). At each vertex~$v$
with operation~$\textit{op}_v$ and input grains
$\grain{R_1}, \ldots, \grain{R_k}$ from predecessor nodes, CalcG
applies the operation's grain-inference rule, which for the relational
algebra is one of the rules of Table~\ref{tab:ra-inference}, to compute
the \emph{inferred} grain $\grain{Res_v}$.
For equi-join vertices with join key~$J_k$, CalcG computes the
$J_k$-grain-portions $G_i^{J_k} = \grain{R_i} \tin J_k$. When these
are $\ssub$-incomparable, either labeling yields the correct grain
(Theorem~\ref{thm:equijoin-grain}); CalcG accepts both candidates.
When they are $\ssub$-comparable, CalcG labels $R_1$ as the input with
the smaller portion ($G_1^{J_k} \ssub G_2^{J_k}$) to produce the
correct (and minimal) grain; only this labeling yields an
irreducible candidate (Theorem~\ref{thm:equijoin-grain}).
Starting from the declared source grains, CalcG computes
$\grain{Res_v}$ at every vertex; the pipeline is \emph{grain-correct}
when the grain it computes at the target equals the target's declared
grain. Only the source and target grains are declared; no intermediate annotation
is required (Corollary~\ref{cor:data-independent}).
\end{definition}

\noindent
For sequential composition, CalcG is function composition:
\[
\text{CalcG}[\textit{op}_1 ; \textit{op}_2](G_{\textit{in}})
\;=\;
\text{CalcG}[\textit{op}_2]\!\bigl(\text{CalcG}[\textit{op}_1](G_{\textit{in}})\bigr).
\]
For fan-out vertices (one output feeding multiple successors), the
computed grain is reused unchanged. For fan-in vertices (binary
operators: joins, set operations), the corresponding inference rule
(Table~\ref{tab:ra-inference}) computes the output grain from the
input grains of both predecessors. Sequential composition, fan-out,
and fan-in are the only vertex patterns in a DAG, so topological
traversal with sound inference rules covers all pipeline topologies.

\begin{theorem}[CalcG Correctness and Complexity]
\label{thm:calcg-correctness}
CalcG correctly computes the grain of every intermediate and final
result in a pipeline DAG, is decidable, and runs in time
$O(|V| \cdot k \cdot w \;+\; d \cdot \|\Delta\|)$ and space
$O(|V| \cdot w)$, where $k$ is the maximum number of inputs to any
operation (typically $k \leq 2$), $w$ is the greatest \emph{width} of a
type, meaning the number of fields it carries, $\Delta$ is the set of
declared determinations and $\|\Delta\|$ the total size of the pairs in
it, and $d$ is the number of types on which a
determination is declared. Charged per vertex instead of per declared
type, the second term is $O(|V| \cdot \|\Delta\|)$. In either form the
running time is linear in the pipeline.
\end{theorem}

\noindent
Correctness is relative to the rules' side conditions: CalcG computes the
grain at every vertex whose rule applies, and flags any vertex where that
rule's condition fails.

\noindent\textit{Proof sketch.}\quad
Correctness: induction on the topological ordering of~$V$, using the
soundness of each inference rule (Table~\ref{tab:ra-inference}).
Decidability: every rule involves only finite type-level set operations
($\tin$, $\tdiff$, $\tun$, $\ssub$) on field sets. Complexity: one pass
over the DAG with $O(k \cdot w)$ field-set work per vertex, plus the
attribute closure that the grouping rule and the grain operator's product
clause run over the declared determinations, which is linear in
$\|\Delta\|$ rather than $O(w)$ and is charged once per declared type
(Section~\ref{sec:calcg}'s closure accounting below), and
$O(|V| \cdot w)$ storage. Full proof in Appendix~\ref{sec:proofs}.

\noindent
CalcG is not tied to the relational algebra: it needs only a
\emph{computable} per-operation grain-inference function over a
finite pipeline DAG. Table~\ref{tab:ra-inference} provides this
function for the relational algebra; windowed collections
(Section~\ref{sec:adt}) provide it for streaming pipelines.

\begin{corollary}[Data-Independent Verification]
\label{cor:data-independent}
A pipeline from source type~$R_S$ to target type~$R_T$ is grain-correct
if $\emph{CalcG}[P](\grain{R_S}) = \grain{R_T}$. This check requires
no data access, since it operates entirely on the type-level schema and
its declared determinations (foreign keys and declared grains), which
are metadata, not data.
\end{corollary}

\noindent
Traditional grain verification materializes intermediate results and
checks uniqueness, which is $O(n)$ in data size. CalcG achieves the same
guarantee in time linear in the pipeline
(Theorem~\ref{thm:calcg-correctness}), on the schema alone and
independent of data size, suitable for compile-time integration.

\begin{remark}[Where source grains come from]
\label{rem:source-grains}
Only the source and target grains are declared
(Definition~\ref{def:pipeline}), and this is realistic, not a burden. A
grain is not a schema-derivable quantity to be discovered but a
\emph{semantic declaration}, a lens fixing what each element represents
(Section~\ref{subsec:grain-semantics}); the same type admits different
grains, so no procedure recovers the intended one from the schema alone.
Establishing a source grain is thus a \emph{semantic} act: from domain
knowledge the engineer declares what each element represents, its key,
entity, and behavioral class (Sections~\ref{sec:foundations},~\ref{sec:entity}).
Source catalogs seldom settle it: integrity constraints are often
missing, or, in modern analytical warehouses such as Snowflake and
Databricks, declared but not enforced. The declaration is therefore
\emph{confirmed} once by data-source analysis (profiling the input to
check the declared grain holds), a one-time data-quality control at the
input boundary. This boundary check is the only point data is consulted;
CalcG then propagates the declared source grain through the pipeline with
no further data access
(Corollary~\ref{cor:data-independent}).\footnote{The obstacle is
semantic, not computational: the same type, indeed the same
instance, admits several candidate grains, so no procedure over schema
or data \emph{determines} the intended one; profiling can confirm or
refute a declaration but not supply it. (Finding a \emph{minimum} key is
another matter: deciding whether a candidate key of at most a given size
exists is NP-complete~\cite{lucchesi1978candidate}, whereas a
\emph{minimal} key, one no proper subset of which is a key, costs
polynomially many attribute closures,
Section~\ref{subsec:compositional-grain}; both are beside the point here.)}
\end{remark}

The comparison CalcG makes at each vertex, computed grain against
declared target, has a classical cost on product types. Grain assertions
and entailment between them are those of
Definition~\ref{def:grain-holds}: a set of assertions entails another
when the latter holds under every declaration set under which the former
hold, the type-level form of FD implication~\cite{armstrong1974dependency}.

\begin{corollary}[Grain Entailment on Product Types is FD Implication]
\label{cor:grain-implication}
Let the attribute universe $\mathcal{A} = \bigcup_R \mathrm{comps}(R)$
be finite, as it is on products. Deciding whether a finite set of grain
assertions entails a target assertion then coincides with
functional-dependency implication: it reduces in linear time to FD
implication (Proposition~\ref{prop:grain-completeness}) and is decided
by the attribute-closure algorithm of~\cite{beeri1979computational}
(reduction in Appendix~\ref{sec:proofs}).
\end{corollary}

\noindent
That cost is charged less often than the corollary alone suggests, and
this is the practical reading of the closure term in
Theorem~\ref{thm:calcg-correctness}.
At a vertex CalcG applies the operation's grain-inference rule to the
predecessors' grains and then reduces the result by any determinations
declared on it. Determinations are declared on the source and the
target and, in a pipeline of any size, essentially nowhere between, so
at an intermediate vertex the closure has nothing to close over and
does nothing: the rule does the work. A run therefore costs one rule
application per vertex, with an attribute closure only where a
declaration sits, which is at the two ends.

\noindent
This product-type correspondence is all we claim. The general
\emph{grain implication problem}, entailment among grain assertions over
sum, inductive and coinductive types, where the functional-dependency
reduction no longer applies, is not treated in this paper; we leave it as
future work (Section~\ref{sec:conclusion}).

\begin{remark}[CalcG as a diagnostician]
\label{rem:calcg-diagnostician}
As defined, CalcG returns a boolean. At the sink it asks whether the
computed grain equals the declared target
(Definition~\ref{def:calcg}), and reports that it does or does not. For
the same work it can report why.

Replace the equality test by an order query. Exactly one of the four
relations of Proposition~\ref{prop:grain-error-classification} holds
between the computed grain and the target, and which one holds names the
failure: strictly finer is duplication, strictly coarser is detail lost,
incomparable is a re-graining. Add the first vertex at which the running
grain left the reference (Corollary~\ref{cor:entry-localization}), and a
run that today reports \emph{false} reports instead \emph{duplication,
entering at the join at vertex~$v$}.

Neither addition costs anything. Deciding $\leg$ runs the same field-set
containment test that deciding $\eqg$ already runs, and the per-vertex
grains the localization needs were computed during the traversal and are
still in hand. Verification is therefore the $\eqg$ case of the
classification, the error catalogue is the other three, and one traversal
settles all four.

This is also what answers the objection that grain-correctness is merely
defined to be whatever CalcG's equation says. Each way of failing that
equation has its own meaning in the data, and a definition by algorithm
alone could not supply that.
\end{remark}

\begin{remark}[Pipeline Repair]
\label{rem:pipeline-synthesis}
Beyond verification, pipeline \emph{repair} is open: given a pipeline
whose inferred grain~$G$ differs from the declared target grain
$\grain{R_T}$, find a minimal rewrite that restores grain-correctness.
The required rewrite depends on the grain relationship: coarsening
(e.g., pre-aggregation) when $G \stg \grain{R_T}$, refinement (joining
with a finer-grained source) when $\grain{R_T} \stg G$, or alignment to
a common grain when $G \incg \grain{R_T}$. Repair is distinct from the
correct-by-construction \emph{synthesis} of the companion
paper~\cite{pddpaper}, which builds pipelines that need no repair; the
minimal-rewrite problem is left as future work.
\end{remark}

\begin{remark}[Practical Cost Implications]
\label{rem:data-independent-cost}
Data-independent verification is not merely an asymptotic improvement.
In pay-per-query cloud data warehouses (BigQuery, Snowflake, Databricks),
costs are proportional to the volume of data scanned. Since CalcG
operates entirely on type-level schema metadata without issuing any data
query, the monetary cost of verification is literally zero: no data scanned means
no charges incurred.
\end{remark}
 
\section{Applications: Detecting Data Transformation Errors}
\label{sec:errors}

Grain theory exposes a spectrum of transformation errors, most of them
decidable from the schema alone, and delimits what it does \emph{not} model.

Proposition~\ref{prop:fan-trap} is proved in Appendix~\ref{sec:proofs}. The
others carry their argument within the statement: each reads a
definition or an earlier theorem against the grain the operation produces,
and the reading is the proof. They are propositions because they are the
checks CalcG performs, not because an argument is held elsewhere.

The catalogue is not a list. Each error is one comparison, of a grain the
calculus computes against a grain the design intends, and the position of
the one relative to the other in the grain ordering is what names the
failure. The classification is stated once; the errors below instantiate
it.

\begin{proposition}[Classification of Grain Errors]
\label{prop:grain-error-classification}
Let $C$ be the grain CalcG computes at a vertex or across a pipeline
segment, and let $D$ be a \emph{reference grain}, the yardstick the
check compares against: a declared target, the operation's input
grain, or the segment's entry grain. Exactly one of four
relations holds between them (antisymmetry,
Theorem~\ref{thm:grain-partial-order}, with
Definition~\ref{def:grain-inc}), and each fixes the failure at the data
level. $C \eqg D$ is grain-correctness with respect to~$D$. $C \stg D$
puts several rows on one $D$-value: \emph{duplication}. $D \stg C$
spreads one row over several $D$-values: \emph{detail lost}, which no
downstream operation recovers. $C \incg D$ is a \emph{re-graining}: the
result makes distinctions $D$ does not and misses distinctions $D$
makes.
\end{proposition}

\noindent
The reference fixes where the check runs: an input or entry grain is
computed, so a violation against it is decided at the vertex that
creates it and cannot be masked, while a declared target is tested
where the declaration lives. What a compensating downstream rewrite
can hide from the sink check is therefore not the level of detail but
the value of an aggregate computed at a grain its function does not
tolerate, the summarizability question (Section~\ref{sec:related});
a declaration at that vertex restores the check.

The segment reference is the incremental staple: a delta extracted at
$\grain{\textit{Delta}}$, enriched and augmented, then merged into its
target on that key; an output grain strictly finer than the entry grain
sends the same key to the merge more than once, the notorious
\emph{unstable merge}, known from the schema alone before the pipeline exists. The operation fixes which case it can produce: joins refine,
aggregation is the sole legitimate coarsening, and a projection that
drops a grain component realizes no coarsening at all, breaking grain
uniqueness instead (Proposition~\ref{prop:projection-dup}). The
segment case is developed in Section~\ref{subsec:entry-grain}, and the
catalogue that follows instantiates the classification vertex by vertex.

\subsection{Segments and the Entry Grain}
\label{subsec:entry-grain}

A \emph{segment} of a pipeline is a sub-DAG with one entry collection
$\textit{In}$ and one exit collection $\textit{Out}$, and CalcG computes
$\grain{\textit{Out}}$ from $\grain{\textit{In}}$ by composing the rules
along it (Definition~\ref{def:calcg}). Taking $\grain{\textit{In}}$ as the
reference grain of Proposition~\ref{prop:grain-error-classification}
requires no new annotation, the entry grain being declared at the source or
else already computed upstream, and it turns the classification into a
statement about a whole stretch of pipeline rather than about one vertex.

Three things the results below all turn on are worth fixing first.

\smallskip\noindent\textbf{When a rule applies.}\quad
Several rules of Table~\ref{tab:ra-inference} carry a \emph{side
condition}, a hypothesis under which alone the grain they state is the
right answer. The projection rule needs $S$ to determine~$R$ (the closure
test of Proposition~\ref{prop:projection-dup}), the cross-
and $\theta$-join rules need a non-functional join predicate, and the
equi-join theorem needs its grain-level independence hypothesis. We say a
vertex's rule \emph{applies} when its side condition holds. A vertex whose
condition fails is not skipped by CalcG but flagged
(Section~\ref{sec:calcg}), and Corollary~\ref{cor:safe-fragment} is where
that turns out to matter.

\smallskip\noindent\textbf{What duplicate-freedom means.}\quad
A segment is \emph{duplicate-free with respect to its entry grain} when
$\textit{Out}$ holds at most one row per value of $\grain{\textit{In}}$.
That is the property a merge keyed on the entry grain requires, and the
property the unstable merge lacks.

\smallskip\noindent\textbf{How the rules move the grain.}\quad
Read down the result column of Table~\ref{tab:ra-inference}. Selection,
grain-preserving projection, extension, rename, the set operations and the
semi- and anti-joins return an input grain unchanged. Grouping returns a
coarser one. Only the join rules can return a grain strictly finer than an
input's, so a join is the only operation that refines.

\begin{corollary}[Entry-Grain Duplication, Loss and Re-graining]
\label{cor:entry-grain}
Let a segment run from $\textit{In}$ to $\textit{Out}$, every vertex's rule
applying. Read Proposition~\ref{prop:grain-error-classification} with
$C = \grain{\textit{Out}}$ and $D = \grain{\textit{In}}$. If
$\textit{Out} \stg \textit{In}$, the types license several exit rows per
value of $\grain{\textit{In}}$, so the segment may duplicate with respect
to its entry. If $\textit{In} \stg \textit{Out}$, entry rows are merged and
detail is lost. If $\textit{Out} \incg \textit{In}$, the segment re-grains.
Each case is decided from the schema alone, before the segment is built.
\end{corollary}

\noindent
The reading is a type-level \emph{may}, as it is for the fan trap: it
records what the grains permit and claims nothing about a particular
instance. Its everyday case is the one named above. A delta extracted at
$\grain{\textit{Delta}}$, enriched along the segment and merged into its
target on that same key, arrives finer than it left exactly when
$\textit{Out} \stg \textit{In}$, and the merge then meets the same key
twice.

\begin{corollary}[Localization of Entry-Grain Duplication]
\label{cor:entry-localization}
Under the hypotheses of Corollary~\ref{cor:entry-grain}, suppose
$\textit{Out} \stg \textit{In}$. Then the first vertex in topological order
at which CalcG's running grain falls strictly below $\grain{\textit{In}}$
is a join, a join being the only operation that refines, and it is the
vertex where the duplication enters.
\end{corollary}

\noindent
So CalcG locates the defect as well as detecting it, and at no added cost,
the running grains having been computed already.

\begin{corollary}[A Duplicate-Free Fragment]
\label{cor:safe-fragment}
If no vertex of a segment returns a grain strictly finer than its running
grain, \emph{and} every vertex's rule applies, then the segment is
duplicate-free with respect to its entry grain.
\end{corollary}

\noindent
The second hypothesis is load-bearing rather than a formality, because the
two ways a segment duplicates differ in kind. Join duplication is
refinement, and refinement is visible in the composed grain, so the
end-to-end comparison of Corollary~\ref{cor:entry-grain} catches it.
Projection duplication is not refinement at all. Dropping a grain component
coarsens nothing and breaks grain uniqueness instead
(Proposition~\ref{prop:projection-dup}), and the composed grain is
\emph{blind} to it, the result having no well-defined grain left to
compare. What reveals it is the closure test failing at
that vertex, which is why a failed condition has to be flagged rather than
passed over.

The condition is sufficient and not necessary, and that is what makes it a
fragment rather than a characterization. A join that refines, followed by
an aggregation that coarsens back to the entry grain, is duplicate-free
with respect to the entry and yet contains a refining vertex. Moving the
grain about is normal in a pipeline of any size. What
Corollary~\ref{cor:entry-grain} reports is where the grain ends up, and
what Corollary~\ref{cor:safe-fragment} certifies is a segment that never
moved it down at all.

\subsection{Fan Traps}
\label{subsec:fan-traps}

\begin{proposition}[Fan Trap Characterization]
\label{prop:fan-trap}
Given collections $C\;R_1$ and $C\;R_2$ joined on~$J_k$ producing
$C\;Res$, a \emph{fan trap} occurs when $Res \stg R_i$
for some~$i$: the result grain is strictly finer than an input grain,
meaning each row from~$R_i$ is duplicated in~$Res$, inflating every
aggregate computed over~$R_i$.
\end{proposition}

\noindent\textbf{Detection.}\quad
Apply Theorem~\ref{thm:equijoin-grain} to compute $\grain{Res}$; check
$Res \stg R_i$. This is a data-independent, schema-only check
(Corollary~\ref{cor:data-independent}).

\smallskip\noindent\textbf{Example: CalcG on the Section~\ref{sec:intro}
motivating query.}\quad
Recall $\grain{\texttt{SalesChannel}} =
\textit{CustomerId} \times \textit{ChannelId} \times \textit{Date}$ and
$\grain{\texttt{SalesProduct}} =
\textit{CustomerId} \times \textit{ProductId} \times \textit{Date}$, with join key
$J_k = \textit{CustomerId} \times \textit{Date}$.
\begin{enumerate}
\item Apply Theorem~\ref{thm:equijoin-grain}:
  $\grain{Res} = \grain{R_1} \tun (\grain{R_2} \tdiff J_k)$.
\item $\grain{R_2} \tdiff J_k = \textit{ProductId}$, so
  $\grain{Res} = (\textit{CustomerId} \times \textit{ChannelId} \times \textit{Date})
  \tun \textit{ProductId}
  = \textit{CustomerId} \times \textit{ChannelId} \times \textit{Date}
  \times \textit{ProductId}$.
\item $Res \stg R_1$ and $Res \stg R_2$, so a \textbf{fan trap is
  detected}, before any data is processed.
\end{enumerate}

\noindent\textbf{Prevention.}\quad
Pre-aggregate each input to the target grain before joining: if
$\grain{R_i'} = \grain{\textit{Target}}$, then
$Res \not\stg R_i'$.

\subsection{Further Grain Errors}
\label{subsec:more-errors}

The fan trap is one case of a single principle: a grain-changing
operation is an error when the grain it produces differs from the one
intended. Four further cases follow, each a schema-only check.

\begin{proposition}[Projection Breaks Grain Uniqueness]
\label{prop:projection-dup}
A projection of $C\,R$ onto a field set $S$ preserves the grain iff $S$
contains a grain of~$R$, some $G' \in \textit{Grains}(R)$ with
$G' \ssub S$; equivalently, iff $S$ determines~$R$,
$\mathrm{comps}(R) \subseteq \mathrm{comps}(S)^{+}$
(Lemma~\ref{lem:sat-closure}). Then $\grain{Res} = G'$, an alternative
grain of~$R$ (Theorem~\ref{thm:multiple-grains}). If $S$ contains no
grain of~$R$, the projection is not injective on the collection: two
valid rows agree on~$S$ and differ elsewhere, so the result holds several
rows per grain value and
\emph{violates the grain-uniqueness invariant}, its grain projection is
no longer injective (Theorem~\ref{thm:grain-uniqueness}). It becomes a
superkey-carrying collection only once the duplicates are removed (by
\textnormal{DISTINCT} or re-aggregation). For instance, projecting away
$\textit{LineItemNo}$ from a collection of grain
$\textit{OrderId} \times \textit{LineItemNo}$ yields multiple rows per
$\textit{OrderId}$ under a nominal grain of $\textit{OrderId}$.
\end{proposition}

\noindent
This is a \emph{latent} fan trap rather than one visible from the schema.
Joining the duplicate-carrying result on its nominal grain inflates the
output, but the inflation comes from the violated uniqueness invariant
\emph{in the data}, and the equi-join rule assumes grain-unique inputs, so it
reports the join well-formed. The error must be caught at the projection,
where the closure test is schema-decidable, not at the downstream
join. When $S$ drops the grain the schema named but retains an alternative
grain, the test passes and that grain becomes the result's; nothing is
lost, the two being isomorphic.

\begin{proposition}[Ungrounded Multi-Version Read]
\label{prop:mv-read}
Let $C\,R$ be \textbf{IsMultiVersion} with
$\grain{R} = \ek{R} \times \textit{FromDtm}$, joined on $\ek{R}$ to a
collection of grain $\ek{R}$. Since $\grain{R} \stg \ek{R}$, the join is
a fan trap (Proposition~\ref{prop:fan-trap}): the other side is
duplicated once per version. The point-in-time read prescribed for
\textbf{IsMultiVersion} (Section~\ref{sec:entity}) first reduces $C\,R$
to grain $\ek{R}$ (one version per entity), removing it; omitting it is
a behavioral-class violation caught at the type level.
\end{proposition}

\smallskip\noindent\textbf{Example: the multi-version \emph{level} (a
real pipeline bug, anonymized).}\quad
A source emits, each day, the full current state of every order: an
extract lists all of an order's lines, and a line removed upstream
simply stops appearing in later extracts. The staging collection is
flattened to lines $(\textit{OrderId}, \textit{ProductId},
\textit{ExtractId})$ and, to keep the grain a superkey, declared
\textbf{IsEvent} with grain $\ek{} \times \textit{ExtractId}$
($\ek{} = \textit{OrderId} \times \textit{ProductId}$). Its
\emph{essential} class, however, is \textbf{IsMultiVersion} at the
\emph{order} grain $\textit{OrderId} \times \textit{ExtractId}$: the
whole order is versioned per extract. Reading it as if versioned at the
\emph{line} grain, ``latest extract per line'', keeps a line present
in an old extract but absent from the order's newest one, silently
resurrecting a deleted line; the read prescribed by
\textbf{IsMultiVersion} at the order grain instead takes every line of
the order's latest extract, so the missing line is flagged deleted. For
an order whose lines are $p_1, p_2$ in extract $E_1$ and only $p_2$ in
the newer $E_2$, the per-line read keeps $p_1$ active while the
per-order read correctly deletes it. The defect is a mis-declared
behavioral-class \emph{level}; fixing the staging table's denotation
$(\grain{}, \ek{}, \BC{})$ exposes it at design time, before any data
is processed. The canonical \emph{write} of a declared class is a second
guardrail. The original design declared the staging table
\textbf{IsMultiVersion}, yet loaded it by a plain append, which is the
canonical write of \textbf{IsEvent}. The canonical \textbf{IsMultiVersion}
write inserts a new version only when an entity's payload changes, so it
would have refused the unchanged lines. The resulting cardinality
mismatch against the input would have failed a test at load time,
surfacing the mis-declaration even before the denotation was
reconstructed.
Appendix~\ref{app:case-study} gives the full design evolution of this
pipeline.

The two remaining classes are instances of the general check
(Corollary~\ref{cor:data-independent}), which, unlike the two
structural characterizations above, needs a \emph{declared} target
grain or behavioral class.

\begin{proposition}[Wrong-Grain Aggregation]
\label{prop:wrong-grain-agg}
Let $C\,R$ be grain-unique, let $\gamma_{G_c}$ group it by columns
$G_c \ssub R$ with aggregation, producing $Res = G_c \times Agg$ with
$\grain{Res} = \grain{G_c}$ (Table~\ref{tab:ra-inference}), and let
\textit{Target} be the declared target. The aggregation is grain-correct
iff $Res \eqg \textit{Target}$, and the two failures are distinguished
from the schema alone:
\begin{enumerate}
\item $Res \stg \textit{Target}$ (grouping columns strictly finer): the
  result holds several rows per target grain value, so each aggregate
  covers only part of its intended group. The aggregation is incomplete,
  and joining the result at the target grain is a fan trap
  (Proposition~\ref{prop:fan-trap}).
\item $\textit{Target} \stg Res$ (grouping columns strictly coarser):
  rows the target distinguishes are merged into a single group, so each
  aggregate is computed over a superset of its intended group. Grouping
  is not injective on the input, so no downstream operation recovers the
  merged distinctions.
\end{enumerate}
\end{proposition}

\noindent
The hypothesis that $C\,R$ is grain-unique is load-bearing. If it is
violated upstream, by a fan trap or by a projection that drops a grain
component (Proposition~\ref{prop:projection-dup}), then additive
aggregates such as \textnormal{SUM} and \textnormal{COUNT} are inflated
by the duplication multiplicity even when $Res \eqg \textit{Target}$
holds and the grain check passes. As with projection, the error must be
caught where it originates, not at the aggregation.

\begin{proposition}[Behavioral-Class Violation]
\label{prop:bc-violation}
Let $C\,R$ be declared of behavioral class~$B$ with entity key $\ek{R}$,
and let $\gamma_B$ be the grain condition of~$B$
(Definition~\ref{def:behavioral-class}, Table~\ref{tab:bc-formal}). The
declaration is sound only if the computed grain satisfies
$\grain{R} = \gamma_B(\ek{R})$; when it does not, the declaration is
unsound and the canonical read and write of~$B$ are not applicable to
$C\,R$.
\end{proposition}

\noindent
Each class fixes $\grain{R}$ as an equation in $\ek{R}$ and a temporal
component and derives its canonical operations from that equation
(Section~\ref{sec:entity}), so a mismatch invalidates them in one of two
directions. A collection declared \textbf{IsEntity} whose computed grain is
$\ek{R} \times \textit{Dtm}$ holds several rows per entity key, so the direct
lookup the class licenses is not single-valued; one declared
\textbf{IsMultiVersion} whose computed grain is $\ek{R}$ carries no versions,
so the point-in-time read degenerates to the identity. Both are decided by
comparing the computed grain against the class equation, given the declared
$\ek{R}$ and~$B$. Proposition~\ref{prop:mv-read} is the special case where
the class is right and the \emph{level} wrong; fan traps and projection
duplication, by contrast, are detected structurally, with no declared
target.

\noindent
That difference fixes where each class can be caught. A fan trap and a
dropped grain field are decided at the vertex that creates them, from the
inferred grains alone. Wrong-grain aggregation and a behavioral-class
violation are decided only where a declaration stands against the computed
grain, at the sink or at an intermediate vertex carrying one
(Definition~\ref{def:pipeline}), so a downstream rewrite that compensates
for an intermediate coarsening leaves the sink check satisfied.

\subsection{Chasm Traps}
\label{subsec:chasm-traps}

\begin{remark}[Chasm Trap Vulnerability]
\label{rem:chasm-trap}
A grain ordering chain $R_1 \leg R_2 \leg \cdots \leg R_n$ is the
type-level structural pattern behind a join chain, typically a
snowflake schema where dimension hierarchies are normalized into
separate tables joined via foreign keys. At the type level the chain
is well-formed (transitivity gives $R_1 \leg R_n$).

A \emph{chasm trap} is a data-instance failure in which nullable
foreign keys cause intermediate rows to lack matches in the next
table, so inner joins silently discard them and lose data. Grain
theory operates at the type level and does not model data-instance
phenomena; its contribution is identifying the structural pattern
vulnerable to such loss. \textbf{Prevention:} use outer joins, or
ensure foreign keys point to a valid record (e.g., a designated
unknown/default entity).
\end{remark}

\subsection{Scope of Grain Correctness}
\label{subsec:scope}

Pipeline correctness spans three layers, each operating at a different
cost:

\smallskip\noindent\textbf{(1) Grain-level correctness (data-independent).}\quad
CalcG (Section~\ref{sec:calcg}) propagates grain through an entire
pipeline DAG, potentially hundreds or thousands of
operations, verifying that the output grain matches the declared target
grain. This is a \emph{mathematical proof}, not a test: it guarantees
the absence of grain errors for \emph{all} possible data inputs, derived
entirely from the pipeline's type-level structure without executing a
single query. Testing can demonstrate the presence of a bug on a
particular dataset; grain verification proves its absence universally.
Fan traps, grain mismatches, and incorrect aggregation granularity are
specific consequences of grain violations that this proof eliminates by
construction. Grain ordering chains identify structural patterns
vulnerable to chasm traps (Remark~\ref{rem:chasm-trap}).

\smallskip\noindent\textbf{(2) Behavioral class correctness (compile
time).}\quad
The behavioral classes of Section~\ref{sec:entity} enforce that
operations respect their declared read/write patterns. A
time-window capture on an entity collection, or a point-in-time query
on an event collection, are type errors. This layer is also automatic: the
type checker rejects violations before any code runs.

\smallskip\noindent\textbf{(3) Domain correctness (specification
time).}\quad
Business rules (referential integrity, domain predicates, specific
aggregation logic) require data-level verification and are outside the
scope of this paper. Extending grain theory with operation contracts
that derive additional constraints (collection subset, cardinality
bounds) automatically from pipeline structure is a natural next step
(Section~\ref{sec:conclusion}).
 
\section{Mechanized Verification}
\label{sec:verification}

The entire theory presented in this paper, Sections~\ref{sec:foundations}
through~\ref{sec:errors}, has been independently validated through two
complementary methodologies: mechanized theorem proving and property-based
testing.

\paragraph{Lean~4 Formalization.}
The theory has been formalized in Lean~4~\cite{demoura2021lean4} with the
Mathlib library, using an abstract axiomatization that mirrors the proof style
of this paper: data types and grain are opaque structures with assumed
properties, and all results are proved from these axioms. The formalization
comprises 46~modules organized by paper section, 356~machine-checked theorems
and lemmas, and zero \texttt{sorry} obligations (unfinished proofs). Its
assumptions are 55~axioms: the fields of the type classes that keep the type
relations, the grain operator, determination, the fixed-point constructors and
the denotation opaque, plus one standalone \texttt{axiom}, the completeness
direction of the Armstrong system (Proposition~\ref{prop:grain-completeness}),
which transfers directly from Armstrong's classical
result~\cite{armstrong1974dependency} rather than being re-derived.

Because the development is abstract, what a proof assistant checks is that the
results follow from the assumed properties, not that anything satisfies them.
That gap is closed by a concrete finite model: a three-attribute relational
schema with one declared functional dependency, in which types are attribute
sets, every relation is decidable, and every axiom of the core and equi-join
structures that attribute-set types can express is settled by decision
procedure rather than by a hand proof (all but four, those on sums and on
products as constructors). The model establishes that the axiom
base is consistent and that it says what it was intended to say. It also
earned its keep. Two axioms proved \emph{false} in it, both semantic twins of
$\subt$ introduced when $\subt$ was split from $\ssub$, and the reason is
instructive: determination does not distribute over set operations the way
containment does, so ``$T$ is determined by $R$ and by $S$'' does not give
``$T$ is determined by $R \tin S$''. Both were removed and their consumers
rerouted through the structural forms. Finally, the model witnesses two facts
the abstract development cannot state about itself: that $\ssub$ is strictly
stronger than $\subt$, so the degenerate reading identifying them is refuted,
and that irreducibility in the sense of Definition~\ref{def:grain} is not
vacuous, a padded candidate isomorphic to~$R$ being exhibited that is not a
grain of~$R$. That witness is precisely the content of the definitional
repair.

One caveat stands between the paper and the formalization, and it has
narrowed to a point. Of Theorem~\ref{thm:grain-adt}'s three claims,
idempotency and commutation with $+$, $\mu$, $\nu$ and~$\times$ are proved
outright, by structural induction over the type grammar with no axioms and no
semantics; the correctness claim is mechanized by the paper's induction, the
passage under binders included, but its $\mu$ and $\nu$ cases rest on two
hypotheses of the fixed-point structure: that isomorphic bodies have
isomorphic fixed points, and that the fixed points carry irreducibility
through. The first is the congruence the proof cites from Bird and de
Moor~\cite{birddemoor1997}. It is stated pointwise, because the opaque universe
carries no morphisms in which naturality could be expressed, so the hypothesis
is stronger than the categorical result it stands for; and the finite model,
having no recursive types, checks neither hypothesis. Also machine-checked are
Theorem~\ref{thm:type-grain-vs-collection-key}, down to the one-row collection
in which every projection identifies; the uniqueness of the grain factor and
the bijection $h \mapsto \gl{h}$
(Proposition~\ref{prop:grain-factorization-codomain}); the grain of a
partitioned collection (Proposition~\ref{prop:partitioned-grain}); the
reduction under declared dependencies
(Proposition~\ref{prop:grain-reduction-fd}); and the error catalogue of
Section~\ref{sec:errors}, whose fan-trap and chasm-trap patterns, projection
duplication, ungrounded multi-version read and wrong-grain aggregation each
have a Lean theorem.

\paragraph{Property-Based Testing.}
The equi-join grain inference theorem
(Theorem~\ref{thm:equijoin-grain}), the most complex result in
the paper, was additionally validated via property-based testing (PBT)
against a PostgreSQL implementation.
Three complementary modes were used:
(a)~exhaustive enumeration of all valid equi-join configurations up to
5~columns (12{,}271~configurations);
(b)~randomized testing via the Hypothesis framework~\cite{maciver2019hypothesis} at sizes up to
22~columns (10{,}000~configurations); and
(c)~16~hand-crafted edge cases under both uniform and skewed data
distributions.
Four properties were tested: uniqueness of both candidate formulas ($F_1$,
$F_2$) on join results, minimality (irreducibility) of the
convention-compliant grain, non-minimality of the non-convention formula
for comparable cases, and both-minimal for incomparable cases.
Across all 22{,}303~configurations and 73{,}818~individual property
checks, zero violations were found.
Independently, 100~worked examples spanning the relational-algebra
operations were validated in PostgreSQL~16 by materializing each result
and confirming the inferred grain as a uniqueness constraint, at 100\% formula
accuracy.

\paragraph{Agda Formalization.}
The core of grain theory
(Sections~\ref{sec:foundations}--\ref{sec:entity}) is independently
formalized in Agda: the grain as a relation between types (an
isomorphism witness together with an irreducibility predicate), the
entity and the entity key, and the three grain relations $\eqg$,
$\leg$, $\incg$ with their Armstrong-style axioms, each rendered as
an Agda relation or record carrying its function witnesses.
The nine behavioral classes of Section~\ref{subsec:behavioral-semantics}, the
seven core and the relational
\textbf{HasParent}/\textbf{IsNMRelation}, are likewise
mechanized as record types (type classes) with explicit
\texttt{grain-cond} fields where applicable (\textbf{HasParent} carries a
relation rather than a grain condition); the implementation confirms that
\textit{EventDtm}, \textit{FromDtm}, and \textit{SnapshotDtm} are
generic type parameters constrained to carry a total or strict total
order, and that the subclass hierarchy
(\textbf{IsMultiVersion}~$\subset$~\textbf{IsEvent},
\textbf{IsSnapshot}~$\subset$~\textbf{IsSeqEvent}) is enforced via
superclass instance fields. The full Agda development is available in
the artifact repository.

\paragraph{Use of AI.}
Large language models (Claude, Opus~4.6, 5 and~5.5, and Fable~5.1) wrote the Lean~4 and Agda
developments and the property-based testing harness under the author's
direction, commit by commit, and assisted in drafting the prose of this
paper. What certifies the formal artifacts is machine checking rather than
authorship. Every Lean proof is checked by Lean~4 with no \texttt{sorry}, so
the 356~theorems above stand or fall with the kernel and not with who typed
them; the Agda development type-checks, its examples postulating base types
and isomorphism witnesses; and every property-based result is reproducible
from the artifact repository, the 22{,}303~equi-join configurations included.
The two axioms the finite model refuted are the case in point: what caught
them was the decision procedure, not a reading of the code.

What machine checking does not reach is named above. The completeness
direction of the Armstrong system
(Proposition~\ref{prop:grain-completeness}) rests on a paper proof. The
two fixed-point hypotheses of Theorem~\ref{thm:grain-adt} are stated
hypotheses, assumed rather than proved, and the fixed-point case of
antisymmetry (Theorem~\ref{thm:grain-partial-order}) is proved only for
witnesses induced from the bodies. The paper proofs, like the
definitions and theorem statements throughout, are the author's.

\paragraph{Artifacts.}
The Lean~4 proofs, the Agda formalization, the PBT framework, and all
experimental artifacts are publicly
available.\footnote{Artifact repository: \url{https://github.com/nkarag/grain-theory-artifacts}.}
 
\section{Related Work}
\label{sec:related}

Grain theory connects to several established research areas.

\subsection{Dimensional Modeling and Grain}
\label{subsec:rw-dimensional}

Kimball~\cite{kimball1996data,kimball2013dw} introduced grain as the
fundamental design decision for fact tables in star schemas, but
Kimball's grain is informal (prose, not math), narrow (fact tables in
star schemas), and non-systematic (no transformation rules). Data
Vault~2.0~\cite{linstedt2015dv} provides an alternative methodology
with its own taxonomy (Hubs, Satellites, Links). Grain theory
formalizes grain as a composable property of arbitrary algebraic data
types with inference rules across the full relational algebra. The
behavioral classes of Section~\ref{sec:entity} unify both paradigms:
Kimball's SCD~Type~2 dimensions and Data~Vault Satellites are both
instances of IsMultiVersion; transaction and periodic snapshot fact
tables map to IsEvent and IsSnapshot
(Table~\ref{tab:bc-paradigms}).

Gray et al.~\cite{gray1997datacube} defined the Data Cube operator,
and Harinarayan et al.~\cite{harinarayan1996implementing} analyzed
the lattice of $2^N$ GROUP BY aggregations: an instance of the grain
lattice where roll-up is grain coarsening and drill-down is grain
refinement. Grain theory provides the algebraic structure underlying
this operational lattice and extends it with cross-cube correctness
guarantees (Section~\ref{sec:errors}).

Summarizability
theory~\cite{lenz1997summarizability,hurtado2002olap} formalizes when
aggregation is correct across grain levels via disjointness,
completeness, and type compatibility. Cabibbo and
Torlone~\cite{cabibbo1998logical}'s level hierarchies are
structurally grain orderings on dimension types. These formalizations
address declared OLAP dimensions; grain theory extends the algebraic
structure to arbitrary product types with inference rules across the
full RA, and fan-trap detection (Section~\ref{sec:errors}) is a
type-level analogue of summarizability conditions.

\subsection{Functional Dependencies and Normalization}
\label{subsec:rw-fd}

Grain theory and FD theory differ in both foundation and purpose.
FDs, MVDs, and JDs reason about decomposition
\emph{within} a single
relation~\cite{beeri1979computational,fagin1977multivalued,abiteboul1995foundations}:
normalization asks ``is this schema well-structured?'' Grain theory
reasons about correctness of transformations \emph{across} types: it
asks ``does this pipeline preserve or intentionally change the level of
detail that data represents?'' The two are complementary, not competing:
on product types, grain theory re-casts the powerful results of
classical dependency theory, Armstrong's axioms, key propagation,
attribute closure, under a data-engineering lens that makes them
actionable for pipeline verification, a problem FD theory was never
designed to address.

Crucially, \emph{grain is not defined in terms of functional
dependencies}. The grain of a type (Definition~\ref{def:grain}) is the
minimal type isomorphic to it, an irreducibility condition on the type's
structure, with no reference to FDs. On product types, this structural definition corresponds \emph{exactly} to minimal-key
reasoning (the grain is a minimal key), but that correspondence is a
theorem
(Proposition~\ref{prop:grain-completeness}), not the definition:
grain ordering $R_1 \leg R_2$ corresponds to functional determination
of $\grain{R_2}$'s components by $\grain{R_1}$'s components.

Despite this different foundation, grain ordering and FDs share
structural parallels. Table~\ref{tab:grain-vs-fd} summarizes the
distinctions.

\begin{table}[t]
\centering
\small
\begin{tabular}{@{}lll@{}}
\toprule
\textbf{Property} & \textbf{Grain ($\leg$)} & \textbf{FDs ($\to$)} \\
\midrule
Foundation & Irreducibility + isomorphism & Attribute determination \\
Domain & Algebraic data types & Attribute sets (flat) \\
Axioms & A1--A9 (\S\ref{subsec:armstrong}) & Armstrong's axioms \\
Antisymmetry & Yes (up to $\cong$) & No (preorder on attribute sets) \\
Structure & Partial order & Preorder (partial order on quotient) \\
Purpose & Transformation correctness & Schema decomposition \\
\bottomrule
\end{tabular}
\caption{Grain ordering vs functional dependencies.}
\label{tab:grain-vs-fd}
\end{table}

\noindent
For antisymmetry: FDs on attribute sets form a preorder; quotienting by
closure equivalence ($A^+ = B^+$) yields a partial order. Grain ordering
is \emph{intrinsically} a partial order on types up to isomorphism
(Theorem~\ref{thm:grain-partial-order}), with no quotient construction needed:
mutual admissible surjections give the two component sets the same
closure, and each type is isomorphic to the product on that closure
(Appendix~\ref{sec:proofs}).
The distinction between projection-witnessed and declaration-witnessed
grain ordering (Section~\ref{sec:relations}) mirrors the separation in
FD theory between reflexivity (schema-derivable) and
augmentation/transitivity (requiring declared dependencies).
This alignment is not coincidental: both assert the existence of a
surjective function from a finer representation to a coarser one. Both
also restrict which populations are legal, and neither is carried by
the bare type: $\grain{R} = T$ is an isomorphism only relative to the
declared determinations, exactly as
Remark~\ref{rem:grain-existence-declared} says. What differs is where
the declaration lives. Classical dependency theory keeps it outside
the type, a schema being a set of attributes \emph{together with} a
set of FDs whose content is a condition on instances. Grain theory
\emph{internalises} it (Section~\ref{sec:prelim}), so that once a
determination is declared the grain is a type and
$\grain{R_1} \thra \grain{R_2}$ is a surjection between types
(Remark~\ref{rem:grain-vs-pk}). The gain is compositionality: grain is
given by the same clauses over $\times$, $+$, $\mu$ and~$\nu$
(Definition~\ref{def:compositional-grain}), whereas an FD needs a
relation to be stated about at all. On product types the two agree,
which is Proposition~\ref{prop:grain-completeness}.
The bounds beyond product types are not established
(Section~\ref{subsec:incomparability-lattice}); their general form is
future work (Section~\ref{sec:conclusion}).

Armstrong~\cite{armstrong1974dependency} established the axiom system
for functional dependencies. Our grain axioms A1--A9
(Section~\ref{subsec:armstrong}) correspond structurally to Armstrong's
axioms, and that correspondence is what transfers completeness to the grain
system (Proposition~\ref{prop:grain-completeness}), stated and proved where
the axioms live rather than here.

\smallskip\noindent
On product types, then, grain ordering and functional dependencies describe
the same structure, and the minimal key~\cite{maier1983theory} is
Definition~\ref{def:grain}'s irreducibility restricted to attribute sets.
What this paper builds on that structure has no dependency-theoretic
analogue: grain for sums and the (co)inductive fixed points
(Section~\ref{sec:adt}), a per-operation calculus whose lift is faithful to
the type level (Sections~\ref{sec:inference}
and~\ref{sec:grain-factorization}), and pipeline-scale decidability
(Section~\ref{sec:calcg}).

Klug~\cite{klug1980calculating} computed the closure of functional
dependencies through relational algebra expressions. On product types,
CalcG's inference rules (Table~\ref{tab:ra-inference}) are the
key-propagation specialization of Klug's general FD rules: the
candidate-grain formula of
Theorem~\ref{thm:equijoin-candidates}, that
$\grain{R_1} \tun (\grain{R_2} \tdiff J_k)$ is isomorphic to
$\grain{Res}$ (in classical terms, a superkey of the result), is derivable from Klug's
calculus. The new content is the
\emph{irreducibility-with-labeling} result of
Theorem~\ref{thm:equijoin-grain}: this candidate is the grain (not
merely $\eqg$ to it) precisely when the labeling convention
$G_1^{J_k} \ssub G_2^{J_k}$ is applied. This minimality statement and its connection to the
asymmetric labeling rule does not appear as a named result in the
key-propagation
literature~\cite{klug1980calculating,maier1983theory,abiteboul1995foundations}.
Reaching the same grain through that literature takes two steps, propagating
dependencies and then enumerating candidate
keys~\cite{lucchesi1978candidate}; Theorem~\ref{thm:equijoin-grain} selects
the irreducible candidate by a single $\ssub$-comparison, with no
enumeration, and on collections that live in no catalogue. It is, like the
rest of Section~\ref{sec:inference}, a product-type result: what reaches
beyond product types is grain itself (Theorem~\ref{thm:grain-adt}), not
this inference rule.

\subsection{Data Modeling and Schema-Level Reasoning}
\label{subsec:rw-categorical}

Spivak~\cite{spivak2012functorial} developed functorial data migration
between database schemas presented as categories; Schultz and
Wisnesky~\cite{schultz2017algebraic} implemented this in AQL/CQL with
provable correctness, and Patterson et al.~\cite{patterson2022categorical}
introduced attributed C-sets as a practical categorical data structure.
The Grain Homomorphism Theorem
(Theorem~\ref{thm:grain-factorization}) is structurally parallel:
every data transformation reduces to its grain lift, which composes
(Corollary~\ref{cor:grain-compositionality}). Grain theory addresses
a narrower problem, grain verification rather than full data
migration, without the full categorical machinery.
Section~\ref{sec:adt}'s extension of grain to inductive and
coinductive types does use one standard ingredient of it, namely
initial-algebra semantics~\cite{birddemoor1997}, under which
naturally isomorphic polynomial functors have isomorphic initial
algebras and final coalgebras; a fuller functorial account remains
future work.

Cheney et al.~\cite{cheney2013practical} provide type-safe embeddings
of relational queries; grain theory targets a different type-level
correctness dimension: granularity preservation rather than type-safe
query construction.

Green et al.~\cite{green2007provenance} introduced provenance
semirings tracking how output tuples derive from inputs
compositionally; Buneman et al.~\cite{buneman2001why} characterized
why- and where-provenance. Grain theory tracks a coarser type-level
invariant, what level of detail outputs represent, but shares the
algebraic, compositional approach.

Chen~\cite{chen1976entity} introduced the entity-relationship model,
defining entities as standalone domain concepts. Our notion
(Section~\ref{sec:entity}) extends this: the \emph{entity of data}
is a grain-derived property of any algebraic data type, not only
standalone entities, with the entity key arising as a derived grain
($\ek{R} = \grain{E}$) (Remark~\ref{rem:data-integration}). In the
data exchange setting~\cite{arenas2014foundations,kolaitis2005schema},
grain relations and entity-key convergence could serve as mapping
invariants ensuring source-to-target transformations preserve
granularity and entity identity.

\section{Conclusion and Future Work}
\label{sec:conclusion}

Grain theory gives the level of detail at which data is represented a
definition, a structure, a calculus and a decision procedure. The
\emph{grain} of a type is fixed by irreducibility and isomorphism alone
(Section~\ref{sec:foundations}). Its ordering is a partial order on types,
with a bounded lattice and a complete axiomatization on product types
(Section~\ref{sec:relations}); grain, entity key and behavioral class form
the type-level triple $(\grain{R},\,\ek{R},\,\BC{R})$ that every type
carries (Section~\ref{sec:entity}); and grain is defined for every
algebraic data type, the grain type-expression operator commuting with each
type constructor, the inductive and coinductive fixed points included
(Section~\ref{sec:adt}). Every transformation
$h : R_1 \to R_2$ has a grain lift $\gl{h} : \grain{R_1} \to \grain{R_2}$
that commutes with grain projection and composes homomorphically
(Section~\ref{sec:grain-factorization}), computed for the relational
algebra by the inference rules of Section~\ref{sec:inference}. Because that
lift is faithful and composes, CalcG (Section~\ref{sec:calcg}) decides a
pipeline's grain-correctness from the schema alone, before any code is
written or any query is executed. Fan traps are thereby characterized as
grain-relation violations and projecting away a grain field as a
grain-uniqueness violation, while wrong-grain aggregations and
behavioral-class violations are type errors, all at the type level; chasm
traps, which the theory does \emph{not} model, are only localized to the
vulnerable pattern (Section~\ref{sec:errors}). The theory is mechanized in
Lean~4~\cite{demoura2021lean4} (Section~\ref{sec:verification}).

Grain theory reuses a relation from dependency theory but not its
problem. Functional dependencies eliminate redundancy while a schema is
being \emph{designed}. Grain theory takes the finished schema's types as
fixed and verifies that pipeline transformations keep data at its
intended \emph{level of detail}. The two relations coincide on product
types and diverge on sums and (co)inductive types. More to the point, the
question they answer is different. Grain theory is the design-time
discipline for granularity correctness, a concern that begins where
schema design ends.

Grain theory captures \emph{structural} correctness, that is, whether each
transformation preserves or intentionally changes the level of detail that data represents. Several directions for future work
emerge.
First, verification answers ``is \emph{this} pipeline correct?''. The
next question is how to \emph{build} one that cannot be incorrect: a
typed algebra of operations, each carrying a
$(\textsf{Pre},\textsf{Intra},\textsf{Post})$ contract, in which every
well-typed composition is grain-correct by construction and the
residual obligations, the ones that genuinely depend on the data, are
discharged as generated verification queries. That is the
\emph{synthesis} problem, and we develop it in a companion
paper~\cite{pddpaper}. It is also what would discharge the one clause
of a behavioral class this paper can only assume, the universally
quantified invariant
(Section~\ref{subsec:behavioral-semantics}).
Second, integrating grain verification with additional correctness
dimensions, referential integrity, domain constraints, temporal
consistency, would broaden the class of errors detectable at the
type level.
Third, lifting summarizability
theory~\cite{lenz1997summarizability} onto the grain
lattice. Every aggregation function tolerates a region of the lattice
in which it may be safely computed; a declared admissible sub-lattice
per aggregate would let CalcG check each aggregation vertex against
it, the fan trap and the average of averages becoming the two boundary
violations of one declaration.
Fourth, the lattice operations $\tin$ (lub) and $\tun$ (glb) are
introduced order-theoretically in this paper
(Sections~\ref{sec:prelim},~\ref{subsec:incomparability-lattice});
their categorical formalization in the category of types and
surjections, with $\tdiff$ as a chosen complement, remains future
work, connecting grain theory to the categorical database
literature~\cite{spivak2012functorial,schultz2017algebraic}. In
that setting, schema-derivable and declared surjection witnesses
would appear as a stratification of the surjection catalog.
Finally, one gap in the theory itself. The grain axioms A1--A9 are
proved complete only for product types
(Proposition~\ref{prop:grain-completeness}): on products, asking
whether some grain facts entail another is asking whether some
functional dependencies imply another
(Corollary~\ref{cor:grain-implication}), and Armstrong's completeness
carries over. Once sums or recursion enter that route closes, there
being no finite set of attributes left to close over, and we know of
no complete axiom system for the general case. This concerns the
axioms, not pipeline checking: CalcG computes a grain forward and
compares it with the declared target, and never asks what follows
from what.

\bibliographystyle{ACM-Reference-Format}
\bibliography{bibliography/references}

\appendix

\section{Type-Theory Primer}
\label{app:primer}

For readers without a functional-programming background, this primer
collects the type-theoretic vocabulary used in the main text, in
database terms.

\begin{description}
\item[Base type.] An atomic type, not built from others
  ($\textit{CustomerId}$, $\textit{Integer}$, $\textit{Date}$); the
  leaves of the type grammar.
\item[Carrier.] The underlying set of values of a type, with the
  structure of its elements forgotten. Comparing carriers by bijection
  is uninformative for infinite types, all countably-infinite carriers
  are in bijection ($\mathbb{N} \cong \mathbb{N} \times \mathbb{N}$), which
  is why grain irreducibility is defined structurally, not on carriers.
\item[Product $\times$, sum $+$.] Record type (all components, ``and'')
  and variant / tagged union (``or''); non-recursive constructors that
  combine already-formed types.
\item[Metavariable $T$ vs.\ type variable $X$.] In
  $T ::= B \mid X \mid T \times T \mid T + T \mid \mu X.\,T \mid \nu X.\,T$,
  each $T$ is grammar shorthand for ``any type,'' so $T \times T$ is
  \emph{composition}; genuine self-reference occurs only through the
  binders $\mu X$/$\nu X$, whose bound variable $X$ names the type being
  defined (e.g.\ the tail in $\textit{List}\,R = \mu X.\,\mathbf{1} + R
  \times X$).
\item[$\mu X.\,F$ (inductive), $\nu X.\,F$ (coinductive).] The least and
  greatest solutions of $X = F(X)$: finite data built by constructors
  (lists, trees) and possibly-infinite data defined by observation
  (streams).
\item[Fixed point.] A solution of $X = F(X)$; $\mu$ least, $\nu$
  greatest.
\item[Polynomial functor.] A type expression with a hole $F(X)$ built
  from constants, $\times$, and $+$ that also acts on functions; e.g.\
  the list shape $\mathbf{1} + R \times X$.
\item[Initial algebra / final coalgebra.] The categorical names for
  $\mu$ and $\nu$: finite data with a unique fold out of it, and
  possibly-infinite data with a unique unfold into it.
\item[Strictly positive.] $X$ occurs only in data positions, never to
  the left of a function arrow; the condition guaranteeing $\mu$/$\nu$
  exist. Every schema-like recursive type qualifies.
\item[Closed vs.\ open.] A \emph{closed} type has no free type variables
  ($\textit{List}\,\textit{Int}$); an \emph{open} one has holes (a
  functor body $F$).
\item[Natural isomorphism.] An isomorphism between two type
  constructors that is \emph{uniform in the hole}: it never inspects
  what sits in the recursion slot but passes it through. Uniformity is
  what lets an isomorphism between the bodies of two recursive
  definitions carry over to the types they generate.
\end{description}

\section{Proofs}
\label{sec:proofs}

This appendix contains the proofs deferred from the body. A handful of
results are proved in place instead, where the argument is a few lines and
is needed immediately; a handful of the error characterizations in
Section~\ref{sec:errors} carry their argument inside the statement; and the
remarks and examples deferred for readability are collected separately in
Appendix~\ref{sec:deferred-remarks}. Each body section says which case
applies to it. Throughout, we use the macros and definitions established in
the body.

\subsection{Foundations (Section~\ref{sec:foundations})}

\begin{proof}[Proof of Theorem~\ref{thm:multiple-grains} (Multiple Grains Isomorphism)]
If $G_1$ and $G_2$ are both grains of~$R$, then by Definition~\ref{def:grain}
there exist isomorphisms
$f_{g1} : G_1 \stackrel{\cong}{\longrightarrow} R$ and
$f_{g2} : G_2 \stackrel{\cong}{\longrightarrow} R$.
Define $g_{12} = f_{g2}^{-1} \circ f_{g1} : G_1 \to G_2$.
Its inverse is $g_{12}^{-1} = f_{g1}^{-1} \circ f_{g2} : G_2 \to G_1$.
We verify:
\begin{align*}
g_{12} \circ g_{12}^{-1}
  &= (f_{g2}^{-1} \circ f_{g1}) \circ (f_{g1}^{-1} \circ f_{g2})
   = f_{g2}^{-1} \circ f_{g2} = \mathrm{id}_{G_2}, \\
g_{12}^{-1} \circ g_{12}
  &= (f_{g1}^{-1} \circ f_{g2}) \circ (f_{g2}^{-1} \circ f_{g1})
   = f_{g1}^{-1} \circ f_{g1} = \mathrm{id}_{G_1}.
\end{align*}
Therefore $g_{12}$ is an isomorphism and $G_1 \cong G_2$.
\end{proof}

\begin{proof}[Proof of Theorem~\ref{thm:grain-uniqueness} (Grain Uniqueness)]
By Definition~\ref{def:grain}, the grain function
$f_g : \grain{R} \stackrel{\cong}{\longrightarrow} R$ is an isomorphism.
The grain projection $\textit{grain} = f_g^{-1} : R \to \grain{R}$ is
therefore also an isomorphism, hence injective. If
$\textit{grain}\;r_1 = \textit{grain}\;r_2$, then
$r_1 = f_g(\textit{grain}\;r_1) = f_g(\textit{grain}\;r_2) = r_2$.
\end{proof}

\begin{proof}[Proof of Theorem~\ref{thm:grain-idempotent} (Grain of a Grain, Idempotency)]
Let $G = \grain{R}$. We show $\grain{G} = G$.

Suppose for contradiction that $G' = \grain{G}$ with $G' \spsub G$.
By the grain definition, $G' \cong G$. Since $G \cong R$
(Definition~\ref{def:grain}), transitivity gives $G' \cong R$. But
$G' \spsub G \subt R$, contradicting the irreducibility of~$G$ as the grain
of~$R$: a proper subset of the grain that is isomorphic to~$R$ cannot exist.

Therefore $\grain{G} = G$, and by substitution
$\grain{\grain{R}} = \grain{R}$.
\end{proof}

\begin{proof}[Proof of Lemma~\ref{lem:slotwise-structure-map} (Structure maps act slotwise)]
By the structure-map paragraph of Section~\ref{sec:prelim}, $\tau$ is a
composition of the canonical projections and pairings of products, the
injections of sums, each fixed point's constructor and destructor, the
declared determinations, and the inverses of isomorphisms in the class.

\emph{No generator mixes slots.}\quad
Source and target are products, so no sum injection and no fixed-point
constructor or destructor occurs. A projection selects one slot of a product
and a pairing assembles slots into a product, so in any composition of
projections and pairings the value in each output slot is the value of
exactly one input slot: neither generator has the shape required to combine
two distinct input slots into one output slot. The only generators of the
class that do are the declared determinations, and by hypothesis there are
none among the components of~$R$. Nor are there any among the
$\grain{R_i}$, since a declared determination
$\grain{R_i} \thra \grain{R_j}$ with $i \neq j$ would compose with the grain
isomorphisms to give a surjection $R_i \thra R_j$, a determination among the
components of~$R$. Hence $\tau$ carries each slot of~$Q$ bijectively onto a
single slot of~$R$, inducing an injection
$\iota : \mathcal{I} \to \{1, \dots, n\}$; as $\tau$ is onto, $\iota$ is a
bijection, and since $\mathcal{I} \subseteq \{1, \dots, n\}$ this forces
$\mathcal{I} = \{1, \dots, n\}$.

\emph{The bijection is the identity.}\quad
Fix $i \in \mathcal{I}$. The components of the slot $Q_i$ are components of
$\grain{R_i}$ and hence of~$R_i$: for a product $R_i$ this is
Lemma~\ref{lem:grain-existence}, which exhibits the grain as a
$\ssub$-minimal structural subtype of~$R_i$, and for an arbitrary algebraic
data type it is the compositional construction of
Definition~\ref{def:compositional-grain}, which builds $\gop{R_i}$ from the
components of~$R_i$. By the previous paragraph the restriction of~$\tau$ to
this slot is an isomorphism $Q_i \to R_{\iota(i)}$ built from projections
and pairings alone, so each component of $R_{\iota(i)}$ is a component
of~$Q_i$, and therefore of~$R_i$. Distinct slots of a product have disjoint
components, so $\iota(i) = i$. Restricting $\tau$ slotwise thus yields
isomorphisms $\tau_i : Q_i \to R_i$ with
$\tau = \tau_1 \times \cdots \times \tau_n$.
\end{proof}

\begin{proof}[Proof of Theorem~\ref{thm:grain-product} (Grain of products and sums)]
\emph{Leg~(1), products.}
Let $R = R_1 \times R_2 \times \cdots \times R_n$. We show
$\grain{R} = \grain{R_1} \times \grain{R_2} \times \cdots \times \grain{R_n}$
by verifying both conditions of Definition~\ref{def:grain}.

\emph{Isomorphism.}\quad
For each~$i$, let $f_{gi} : \grain{R_i} \stackrel{\cong}{\longrightarrow} R_i$
be the grain function. Construct
\[
f_g(g_1, \ldots, g_n) = (f_{g1}(g_1), \ldots, f_{gn}(g_n)).
\]
Since each $f_{gi}$ is an isomorphism, $f_g$ is an isomorphism with inverse
\[
f_g^{-1}(r_1, \ldots, r_n) = (f_{g1}^{-1}(r_1), \ldots, f_{gn}^{-1}(r_n)).
\]

\emph{Irreducibility.}\quad
Suppose for contradiction that
$G' \spsub \grain{R_1} \times \cdots \times \grain{R_n}$ with $G' \cong R$.
Being a proper structural subtype, $G'$ either omits some slot of the
product entirely or satisfies $G'_i \spsub \grain{R_i}$ in some slot~$i$;
in both cases its slots meet the hypothesis
$G'_i \ssub \grain{R_i}$ of Lemma~\ref{lem:slotwise-structure-map}. The
theorem's hypothesis, that no determination is declared among the
components of~$R$, is precisely what that lemma requires, and it is here
that the hypothesis is used: the isomorphism $G' \cong R$ therefore acts
slotwise, matching the slots of~$G'$ with those of~$R$ one for one and
restricting to an isomorphism in each. This rules out the first case, in
which an omitted slot would leave the lemma's index set a proper subset of
$\{1, \dots, n\}$. In the second, the restriction is
an isomorphism $G'_i \cong R_i$ with $G'_i \spsub \grain{R_i}$: a proper
structural subtype of $\grain{R_i}$ isomorphic to~$R_i$, contradicting the
irreducibility of $\grain{R_i}$ (Definition~\ref{def:grain}).

Note that composing $G' \cong R$ with the $i$-th projection of~$R$ does
\emph{not} by itself yield the contradiction, since that composite is a
surjection onto the whole of~$R_i$; what is needed is that the isomorphism
respects slots, which is what the lemma supplies and what the
determination-free hypothesis buys.

\emph{Leg~(2), sums.}

Let $R = R_1 + R_2 + \cdots + R_n$. We show
$\grain{R} = \grain{R_1} + \grain{R_2} + \cdots + \grain{R_n}$.

\emph{Isomorphism.}\quad
For each~$i$, let $f_{gi} : \grain{R_i} \stackrel{\cong}{\longrightarrow} R_i$.
Define $f_g$ by case analysis:
$f_g(g) = f_{gi}(g)$ when $g : \grain{R_i}$.
Since each $f_{gi}$ is an isomorphism on its component and the sum components
are disjoint, $f_g$ is an isomorphism.

\emph{Irreducibility.}\quad
Suppose $G' \spsub \grain{R_1} + \cdots + \grain{R_n}$ with $G' \cong R$
by a structure map. A structural subtype of a sum is summand-wise,
$G' = G'_1 + \cdots + G'_n$ with $G'_i \ssub \grain{R_i}$, proper in at
least one summand. A structure map out of a sum is a copairing, and a
structure map out of a product reaches a sum only through an injection
(the sum case of the antisymmetry proof,
Theorem~\ref{thm:grain-partial-order}), so the isomorphism and its inverse
restrict to isomorphisms $G'_i \cong R_{\tau(i)} \cong \grain{R_{\tau(i)}}$
for a permutation~$\tau$ of the summands. Follow $\tau$ around the cycle
through a summand~$i$ at which $G'_i$ is proper:
$\grain{R_i} \sqsupseteq G'_i \cong \grain{R_{\tau(i)}} \sqsupseteq
G'_{\tau(i)} \cong \cdots \cong \grain{R_i}$ is a chain of surjective
structure maps from $\grain{R_i}$ to itself that passes through a proper
structural subtype of $\grain{R_i}$. By Lemma~\ref{lem:sat-closure} every
component of $\grain{R_i}$ then lies in the closure of a proper subset of
its components, so some component of $\grain{R_i}$ is determined by the
others, and the structural subtype dropping it is isomorphic to
$\grain{R_i}$, hence to~$R_i$, by pairing with the derived determination.
That contradicts the irreducibility of $\grain{R_i}$.
\end{proof}

\subsection{Grain Relations (Section~\ref{sec:relations})}

\begin{proof}[Proof of Theorem~\ref{thm:grain-equality} (Grain Equality)]
We prove $R_1 \eqg R_2 \;\Leftrightarrow\; R_1 \cong R_2$.

($\Rightarrow$)\quad
If $R_1 \eqg R_2$, then by Definition~\ref{def:grain-eq} there exists
$f : \grain{R_1} \stackrel{\cong}{\longrightarrow} \grain{R_2}$.
Let $f_{g1} : \grain{R_1} \stackrel{\cong}{\longrightarrow} R_1$ and
$f_{g2} : \grain{R_2} \stackrel{\cong}{\longrightarrow} R_2$ be the
respective grain functions. Then
$h = f_{g2} \circ f \circ f_{g1}^{-1} : R_1 \stackrel{\cong}{\longrightarrow} R_2$
is a composition of isomorphisms, hence an isomorphism.

($\Leftarrow$)\quad
If $h : R_1 \stackrel{\cong}{\longrightarrow} R_2$, construct
$f = f_{g2}^{-1} \circ h \circ f_{g1} : \grain{R_1} \stackrel{\cong}{\longrightarrow} \grain{R_2}$.
This is a composition of isomorphisms, so by Definition~\ref{def:grain-eq},
$R_1 \eqg R_2$.
\end{proof}

\begin{proof}[Proof of Corollary~\ref{thm:grain-ordering} (Type-Level Reading of Grain Ordering)]
$R_1 \leg R_2$ if and only if there exists a surjective
$h : R_1 \thra R_2$ establishing a many-to-one correspondence.

($\Rightarrow$)\quad
If $R_1 \leg R_2$, then by Definition~\ref{def:grain-ord} there exists
a surjective $f : \grain{R_1} \thra \grain{R_2}$. Construct
$h = f_{g2} \circ f \circ f_{g1}^{-1} : R_1 \to R_2$, where $f_{g1}^{-1}$
and $f_{g2}$ are the grain projection and grain function respectively.
Since $f_{g1}^{-1}$ and $f_{g2}$ are bijections and $f$ is surjective,
$h$ is surjective (composition of surjections), establishing a many-to-one
correspondence where each $R_2$ element has at least one $R_1$ preimage.

($\Leftarrow$)\quad
If $h : R_1 \thra R_2$ is surjective, construct
$f = f_{g2}^{-1} \circ h \circ f_{g1} : \grain{R_1} \to \grain{R_2}$.
Since $f_{g1}$ and $f_{g2}^{-1}$ are bijections and $h$ is surjective,
$f$ is surjective, so $R_1 \leg R_2$ by Definition~\ref{def:grain-ord}.
\end{proof}

\begin{proof}[Proof of Theorem~\ref{thm:grain-partial-order} (Grain Ordering is a Partial Order)]
\emph{Reflexivity.}\quad The identity on $\grain{R}$ is a surjection
$\grain{R} \thra \grain{R}$, so $R \leg R$.

\emph{Antisymmetry.}\quad Suppose $R_1 \leg R_2$ and $R_2 \leg R_1$,
witnessed by structure maps $f : \grain{R_1} \thra \grain{R_2}$ and
$g : \grain{R_2} \thra \grain{R_1}$ (Definition~\ref{def:grain-ord}),
under a coherent declaration set~$\Delta$
(Definition~\ref{def:coherent-declarations}). We show
$\grain{R_1} \cong \grain{R_2}$ by structure maps, which is
$R_1 \eqg R_2$ (Definition~\ref{def:grain-eq}), by the shape of the
types.

\emph{Product types.}\quad Write $X$ and $Y$ for the component sets of
$R_1$ and~$R_2$ and $(\cdot)^{+}$ for closure under~$\Delta$. By
Lemma~\ref{lem:sat-closure} the two hypotheses read $Y \subseteq X^{+}$
and $X \subseteq Y^{+}$, so $X^{+} = Y^{+}$; write $Z$ for this set and
also for the product type on it. The pairing of the projections onto the
components of~$X$ with the derived determinations from $X$ to the
remaining components of~$Z$ is a structure map $R_1 \to Z$. It lands in
the constrained product on~$Z$, because coherence makes the derived
determinations agree with one another on the record so built; it is
injective, being the identity on the components of~$X$; and it is onto,
because a valid record over~$Z$ restricts to a valid record over~$X$
from which the pairing rebuilds it, its other components being
determined by those of~$X$. Its inverse is the canonical projection
$Z \thra R_1$, a generator. So $R_1 \cong Z$ by structure maps, and
likewise $R_2 \cong Z$, whence $R_1 \cong R_2$, and composing with the
grain isomorphisms, $\grain{R_1} \cong \grain{R_2}$. Nothing here
depends on how the base types are inhabited or on how many values they
carry. Remark~\ref{ex:employee-email} is the case
$X = \{\textit{EmployeeId}\}$, $Y = \{\textit{EmailAddr}\}$ and
$Z = X \cup Y$: the pairing $e \mapsto (e, m(e))$ lands in the
constrained product exactly because coherence gives $o(m(e)) = e$.

\emph{Sums of products.}\quad Let $R_1 = \sum_{i \in I} P_i$ and
$R_2 = \sum_{j \in J} Q_j$ with product summands and $I$, $J$ finite. A
structure map out of a sum is a copairing, and a structure map out of a
product reaches a sum only through an injection, since no generator
applied to a product produces a tag. So $f$ restricts on each $P_i$ to a
structure map $f_i : P_i \to Q_{\tau(i)}$ followed by the injection, for
a function $\tau : I \to J$, and likewise $g$ restricts to
$g_j : Q_j \to P_{\tau'(j)}$ for a function $\tau' : J \to I$. Summands
are nonempty, so $f$ onto forces $\tau$ onto and $g$ onto forces $\tau'$
onto; with $I$ and $J$ finite, both are bijections, and then each $f_i$
is onto $Q_{\tau(i)}$ and each $g_j$ onto $P_{\tau'(j)}$, each summand
of the target being reached from one summand of the source. Following
the permutation $\tau' \circ \tau$ of~$I$ around the cycle through~$i$
gives a chain of surjective structure maps
$P_i \thra Q_{\tau(i)} \thra P_{\tau'\tau(i)} \thra \cdots \thra P_i$,
so by Lemma~\ref{lem:sat-closure} every product on the chain has the
same closure, and the product case gives $P_i \cong Q_{\tau(i)}$ for
every~$i$. The copairing of these isomorphisms with the injections
along~$\tau$ is a structure map $R_1 \cong R_2$, and the grain
isomorphisms finish as before.

\emph{Fixed points.}\quad When $f$ and $g$ are the folds or unfolds
induced by mutual surjective structure maps between the bodies, with the
recursion variable held fixed as a component, the two cases above apply
to the bodies and give isomorphic bodies, and the first fixed-point
hypothesis of Theorem~\ref{thm:grain-adt}, that isomorphic bodies have
isomorphic fixed points, lifts the isomorphism. A witness that
traverses the recursion, a fold whose algebra consumes more than one
layer, is not of this form, and antisymmetry for such witnesses is left
open; the statement records the scope.

\emph{Transitivity.}\quad If
$f : \grain{R_1} \thra \grain{R_2}$ and
$g : \grain{R_2} \thra \grain{R_3}$ are surjective, then
$g \circ f : \grain{R_1} \thra \grain{R_3}$ is surjective
(composition of surjections), so $R_1 \leg R_3$.
\end{proof}

\begin{proof}[Proof of Lemma~\ref{thm:grain-subset} (Grain Subset--Ordering Equivalence)]
Both directions are immediate: $\grain{R_1} \subt \grain{R_2}$ asserts
a surjection $\grain{R_2} \thra \grain{R_1}$
(Definition~\ref{def:type-subset}), which is exactly
Definition~\ref{def:grain-ord}'s condition for $R_2 \leg R_1$.
\end{proof}

\begin{proof}[Proof of Corollary~\ref{cor:grain-all-subsets} (Grain Determines All Type Subsets)]
Let $G = \grain{R}$ and $R' \subt R$. The grain function
$f_g : G \stackrel{\cong}{\longrightarrow} R$ is surjective (bijection);
the projection $p : R \thra R'$ is surjective (by definition of $\subt$).
Their composition $p \circ f_g : G \thra R'$ is surjective, so $G \leg R'$
by Definition~\ref{def:grain-ord}.
\end{proof}

\begin{proof}[Proof of Theorem~\ref{thm:grain-inference-sufficient} (Grain Inference)]
We show that if $G$ satisfies (i)~$G \subt R$, (ii)~$G \leg R$, and
(iii)~$\grain{G} = G$, then $\grain{R} = G$
(Definition~\ref{def:grain}: isomorphism + irreducibility).

\emph{Step~1: $R \leg G$.}\quad
From (i), $G \subt R$, so there is a surjective projection
$p : R \thra G$. The composition $p \circ f_g : \grain{R} \thra G$
is surjective (bijection composed with surjection), so $R \leg G$
by Definition~\ref{def:grain-ord}.

\emph{Step~2: Grain equivalence.}\quad
Hypothesis~(ii) gives $G \leg R$, and Step~1 gives $R \leg G$. By
antisymmetry (Theorem~\ref{thm:grain-partial-order}), $G \eqg R$.

\emph{Step~3: Isomorphism ($G \cong R$).}\quad
From Step~2, $\grain{G} \cong \grain{R}$. By~(iii), $\grain{G} = G$,
so $G \cong \grain{R}$; since $\grain{R} \cong R$
(Definition~\ref{def:grain}), $G \cong R$ by transitivity.

\emph{Step~4: Irreducibility.}\quad
Definition~\ref{def:grain} asks that no \emph{proper structural}
subtype of~$G$ be isomorphic to~$R$; the semantic reading is vacuous,
as that definition notes. So suppose $G' \spsub G$ with $G' \cong R$.
By Step~3, $G \cong R$, hence $G' \cong G$ by transitivity. Then $G'$
is a proper structural subtype of~$G$ isomorphic to~$G$, which
contradicts~(iii): $\grain{G} = G$ says precisely that $G$ has no
such subtype.

\smallskip\noindent
Steps~3--4 establish Definition~\ref{def:grain} for $G$ and $R$:
$\grain{R} = G$.
\end{proof}

\paragraph{Soundness of Armstrong Axioms A1--A9.}
We verify that each axiom preserves the existence of a function between
grains.

\begin{itemize}
\item \textbf{A1 (Self-determination).} $R \leg R$: the identity function
  $\mathrm{id} : \grain{R} \to \grain{R}$ witnesses the ordering.

\item \textbf{A2 (Reflexivity).} If $R_1 \ssub R_2$, then $R_2 \leg R_1$:
  the canonical projection $R_2 \thra R_1$ is a surjection, so
  $R_1 \subt R_2$ and Corollary~\ref{cor:grain-all-subsets} applies.

\item \textbf{A3 (Augmentation).} If $R_1 \leg R_2$ via
  $f : \grain{R_1} \thra \grain{R_2}$, then
  $R_1 \tun R_3 \leg R_2 \tun R_3$ for any~$R_3$.

  When $R_3$ shares no field with $R_1$ or $R_2$ the union is a product of
  disjoint components, so
  $\grain{R_1 \tun R_3} = \grain{R_1} \times \grain{R_3}$ and
  $\grain{R_2 \tun R_3} = \grain{R_2} \times \grain{R_3}$, and
  $f \times \mathrm{id}$ is the witness. That is the easy case and not the
  general one, because $\tun$ \emph{identifies} shared components. The
  reduction $R \tun S = R \times (S \tdiff R)$ of
  Section~\ref{sec:prelim} restores disjointness, but only at the price of
  making the added factor differ on the two sides, since
  $R_1 \tun R_3 = R_1 \times (R_3 \tdiff R_1)$ while
  $R_2 \tun R_3 = R_2 \times (R_3 \tdiff R_2)$. So $f \times \mathrm{id}$
  is unavailable, and the witness has to come from the grains of the
  unions themselves.

  Write $\mathrm{comps}(T)$ for the components of a type~$T$, abbreviate
  $X_i = \mathrm{comps}(R_i)$, and write $(\cdot)^{+}$ for closure under
  the declaration set in force (Definition~\ref{def:grain-holds}). By
  Lemma~\ref{lem:sat-closure} the hypothesis $R_1 \leg R_2$ reads
  $X_2 \subseteq X_1^{+}$, and the conclusion
  $R_1 \tun R_3 \leg R_2 \tun R_3$ reads
  $\mathrm{comps}(R_2 \tun R_3) \subseteq \mathrm{comps}(R_1 \tun R_3)^{+}$.
  A union carries every field of each side and no other, so
  $\mathrm{comps}(R_i \tun R_3) = X_i \cup X_3$, and
  \[
    X_2 \cup X_3 \;\subseteq\; X_1^{+} \cup X_3 \;\subseteq\;
    (X_1 \cup X_3)^{+},
  \]
  closure being monotone and extensive. By Lemma~\ref{lem:sat-closure}
  again, $R_1 \tun R_3 \leg R_2 \tun R_3$.

  The identification is harmless: closure is computed on \emph{sets} of
  components, so a field shared by $R_2$ and $R_3$ appears once in
  $X_2 \cup X_3$ and once in the union schema, and the two readings agree.

\item \textbf{A4 (Transitivity).} If $R_1 \leg R_2$ via~$f$ and
  $R_2 \leg R_3$ via~$g$, then $g \circ f : \grain{R_1} \to \grain{R_3}$
  witnesses $R_1 \leg R_3$.

\item \textbf{A5 (Union).} If $R_1 \leg R_2$ and $R_1 \leg R_3$, then
  by A3 (augment $R_1 \leg R_3$ with $R_1$) we get
  $R_1 \tun R_1 \leg R_3 \tun R_1$, that is $R_1 \leg R_1 \tun R_3$,
  $\tun$ being idempotent and commutative; by A3 again (augment
  $R_1 \leg R_2$ with $R_3$) we get $R_1 \tun R_3 \leg R_2 \tun R_3$;
  and A4 composes the two into $R_1 \leg R_2 \tun R_3$.

\item \textbf{A6 (Decomposition).} If $R_1 \leg (R_2 \tun R_3)$, then
  since $\grain{R_2} \subt \grain{R_2 \tun R_3}$ and
  $\grain{R_3} \subt \grain{R_2 \tun R_3}$, the projections composed
  with the witnessing function give $R_1 \leg R_2$ and $R_1 \leg R_3$.

\item \textbf{A7 (Composition).} From $R_1 \leg R_2$ and $R_3 \leg R_4$,
  A3 gives $(R_1 \tun R_3) \leg (R_2 \tun R_3)$ and
  $(R_2 \tun R_3) \leg (R_2 \tun R_4)$. By A4,
  $(R_1 \tun R_3) \leg (R_2 \tun R_4)$.

\item \textbf{A8 (Pseudotransitivity).} From $R_1 \leg R_2$ and
  $(R_2 \tun R_4) \leg R_3$: A3 gives
  $(R_1 \tun R_4) \leg (R_2 \tun R_4)$. By A4,
  $(R_1 \tun R_4) \leg R_3$.

\item \textbf{A9 (Darwen's Theorem).} From $R_1 \leg R_2$ and
  $R_3 \leg R_4$: by A3 (augment $R_1 \leg R_2$ with $R_3 \tdiff R_2$),
  $(R_1 \tun (R_3 \tdiff R_2)) \leg (R_2 \tun (R_3 \tdiff R_2))$, and
  $R_2 \tun (R_3 \tdiff R_2) = R_2 \tun R_3$ by the field-set identity
  $A \tun B = A \times (B \tdiff A)$ of Section~\ref{sec:prelim}. By A3
  again (augment $R_3 \leg R_4$ with $R_2$),
  $(R_2 \tun R_3) \leg (R_2 \tun R_4)$. By A4,
  $(R_1 \tun (R_3 \tdiff R_2)) \leg (R_2 \tun R_4)$.
\end{itemize}
\smallskip\noindent
A3 was the only axiom resting on disjointness. A5, A7, A8 and A9 are
derived from it together with A2 and A4, so they inherit the general case;
A6 uses only that a union carries every field of each side, which holds
whether or not the two overlap; and A9's use of $\tdiff$ is overlap-aware
by construction.

\subsection{Entity Key (Section~\ref{sec:entity})}

\begin{proof}[Proof of Theorem~\ref{thm:ek-hierarchy} (EK--Grain Hierarchy)]
By Definition~\ref{def:entity-key}, $\ek{R} = \grain{E}$.
We construct a surjection $g_{ek} : \grain{R} \thra \grain{E}$, which
gives $\grain{E} \subt \grain{R}$ by Definition~\ref{def:type-subset}.

By Definition~\ref{def:entity}, $\textit{entity} : R \thra E$ is a
surjection. The grain function $f_{g_R} : \grain{R} \to R$ is a bijection
(Definition~\ref{def:grain}), so the composition
$g_e = \textit{entity} \circ f_{g_R} : \grain{R} \thra E$
is a surjection. The grain function $f_{g_E} : \grain{E} \to E$ is also
a bijection, so its inverse $f_{g_E}^{-1} : E \to \grain{E}$ exists.
Define $g_{ek} = f_{g_E}^{-1} \circ g_e : \grain{R} \to \grain{E}$.
As the composition of a surjection with a bijection, $g_{ek}$ is a
surjection, giving $\ek{R} = \grain{E} \subt \grain{R}$.
\end{proof}

\subsection{Grain Inference (Section~\ref{sec:inference})}

\begin{proof}[Proof of Theorem~\ref{thm:equijoin-candidates} (Equi-Join Candidate Grains)]
Given collections $C\;R_1$ and $C\;R_2$ with join key
$J_k \ssub R_1$ and $J_k \ssub R_2$, the equi-join produces
$C\;Res$ where $Res = (R_1 \tdiff J_k) \times (R_2 \tdiff J_k) \times J_k$.
Each result element is a pair $(r_1, r_2)$ with matching $J_k$ values.

\paragraph{Notation.}
Let $G_i^{J_k} = \grain{R_i} \tin J_k$ (the $J_k$-portion of each grain)
and $G_i^{rest} = \grain{R_i} \tdiff J_k$ (the non-$J_k$ portion).
Then $G_{cand}(i,j) = \grain{R_i} \tun G_j^{rest}$.

\paragraph{$G_{cand}(i,j) \eqg Res$.}
By Theorem~\ref{thm:grain-equality},
$G_{cand}(i,j) \eqg Res$ iff
$G_{cand}(i,j) \cong Res$. We verify conditions~(i)~$\subt$ and
(ii)~$\leg$ of Theorem~\ref{thm:grain-inference-sufficient} for
$G_{cand}(i,j)$ and~$Res$; together they give the isomorphism.

\emph{Condition~(i): $G_{cand}(i,j) \subt Res$.}\quad
$\grain{R_i} \subt R_i \subt Res$ and
$G_j^{rest} \subt R_j \tdiff J_k \subt Res$, so
$G_{cand}(i,j) \subt Res$.

\emph{Condition~(ii): $G_{cand}(i,j) \leg Res$.}\quad
We show $G_{cand}(i,j)$ determines $Res$ via bootstrapping. Take
$(i,j) = (1,2)$ without loss of generality (the argument is
symmetric). Suppose two result pairs $(r_1, r_2)$ and
$(r_1', r_2')$ agree on $G_{cand}(1,2)$:
\begin{enumerate}
\item $G_{cand}(1,2)$ contains all of $\grain{R_1}$, so
  $\textit{grain}(r_1) = \textit{grain}(r_1')$, hence
  $r_1 = r_1'$ (Theorem~\ref{thm:grain-uniqueness}).
\item In particular $r_1.J_k = r_1'.J_k$. The equi-join condition gives
  $r_2.J_k = r_1.J_k = r_1'.J_k = r_2'.J_k$.
\item $G_{cand}(1,2)$ contains $G_2^{rest}$, so $r_2$ and $r_2'$ agree on
  $G_2^{rest}$. Combined with step~2: $r_2$ and $r_2'$ agree on
  $G_2^{rest} \tun G_2^{J_k} = \grain{R_2}$, hence
  $r_2 = r_2'$ (Theorem~\ref{thm:grain-uniqueness}).
\end{enumerate}
Therefore $(r_1, r_2) = (r_1', r_2')$, establishing
$G_{cand}(1,2) \leg Res$. By symmetry, also $G_{cand}(2,1) \leg Res$.

By Theorem~\ref{thm:grain-equality}, $G_{cand}(i,j) \eqg Res$
for both labelings.

\paragraph{Verifying the bounds.}
\emph{Lower bound $(\grain{R_1} \times \grain{R_2}) \leg Res$:}\quad
$Res \eqg \grain{R_1} \tun G_2^{rest} \subt
\grain{R_1} \tun \grain{R_2} \subt \grain{R_1} \times \grain{R_2}$
(since $G_2^{rest} \subt \grain{R_2}$, and $\tun \subt \times$ because
the union merges shared fields while the product keeps both copies).
By Lemma~\ref{thm:grain-subset},
$(\grain{R_1} \times \grain{R_2}) \leg Res$.
The product $\times$ coincides with the canonical glb ($\tun$)
when the two grains share no fields in~$Res$, i.e., when $J_k$
contains no grain fields from either input.

\emph{Upper bound $Res \leg (\grain{R_1} \tin \grain{R_2})$:}\quad
Since $\grain{Res}$ contains all of $\grain{R_1}$ (by the
candidate-grain isomorphism), $Res \leg R_1$. Since
$\grain{R_1} \tin \grain{R_2} \subt \grain{R_1}$,
Lemma~\ref{thm:grain-subset} gives
$R_1 \leg (\grain{R_1} \tin \grain{R_2})$.
By transitivity, $Res \leg (\grain{R_1} \tin \grain{R_2})$.
\end{proof}

\begin{proof}[Proof of Theorem~\ref{thm:equijoin-grain} (Grain of an Equi-Join)]
Each candidate is a sub-product of $Res$ isomorphic to it
(Theorem~\ref{thm:equijoin-candidates}), and the grain is the irreducible
one (Definition~\ref{def:grain}). Consider
\[
  G_{cand}(1,2) = \grain{R_1} \tun G_2^{rest}, \qquad
  G_2^{rest} = \grain{R_2} \tdiff J_k;
\]
its fields are those of
$\grain{R_1}$, split as $G_1^{rest}$ and $G_1^{J_k}$, together with
$G_2^{rest}$.

A field of $G_1^{rest}$ or $G_2^{rest}$ is never redundant: it is a
non-$J_k$ field of an irreducible grain, so it is not recovered from
the remaining fields of that grain; lying outside $J_k$, it is not
recovered through the join either; and by the independence hypothesis
it is not determined by the opposite input's grain. The third clause
is not implied by the first two: a determination declared between the
inputs outside $J_k$ would recover such a field directly, and the labeling condition below does not exclude that; see
Remark~\ref{rem:info-independence}. A field $f \in G_1^{J_k}$ can be recovered only via the
opposite grain $\grain{R_2}$ through the join, which requires the
candidate minus~$f$ to still contain all of
$\grain{R_2} = G_2^{rest} \tun G_2^{J_k}$. Since the candidate's
$J_k$-fields are exactly $G_1^{J_k}$, this holds for some such~$f$
iff $G_2^{J_k} \spsub G_1^{J_k}$. Hence
\[
G_{cand}(1,2)\ \text{is reducible} \iff
G_2^{J_k} \spsub G_1^{J_k},
\]
and symmetrically $G_{cand}(2,1)$ is reducible iff
$G_1^{J_k} \spsub G_2^{J_k}$.

\paragraph{Canonical labeling ($G_1^{J_k} \ssub G_2^{J_k}$).}
Then $G_2^{J_k} \spsub G_1^{J_k}$ fails, so $G_{cand}(1,2)$ is
irreducible. It is a subtype of $Res$ and grain-equal to it
(Theorem~\ref{thm:equijoin-candidates}), so it satisfies
Definition~\ref{def:grain}:
$\grain{Res} = G_{cand}(1,2) = \grain{R_1} \tun (\grain{R_2} \tdiff J_k)$.
If the inclusion is strict, $G_{cand}(2,1)$ is reducible, since it
properly contains the grain $G_{cand}(1,2)$.

\paragraph{Incomparable case.}
Neither $G_1^{J_k} \spsub G_2^{J_k}$ nor
$G_2^{J_k} \spsub G_1^{J_k}$ holds, so both candidates are
irreducible; each is a grain of $Res$, and the two are isomorphic
(Theorem~\ref{thm:multiple-grains}).
\end{proof}

\begin{proof}[Proof of Proposition~\ref{prop:join-special-cases} (Join Special Cases)]
Each case is a specialization of Theorem~\ref{thm:equijoin-grain}.

\emph{(1) Equal grains ($R_1 \eqg R_2$).}\quad
If the join key covers both grains ($\grain{R_1} \ssub J_k$ and
$\grain{R_2} \ssub J_k$), then $\grain{R_1} \tdiff J_k =
\grain{R_2} \tdiff J_k = \emptyset$, so by
Theorem~\ref{thm:equijoin-grain}
$\grain{Res} = \grain{R_1} \tun (\grain{R_2} \tdiff J_k)
= \grain{R_1} \tun \emptyset = \grain{R_1}$; by symmetry
$Res \eqg R_2$.

\emph{(2) Ordered grains ($R_1 \leg R_2$, finer $R_1$).}\quad
$\grain{R_2} \ssub J_k$ gives $G_2^{J_k} = \grain{R_2}$ and
$\grain{R_2} \tdiff J_k = \emptyset$. The proviso excludes exactly the
strict reverse labeling $G_2^{J_k} \spsub G_1^{J_k}$, so either
$G_1^{J_k} \ssub G_2^{J_k}$, the canonical labeling, and
Theorem~\ref{thm:equijoin-grain} gives
$\grain{Res} = \grain{R_1} \tun \emptyset = \grain{R_1}$; or the two
portions are incomparable, both candidates are grains, and
$G_{cand}(1,2) = \grain{R_1}$ again. When $J_k = \grain{R_2}$ the
proviso holds outright, $G_1^{J_k}$ being a sub-product of~$J_k$. The labeling
condition is not implied by $R_1 \leg R_2$ alone: for
$R_1 = A \times x \times y$ with $\grain{R_1} = A \times x$,
$R_2 = A \times x \times z$ with $\grain{R_2} = A$, and
$J_k = A \times x$, both $R_1 \leg R_2$ and $\grain{R_2} \ssub J_k$
hold, yet $\grain{R_2} = A \spsub A \times x = G_1^{J_k}$, and in
the result $A$ determines $x$ and~$z$ through $R_2$, then $y$ through
$R_1$, so $\grain{Res} = A$ and not $A \times x$. That is the
theorem's reverse labeling, $\grain{R_2} \tun (\grain{R_1} \tdiff J_k)$,
as the case states.

\emph{(3) Incomparable grains ($R_1 \incg R_2$).}\quad
The unified formula applies directly, with $R_1$ assigned to the input
with the smaller $J_k$-grain-portion (the labeling scope argument above
applies, since incomparable grains do not guarantee incomparable
$J_k$-portions). When $\grain{R_2} \tdiff J_k \neq \emptyset$, the
result grain $\grain{R_1} \tun (\grain{R_2} \tdiff J_k)$ is strictly
finer than either input, the source of fan traps
(Section~\ref{sec:errors}).

\emph{(4) Natural join ($J_k = R_1 \tin R_2$).}\quad
With the inputs labeled so that $G_2^{J_k} \not\spsub G_1^{J_k}$,
Theorem~\ref{thm:equijoin-grain} returns $G_{cand}(1,2)$ either as the
canonical candidate or, when the portions are incomparable, as one of
two isomorphic grains; substituting $J_k = R_1 \tin R_2$ into it gives
$\grain{Res} = \grain{R_1} \tun (\grain{R_2} \tdiff (R_1 \tin R_2))
= \grain{R_1} \tun (\grain{R_2} \tdiff R_1)$,
capturing all of $\grain{R_1}$ plus grain fields of $R_2$ that lie
outside $R_1$. The labeling is load-bearing: for $R_1 = A \times B$
with $\grain{R_1} = A \times B$ and $R_2 = A \times B \times D$ with
$\grain{R_2} = A \times D$, the natural join is on $A \times B$,
independence holds, and $G_2^{J_k} = A \spsub A \times B = G_1^{J_k}$,
so the canonical labeling is $(2, 1)$ and
$\grain{Res} = \grain{R_2} \tun (\grain{R_1} \tdiff R_2) = A \times D$;
the formula with the labels as written would return $A \times B \times D$,
which is reducible since every result row is an $R_2$ row and
$A \times D$ determines~$B$ there.
\end{proof}

\begin{proof}[Proof of Theorem~\ref{thm:generalized-equijoin} (Generalized Equi-Join)]
When join key types satisfy $J_{k1} \eqg J_{k2}$ via isomorphism
$\phi : J_{k1} \stackrel{\cong}{\longrightarrow} J_{k2}$, the join matches
rows where $\phi(\pi_1(r_1)) = \pi_2(r_2)$.

\emph{Reduction to Theorem~\ref{thm:equijoin-candidates}.}\quad
Let $J_k$ be a canonical representative of the isomorphism class.
Via $\phi$, we can treat both join keys as $J_k$. The type-level
operations $\tun$ and $\tdiff$ in
$Res \eqg \grain{R_1} \tun (\grain{R_2} \tdiff J_k)$
are replaced by their grain-level counterparts (operating modulo the
isomorphism $\phi$). The bootstrapping argument of the proof of
Theorem~\ref{thm:equijoin-candidates} carries over unchanged.

When $J_{k1} = J_{k2} = J_k$, the isomorphism $\phi$ is the identity and
the grain-level operations reduce to type-level operations, and the
formula recovers Theorem~\ref{thm:equijoin-candidates} exactly.
\end{proof}

\paragraph{Grain Inference for Other Relational Operations.}
The following proofs establish the rules in Table~\ref{tab:ra-inference}.

\begin{proof}[Selection $\sigma_\theta$]
Selection filters rows without changing the type: $Res = R$.
Since grain is a type-level property and $Res = R$,
$\grain{Res} = \grain{R}$. At the data level, any subset of a
collection retains unique identification by the same grain.
\end{proof}

\begin{proof}[Projection $\pi_S$]
$Res = S$ where $S \ssub R$. By Lemma~\ref{lem:sat-closure}, $S$
determines~$R$ iff $\mathrm{comps}(R) \subseteq \mathrm{comps}(S)^{+}$.
If so, a minimal subset of $S$ determining~$R$ is irreducible and
isomorphic to~$R$, hence a grain $G'$ of~$R$ with $G' \ssub S$,
isomorphic to $\grain{R}$ by Theorem~\ref{thm:multiple-grains}; $\pi_S$
is injective on the valid rows, so $Res \cong R$ and $\grain{Res} = G'$.
When $\grain{R} \ssub S$ this is $\grain{R}$ itself. If
$\mathrm{comps}(R) \not\subseteq \mathrm{comps}(S)^{+}$, the rule does
not apply: $S$ contains no grain of~$R$, two valid rows may agree on~$S$
and differ elsewhere, the projection need not be injective, and the
vertex is flagged (Proposition~\ref{prop:projection-dup}); no grain is
assigned to~$Res$.
\end{proof}

\begin{proof}[Extension by $D = f(R)$]
$Res = R \times D$. Since $D = f(R)$ is determined by~$R$, and
$\grain{R}$ determines all of~$R$ (Theorem~\ref{thm:grain-uniqueness}),
$\grain{R}$ also determines~$D$. For any two elements $(r_1, d_1)$ and
$(r_2, d_2)$: if $\textit{grain}\;r_1 = \textit{grain}\;r_2$, then
$r_1 = r_2$, so $d_1 = f(r_1) = f(r_2) = d_2$.
Therefore $\grain{Res} = \grain{R}$.
\end{proof}

\begin{proof}[Rename $\rho_{a \to b}$]
Renaming changes attribute names but not structure or values, establishing
a bijection between elements. If $\grain{R} = A_1 \times \cdots \times A_n$,
then $\grain{Res} = A_1' \times \cdots \times A_n'$ where
$A_i' = b$ if $A_i = a$, else $A_i' = A_i$. The bijection preserves
uniqueness, so $Res \eqg R$.
\end{proof}

\begin{proof}[Grouping $\gamma_{G_c, \textit{agg}}$]
$Res = G_c \times \textit{AggResult}$. Grouping produces exactly one row per
unique $G_c$-value combination, so $G_c$ uniquely identifies result rows.
The result grain is $\grain{G_c}$: if some grouping columns are
functionally dependent on others (e.g., \textit{DepartmentId} $\to$
\textit{DepartmentName}), a proper subset suffices. When all grouping
columns are mutually independent, $\grain{G_c} = G_c$.
\end{proof}

\begin{proof}[Set operations $\cup$, $\cap$, $-$]
These require union-compatible types ($R_1 = R_2 = R$), so $Res = R$.
Since the result type equals the input type,
$\grain{Res} = \grain{R}$.
\end{proof}

\begin{proof}[Cross product and theta join $\bowtie_\theta$]
$Res = R_1 \times R_2$ (full product structure). By
Theorem~\ref{thm:grain-product},
$\grain{Res} = \grain{R_1} \times \grain{R_2}$ \emph{provided the inputs are independent in the result}, that is, no
determination holds across them, which is the hypothesis that theorem carries
(Remark~\ref{rem:compositional-independence}). Unlike equi-joins, no
bootstrapping occurs without equality constraints.

Independence is not defeated by a declared foreign key between $R_1$
and~$R_2$: the product ranges over \emph{all} pairs passing~$\theta$, so a
determination holding between matching elements does not hold of an
arbitrary pair, and $\grain{R_1}$ does not determine the second component
of a result row. It is defeated exactly when $\theta$ is \emph{functional}
in the sense that each $r_1$ passes with at most one~$r_2$; then the
matched $r_2$, hence $\grain{R_2}$, is a function of $r_1$, so
$\grain{R_1} \times \grain{R_2}$ is reducible and is not the grain. Two
such cases arise, and both are handled by decomposition rather than by this
rule:
\begin{enumerate}
\item $\theta$ is an equality on computed keys, $f(r_1) = g(r_2)$; and
\item $\theta$ is a range predicate selecting the row whose validity
  interval contains a probe value, a SQL \textnormal{BETWEEN} against a
  non-overlapping interval table, such as pricing an order by the rate row
  effective on its order date. This case is functional without being an
  equality, so it is not covered by~(1).
\end{enumerate}
In either case, extending each input with the matching key reduces the
operation to an equi-join on that key, whose grain is given by
Theorem~\ref{thm:equijoin-grain} and is coarser than the product of grains
(Section~\ref{subsec:generalized-join}).
\end{proof}

\begin{proof}[Semi-join $\ltimes$ and anti-join $\rhd$]
Both return rows from~$R_1$ only: $Res = R_1$. Semi-join keeps rows with
a match in~$R_2$; anti-join keeps those without. Both are equivalent to
a selection on $C\;R_1$, so by the selection rule,
$\grain{Res} = \grain{R_1}$.
\end{proof}

\subsection{Grain Axiom Completeness (Section~\ref{subsec:armstrong})}

\begin{proof}[Proof of Lemma~\ref{lem:sat-closure} (Grain ordering on product types is attribute closure)]
Fix~$\Delta$. On product types the structure-map class generated by~$\Delta$
(Definition~\ref{def:grain-holds}) is generated by the canonical projections
and pairings together with one surjection per declared determination, and
is closed under composition. Since $\grain{R_i} \cong R_i$ by maps of the
class, a surjection $\grain{R_1} \thra \grain{R_2}$ exists in it exactly
when a surjection $R_1 \thra R_2$ does, so we argue on the types. A grain
has the closure of its type, being a sub-product that determines every
component of it, which gives the last sentence of the statement.

\emph{Only if.} A projection discards components and a pairing assembles
them, so neither adds a component that was not already present; a
determination $(U, V) \in \Delta$ adjoins the components of~$V$ to any
product already carrying those of~$U$. Nothing else enters the class: an
inverse is in it only as a composite of these same generators
(\emph{Structure maps}, Section~\ref{sec:prelim}), so no map reaches a
component that a generator does not reach. A composite therefore lands in
a product only if that product's components lie in the closure of the
source's components under~$\Delta$, which is the standard attribute
closure.

\emph{If.} Every step of the closure is realized by one determination
generator followed by a pairing,
$\langle \mathrm{id}, d \circ \pi_U \rangle$, which lands in its target
because coherence (Definition~\ref{def:coherent-declarations}) makes
every determination $\Delta$ derives among the target's components agree
with~$d$ on the record so built, and is onto it because under~$\Delta$
the target is inhabited by exactly the records satisfying those
determinations (Definition~\ref{def:grain-holds}); the projection
generators supply the final discard. So a surjection $\grain{R_1} \thra \grain{R_2}$ exists
whenever the inclusion holds. How the base types are inhabited plays no
part, since the class is generated by declarations alone; that is why the
count is exact.
\end{proof}

\begin{proof}[Proof of Proposition~\ref{prop:grain-completeness} (Soundness and completeness on product types)]
For product types $R = A_1 \times \cdots \times A_n$, grains are
sub-products: $\grain{R} = A_{i_1} \times \cdots \times A_{i_k}$.
Identifying each product with its component set,
Lemma~\ref{lem:sat-closure} reads $R_1 \leg R_2$ under a declaration
set~$\Delta$ as functional determination of $\mathrm{comps}(R_2)$ by
$\mathrm{comps}(R_1)$ under the functional dependencies~$\Delta$. A
declaration set therefore plays the part of a set of functional
dependencies, and entailment the part of FD implication. Write
$X = \mathrm{comps}(R_1)$ and $Y = \mathrm{comps}(R_2)$; these are the
sets $\Gamma^{\flat}$ is built from.

\emph{(a)$\Leftrightarrow$(b).}\quad $\Gamma^{\flat}$ is a declaration
set in the sense of Definition~\ref{def:grain-holds}: its pairs have
coherent realizations, for instance the maps between projections that any
relation satisfying $\Gamma^{\flat}$ as functional dependencies induces,
so the class quantified over in~(a) contains it. Every member of $\Gamma$
holds under $\Gamma^{\flat}$, its determination being in $\Gamma^{\flat}$
outright, so
(a) gives $R_1 \leg R_2$ under $\Gamma^{\flat}$, and
Lemma~\ref{lem:sat-closure} gives $Y \subseteq X^{+}$ with the closure
taken under $\Gamma^{\flat}$; that is~(b). Conversely let $\Delta$ be any
declaration set under which every member of $\Gamma$ holds. By the lemma
each determination $(X', Y')$ of $\Gamma^{\flat}$ has
$Y' \subseteq X'^{+}$ under~$\Delta$, so closure under $\Gamma^{\flat}$
is contained in closure under~$\Delta$, and (b) gives
$Y \subseteq X^{+}$ under~$\Delta$; the lemma again gives
$R_1 \leg R_2$ under~$\Delta$. As $\Delta$ was arbitrary, this is~(a).

\emph{(b)$\Leftrightarrow$(c).}\quad Under the identification, A1--A9
correspond to Armstrong's rules:
\begin{itemize}
\item A1 (self-determination) maps to the FD reflexivity axiom
  ($X \to X$).
\item A2 (reflexivity from $\ssub$) maps to the FD reflexivity rule
  ($Y \subseteq X \Rightarrow X \to Y$), the premise $R_1 \ssub R_2$
  being exactly $\mathrm{comps}(R_1) \subseteq \mathrm{comps}(R_2)$.
\item A3 (augmentation) maps to Armstrong's augmentation
  ($X \to Y \Rightarrow XZ \to YZ$).
\item A4 (transitivity) maps to Armstrong's transitivity
  ($X \to Y, Y \to Z \Rightarrow X \to Z$).
\end{itemize}

A5--A9 are derivable from A1--A4 (as shown in the soundness argument
above), just as the corresponding FD rules are derivable from Armstrong's
three axioms. A derivation of $R_1 \leg R_2$ from $\Gamma$ by A1--A9 is
therefore the same object as a derivation of $X \to Y$ from
$\Gamma^{\flat}$ by Armstrong's rules, and Armstrong's closure
theorem~\cite{armstrong1974dependency} states that the latter exists
exactly when $Y \subseteq X^{+}$ under $\Gamma^{\flat}$. That is
(b)$\Leftrightarrow$(c).

Soundness is (c)$\Rightarrow$(a) and completeness is (a)$\Rightarrow$(c),
both through~(b). No countermodel is needed. In Armstrong's setting the
two-row relation separates derivability from validity over instances;
here validity is holding in the class a declaration set generates, and
Lemma~\ref{lem:sat-closure} identifies that with closure directly.
\end{proof}

\subsection{Grain Projection as a Homomorphism (Section~\ref{sec:grain-factorization})}

\begin{proof}[Proof of Theorem~\ref{thm:grain-factorization} (Grain Homomorphism)]
Given any function $h : R_1 \to R_2$ with grain lift
$\gl{h} = \textit{grain}_{R_2} \circ h \circ f_{g_{R_1}}$
(Definition~\ref{def:grain-lift}), we verify the homomorphism
condition $\textit{grain}_{R_2} \circ h = \gl{h} \circ
\textit{grain}_{R_1}$:
\[
  \gl{h} \circ \textit{grain}_{R_1}
  = \textit{grain}_{R_2} \circ h \circ
    \underbrace{f_{g_{R_1}} \circ \textit{grain}_{R_1}}_{\mathrm{id}_{R_1}}
  = \textit{grain}_{R_2} \circ h. \quad\checkmark
\]
\end{proof}

\begin{proof}[Proof of Corollary~\ref{cor:grain-compositionality} (Grain Compositionality)]
Let $h_1 : R_1 \to R_2$ and $h_2 : R_2 \to R_3$ be composable
transformations. By Definition~\ref{def:grain-lift}, their grain lifts
are:
\begin{align*}
  \gl{h_1} &= \textit{grain}_{R_2} \circ h_1 \circ f_{g_{R_1}} \\
  \gl{h_2} &= \textit{grain}_{R_3} \circ h_2 \circ f_{g_{R_2}} \\
  \gl{h_2 \circ h_1} &= \textit{grain}_{R_3} \circ (h_2 \circ h_1)
    \circ f_{g_{R_1}}
\end{align*}
Composing the first two:
\begin{align*}
  \gl{h_2} \circ \gl{h_1}
  &= (\textit{grain}_{R_3} \circ h_2 \circ f_{g_{R_2}})
     \circ (\textit{grain}_{R_2} \circ h_1 \circ f_{g_{R_1}}) \\
  &= \textit{grain}_{R_3} \circ h_2 \circ
     \underbrace{(f_{g_{R_2}} \circ \textit{grain}_{R_2})}_{\mathrm{id}_{R_2}}
     \circ\; h_1 \circ f_{g_{R_1}} \\
  &= \textit{grain}_{R_3} \circ h_2 \circ h_1 \circ f_{g_{R_1}} \\
  &= \gl{h_2 \circ h_1}
\end{align*}
The key step is that $f_{g_{R_2}} \circ \textit{grain}_{R_2} =
\mathrm{id}_{R_2}$: the grain function composed with its inverse is
the identity. Thus intermediate types cancel, and the grain lift
of the composition equals the composition of the grain lifts.
\end{proof}

\subsection{CalcG and Error Detection (Sections~\ref{sec:calcg}--\ref{sec:errors})}

\begin{proof}[Proof of Theorem~\ref{thm:calcg-correctness} (CalcG Correctness and Complexity)]
Let $P = (V, E)$ be a pipeline DAG with vertices in topological order
$v_1, v_2, \ldots, v_{|V|}$.

\emph{Correctness (by induction on topological order).}\quad
\emph{Base case:} Source nodes $v$ have no predecessors; their grain
$\grain{R_v}$ is given by the user's grain annotation. CalcG uses this
annotation directly, which is correct by assumption.

\emph{Inductive step:} Assume CalcG has correctly computed grains for all
predecessors of~$v$. Vertex~$v$ applies operation $\textit{op}_v$ to inputs
with grains $\grain{R_1}, \ldots, \grain{R_k}$ (correctly computed by the
induction hypothesis). CalcG applies the corresponding inference rule from
Table~\ref{tab:ra-inference}. Each rule is sound (proven above for all
relational algebra operations), so the inferred grain $\grain{Res_v}$ is
correct.

\emph{Decidability.}\quad
Every inference rule involves finite type-level set operations ($\tin$,
$\tdiff$, $\tun$) and decidable predicates ($\ssub$, $=$) on finite
field sets. Since each step is decidable
and the DAG is finite, CalcG terminates and is decidable.

\emph{Complexity.}\quad
CalcG performs one pass over the DAG in topological order: $|V|$
vertices. At each vertex, it applies one inference rule involving at most
$k$ input grains, where $k$ is the maximum in-degree (typically
$k \leq 2$ for binary operations). Each rule performs a constant number
of field-set operations ($\tin$, $\tdiff$, $\tun$, $\ssub$), each
costing $O(w)$ where $w$ is the greatest width of a type, the number of
fields it carries. That accounts for the $O(|V| \cdot k \cdot w)$ term.

The field-set operations are not the whole cost. Three places consult the
declared determinations rather than the field sets alone: the projection
rule, whose side condition is the closure test of
Proposition~\ref{prop:projection-dup}; the grouping rule, whose
result is the grain of the grouping type; and the product clause of the
grain operator (Definition~\ref{def:compositional-grain}), which selects a
minimal set of components whose closure under the declarations is the
whole product. Each runs an attribute closure over~$\Delta$. Attribute
closure is linear in $\|\Delta\|$, the total size of the declared
determinations~\cite{beeri1979computational}, and is \emph{not}
$O(w)$: the declarations are an input of their own and their size is
unrelated to the width of a type. Charged once per vertex this
contributes $O(|V| \cdot \|\Delta\|)$, for a total of
$O(|V| \cdot (k \cdot w + \|\Delta\|))$.

It need not be charged that often. A closure depends only on the type it
is taken over and on~$\Delta$, never on the vertex, so each declared
type's closure may be computed once and reused for the rest of the run.
The closure cost over a whole run is then $O(d \cdot \|\Delta\|)$, with
$d$ the number of types carrying a declaration. This is the formal
counterpart of the observation in Section~\ref{sec:calcg} that
determinations are declared on the sources and the target and essentially
nowhere between: $d$ is bounded by the pipeline's declared endpoints and
does not grow with~$|V|$, so the traversal term dominates and the running
time is linear in the pipeline, which is the form the decidability claim
uses.

Space: one grain value (a field set) per vertex, $O(|V| \cdot w)$,
together with $O(d \cdot w)$ for the cached closures.
\end{proof}

\begin{proof}[Proof of Corollary~\ref{cor:data-independent} (Data-Independent Verification)]
By Theorem~\ref{thm:calcg-correctness}, CalcG correctly computes the grain
at every vertex. Checking $\text{CalcG}[P](\grain{R_S}) = \grain{R_T}$
requires only the grain annotations (type-level schema information) and
the pipeline structure, no data access. The verification cost is the one
Theorem~\ref{thm:calcg-correctness} gives, in schema and declaration size
and entirely independent of data size: a data-independent verification.
\end{proof}

\begin{proof}[Proof of Corollary~\ref{cor:grain-implication} (Grain Implication on Product Types is FD Implication)]
We exhibit a linear reduction to functional-dependency implication
and invoke the classical linear-time algorithm of
\cite{beeri1979computational}.

\emph{Reduction $\Gamma \mapsto \Sigma$.}\quad
Let $\mathcal{A} = \bigcup_{R \in \mathcal{U}} \mathrm{comps}(R)$ be the
attribute universe, the union of component sets over the type universe,
finite by assumption. For each $\leg$-assertion $R_1 \leg R_2$ in
$\Gamma$, emit the FD $\mathrm{comps}(R_1) \to \mathrm{comps}(R_2)$
on~$\mathcal{A}$. For each $\eqg$-assertion $R_1 \eqg R_2$, emit
$\mathrm{comps}(R_1) \to \mathrm{comps}(R_2)$ and
$\mathrm{comps}(R_2) \to \mathrm{comps}(R_1)$. For each absolute assertion
$\grain{R} = T$, emit $\mathrm{comps}(T) \to \mathrm{comps}(R)$ and
$\mathrm{comps}(R) \to \mathrm{comps}(T)$, its reading as $T \eqg R$
(Definition~\ref{def:grain-holds}). Each grain assertion contributes
at most two FDs, each of size $O(w)$ where $w$ is the maximum
field-set size; for a fixed field universe,
$|\Sigma| = O(|\Gamma|)$. The target $\sigma$ is translated to an FD
or pair of FDs $\sigma'$ of size $O(|\sigma|)$ by the same rules.

\emph{Soundness and completeness of the reduction.}\quad
$\Sigma$ is the declaration set $\Gamma^{\flat}$ of
Definition~\ref{def:grain-holds}, written as FDs. By
Proposition~\ref{prop:grain-completeness}, on product types
$\Gamma \models \sigma$ holds iff each FD of $\sigma'$ has its right side
in the closure of its left side under $\Gamma^{\flat}$, which is exactly
$\Sigma \models \sigma'$ in FD-implication semantics.

\emph{Linear-time decision.}\quad
$\Sigma \models \sigma'$ is decidable in time $O(|\Sigma| + |\sigma'|)$
via attribute-closure computation~\cite{beeri1979computational}.
Composing with the linear reduction yields
$O(|\Gamma| + |\sigma|)$ for grain implication.
\end{proof}

\begin{proof}[Proof of Proposition~\ref{prop:fan-trap} (Fan Trap Characterization)]
Given the equi-join of $C\;R_1$ and $C\;R_2$ on~$J_k$ producing
$C\;Res$, suppose $Res \stg R_i$ for some~$i$.

By the strict order of Section~\ref{sec:relations}, $Res \stg R_i$ means
$Res \leg R_i$ (there exists a surjection
$f : \grain{Res} \thra \grain{R_i}$) but $R_i \not\leg Res$ (no
surjection $\grain{R_i} \thra \grain{Res}$ exists).

At the data level: the grain ordering $Res \leg R_i$ means each element
of~$R_i$ maps to one or more elements of~$Res$ (via the join). The strict
ordering ($\not\eqg$) means this mapping is not bijective, since some
$R_i$ elements map to \emph{multiple} $Res$ elements. Each such duplication
multiplies the contribution of that $R_i$ row to any aggregate, inflating
metrics computed over~$R_i$.

Detection is a data-independent schema check: compute $\grain{Res}$ via
Theorem~\ref{thm:equijoin-grain} and test
$Res \stg R_i$ using type-level subset comparison
(Corollary~\ref{cor:data-independent}).
\end{proof}

\section{Deferred Remarks and Examples}
\label{sec:deferred-remarks}

This appendix collects remarks and examples deferred from the body for space.

\subsection{Grain Inference (Section~\ref{sec:inference})}

\begin{remark}[Informational Independence]
\label{rem:info-independence}
Two conditions are needed for Theorem~\ref{thm:equijoin-grain}, and
they are independent of one another.

The first concerns \emph{redundant $J_k$-fields}. The two candidates of
Theorem~\ref{thm:equijoin-candidates} both determine $Res$; they differ
only in whether a candidate carries a $J_k$-field the join recovers
from the other side. As the proof of Theorem~\ref{thm:equijoin-grain}
shows, such a field is redundant exactly when the candidate also
contains the opposite input's grain in full, and the canonical labeling
$G_1^{J_k} \ssub G_2^{J_k}$ is the choice that rules this out.

The second is \emph{informational independence} proper: no field of one
candidate component is determined by the other. This is what licenses
the step that a non-$J_k$ field of an input grain is never redundant,
and the labeling does \emph{not} supply it. A determination declared
between the inputs outside $J_k$ leaves the labeling canonical while
making the candidate reducible. Take
$R_1 = \textit{CustomerId} \times \textit{CustomerName}$,
$R_2 = \textit{CustomerName} \times \textit{Email}$ joined on
$J_k = \textit{CustomerName}$, under the declared
$\textit{CustomerId} \thra \textit{CustomerName}, \textit{Email}$. Then
$\grain{R_1} = \textit{CustomerId}$ and
$\grain{R_2} = \textit{CustomerName} \times \textit{Email}$, so
$G_1^{J_k} = \emptyset \ssub \textit{CustomerName} = G_2^{J_k}$ and the
labeling is canonical. Yet the candidate
$\textit{CustomerId} \times \textit{Email}$ properly contains the smaller
isomorphic type $\textit{CustomerId}$, since $\textit{CustomerId}$
determines $\textit{Email}$ across the inputs, outside $J_k$. The
candidate is therefore reducible, and the grain of the result is
$\textit{CustomerId}$.

Independence is thus a hypothesis of Theorem~\ref{thm:equijoin-grain},
not a consequence of the labeling. It is the same condition
Theorem~\ref{thm:grain-product} carries for products and the
$\bowtie_\theta$ rule of Table~\ref{tab:ra-inference} inherits: a
determination declared \emph{across} the operands. Where one holds, the
operation is an equi-join in disguise on the determined key, and its
grain follows by composing the rules.
\end{remark}
\section{Practical Impact and Cost Analysis (Non-Normative)}
\label{sec:app-practical-impact}

\emph{This appendix preserves practical, engineering-oriented material
from an earlier version of this work. It is \textbf{non-normative}: none
of the formal development in the main text depends on it, and it makes no
claim to the rigor of the theorems above. It is retained here, in the
unbounded space of the preprint, for readers interested in the economic
motivation and the implementation spectrum of grain-aware typing.}

\subsection{Encoding Grain in the Type System}
\label{sec:app-encoding}

Grain verification involves two complementary aspects: \emph{grain
computation} (determining what grain a transformation produces) and
\emph{grain enforcement} (ensuring data is used correctly given its
grain). Grain computation applies the inference rules of
Section~\ref{sec:inference} (Table~\ref{tab:ra-inference}) using set
operations on field names; it can be implemented in any language and
answers ``what grain does this transformation produce?'' Grain
enforcement, in contrast, encodes grain in the type system itself, so
that the compiler catches functions receiving data of the wrong grain;
it answers ``is this data being used correctly given its grain?''
Together they enable zero-cost verification: computation happens at
build time without data processing, while enforcement prevents grain
mismatches before execution.

The power of grain-aware type checking depends on the expressiveness of
the underlying type system. We identify a \emph{progressive power
spectrum} across language categories:

\begin{enumerate}
    \item \textbf{No type support} (SQL): grain checking requires
    external static analyzers operating on comment annotations or
    metadata~\cite{dintyala2020sqlcheck}; no compile-time guarantees
    from the language itself.

    \item \textbf{Declarative metadata} (dbt): grain definitions in YAML
    configurations drive schema tests that run after each model is
    built~\cite{dbt2024};
    custom tests enforce grain constraints before query execution.

    \item \textbf{Gradual typing with generics} (Python/PySpark):
    \texttt{GrainedDataFrame[G]} and similar parameterized types enable
    optional static checking via tools such as mypy~\cite{mypy2024},
    with runtime validation as a safety net. Gradual
    typing~\cite{siek2006gradual} allows incremental adoption of type
    annotations.

    \item \textbf{Strong static typing} (Haskell, Scala): phantom type
    parameters, type families~\cite{chakravarty2005associated}, and
    GADTs~\cite{peytonjones2003gadts} enable compile-time grain
    computation with zero runtime overhead; Scala's type
    system~\cite{odersky2014scala} provides similar capabilities through
    higher-kinded types.

    \item \textbf{Dependent types} (Lean~4, Agda): grain relations such
    as \texttt{IsGrainOf G R} carry mathematical proofs of isomorphism,
    enabling formal verification with machine-checkable
    guarantees~\cite{demoura2021lean4,agda2024}.
\end{enumerate}

The key insight is that \emph{any} type system supporting parameterized
types can encode grain at the type level. More expressive systems add
formal proofs, but even simple generics (Python's \texttt{Generic[G]})
catch grain mismatches before runtime. Production systems can thus use
practical encodings (Python, Scala) while formal verification uses
rigorous systems (Lean~4, Agda); concrete implementations across all
three approaches accompany the artifacts of
Section~\ref{sec:verification}.

\subsection{The Economics of Type-Level Verification}
\label{sec:app-economics}

Because grain inference (Section~\ref{sec:inference}) and the
\textsc{CalcG} algorithm (Section~\ref{sec:calcg}) verify a pipeline
without scanning any data, verification cost is dominated by human
review rather than compute. Consider an enterprise platform with 500
pipelines averaging 20--50 transformation nodes, source tables of
$10^8$--$10^{10}$ rows, cloud compute of \$5--\$50 per full pipeline
run, 3--7 debugging iterations per grain bug, and engineer time at
\$100--\$200/hour.

\paragraph{Traditional debugging cost per grain bug.}
\$15--\$350 in compute (3--7 iterations $\times$ \$5--\$50) plus 4--12
hours of engineering time (\$400--\$2{,}400), totaling
\textbf{\$415--\$2{,}750 per bug} when caught during testing. If the bug
reaches production, add incident response (\$5{,}000--\$50{,}000) plus
business impact from incorrect decisions.

\paragraph{Type-level verification cost.}
Zero compute (analysis runs at design time on schema metadata) plus
5--15 minutes of human review (\$8--\$50) to confirm grain annotations
and any machine-checked proof, totaling \textbf{\$8--\$50 per pipeline}.
Across 500 pipelines this is \$4{,}000--\$25{,}000 of verification
effort versus \$207{,}500--\$1{,}375{,}000 in potential debugging cost
at one grain bug per pipeline, a \textbf{98--99\% reduction}. The
advantage is robust to the assumed bug rate: even at 0.25 bugs per
pipeline the saving is 92--97\%, rising to 99\%+ at two bugs per
pipeline.

\paragraph{The AI-era angle.}
Historically, formal verification required expertise with proof
assistants, limiting adoption to academia. Large language models can now
\emph{generate} grain-correctness proofs from inference formulas and
pipeline specifications, shifting the human role from \emph{writing}
proofs (high barrier) to \emph{verifying} them (low barrier), confirming that a generated proof
type-checks and matches the intended semantics.
Combined with the design-time guarantees of the main text, this makes
machine-checkable correctness for AI-generated pipelines practical
rather than aspirational.

\section{Case Study: A Real Pipeline Bug (Anonymized)}
\label{app:case-study}

\emph{This case study belongs to the extended version and may be read
independently of the technical development. The domain and names are
anonymized; the structure is that of a production ingestion pipeline that
had never been formally designed. It expands the example of
Section~\ref{sec:errors}.}

A source system publishes, on each run, the \emph{full current state} of
every order: one extract lists all of an order's lines. There is no
explicit delete signal, a line removed upstream simply stops appearing
in later extracts (\emph{delete-by-omission}). An ingestion pipeline
consumes a batch of accumulated extracts and maintains a \emph{current}
table of orders with their lines and an append-only \emph{history} table.
Reconstructing the denotation $(\grain{}, \ek{}, \BC{})$ of each table
exposed a chain of denotation errors; the behavioral-class \emph{level}
was the subtle one.

\paragraph{The setup.}
The staging collection is flattened to individual lines
$(\textit{OrderId}, \textit{ProductId}, \textit{ExtractId})$, with line
key $\ek{} = \textit{OrderId} \times \textit{ProductId}$ and version
column $\textit{ExtractId}$ (which extract a row came from).

\paragraph{Original design, wrong.}
Staging was declared \textbf{IsMultiVersion} at the \emph{line} grain
$\ek{} \times \textit{ExtractId}$, yet its operations did not match the
declared class. The current state was read as ``pick the latest extract
\emph{per line},'' versioning each line independently; and rows were
written by a plain \textbf{append}, which is the canonical write of
\textbf{IsEvent}, not of \textbf{IsMultiVersion}. Take an order whose
lines are $\{p_1, p_2\}$ in extract $E_1$ and only $\{p_2\}$ in the newer
$E_2$ ($p_1$ removed upstream): the per-line read keeps $p_1$ (its latest
extract is $E_1$) as an active line, a \emph{deleted line silently
resurrected}.

\paragraph{Canonical operations as a guardrail.}
The write mismatch is instructive. The canonical \textbf{IsMultiVersion}
write inserts a new version of an entity \emph{only when its payload
changes}. Had it been used instead of a blind append, it would have
refused to re-insert unchanged lines. The resulting row-count mismatch
against the input would very likely have failed a cardinality check in
testing, surfacing the mis-declaration early. Using the read
\emph{and} the write prescribed by a table's declared class is thus a
guardrail in its own right: it turns a silent denotation error into a
visible test failure.

\paragraph{Naive fix, correcting the class \emph{tag}, still wrong.}
The class was re-tagged from \textbf{IsMultiVersion} to \textbf{IsEvent},
still at the \emph{line} grain $\ek{} \times \textit{ExtractId}$. Because
that grain condition already matched the rows, the same per-line read was
derived and the bug persisted. Fixing the class \emph{tag} without fixing
the grain \emph{level} leaves the bug in place.

\paragraph{Correct design.}
The essential behavioral class is \textbf{IsMultiVersion} at the
\emph{order} grain $\textit{OrderId} \times \textit{ExtractId}$: it is
the whole \emph{order} that is versioned per extract, not each line
independently. Because the rows are flattened to lines, the type is
still \emph{declared} \textbf{IsEvent} with the line-level superkey (to
keep grain a superkey, Section~\ref{sec:foundations}); the fix is not a
further re-tag but \emph{recognizing} that its essential class is
order-level multi-version and using \emph{that} class's operations, the
declared grain and the essential behavioral class must be tracked
separately. The prescribed \textbf{IsMultiVersion} checkout picks the
latest extract \emph{per order} and takes all of its lines; a line absent
from that extract is flagged deleted. On the example it selects $E_2$ for
the order and returns $\{p_2\}$, and $p_1$ is correctly marked deleted.

\paragraph{A second denotation error.}
Independently, the history table had been given grain
$\grain{\textit{current}} \times \textit{LoadDtm} \times
\textit{ExtractId}$. An append-only history of an entity is an
\textbf{IsEvent} of grain $\grain{\textit{current}} \times
\textit{LoadDtm}$ (i.e.\ $\ek{} \times \textit{LoadDtm}$); the extra
$\textit{ExtractId}$ component over-refines the grain and admits
duplicate history rows for a single load.

\paragraph{Lesson.}
Every defect was a mis-stated element of a table's denotation
$(\grain{}, \ek{}, \BC{})$, most sharply the behavioral-class
\emph{level} (order versus line), which no class \emph{tag} alone could
fix. None was visible in the SQL, which compiled and ran, and a small
test would have passed on data without a mid-batch deletion. Two habits
would each have caught it earlier: reconstructing the denotation of every
input and output table, which exposes the wrong level directly; and using
the canonical read \emph{and} write of each table's declared class, whose
cardinality discipline would have failed a test at load time. Correcting
a class \emph{tag} is not the same as correcting a grain \emph{level}.
 
\end{document}